\documentclass[letterpaper]{article}

\usepackage[ruled,vlined,linesnumbered]{algorithm2e}

\usepackage{incl/uai2019}
\usepackage[margin=1in]{geometry}

\usepackage{times}

\newcommand{\sectionc}[1]{\section{\uppercase{#1}}}

\usepackage{pdfpages}

\usepackage{float}

\usepackage{pgfplots}
\usepackage{paralist}
\usepackage{enumitem}

\usepackage[utf8]{inputenc}
\usepackage[english]{babel}
\usepackage{amsmath,amssymb}
\usepackage{amsthm}
\usepackage{tikz}
\usetikzlibrary{arrows,backgrounds,fadings}
\usetikzlibrary{bayesnet}
\usepackage[outline]{contour}  % for "halo" in bayes net
\usepackage{natbib}
\usepackage{hyperref}  % [draft] just intermedita solution because of pagebreak errors
\usepackage{thmtools}  % this is to use namref for theorems <https://tex.stackexchange.com/questions/344495/retrieving-the-title-of-a-theorem-e-g-my-name-from-lemma-1-my-name/347100#347100>

\usepackage{varwidth}  %for multi-line text in tikz nodes

\usepackage{hhline}

\usepackage{tabularx}
\usepackage{tabulary}
\usepackage{subcaption}
\usepackage{nameref}

\usepackage{xspace} % to avoid missing spaces after commands without arguments (https://tex.stackexchange.com/questions/31091/space-after-latex-commands)

\newcommand{\stress}{\textbf}

\newcommand{\defi}{\emph}
\newcommand{\sdefi}[1]{\textbf{\emph{#1}}}
\newcommand{\socalled}{\emph} % similar to defi, but when introducing a well-established term
\newcommand{\termtech}{} % similar to socalled but less stylized
\newcommand{\todo}[1]{\textcolor{red}{#1}}
\newcommand{\remark}[1]{}
\newcommand{\reasonwhy}[1]{}
\newcommand{\notsure}[1]{}
\newcommand{\idea}[1]{}
\newcommand{\maybe}[1]{}

\usepackage{soul,xcolor}
\setstcolor{red}

\newcommand{\michelcmt}[1]{\textcolor{orange}{[\textit{michel}: #1]}}

\newtheorem{Theorem}{Theorem}

\newtheorem{Proposition}{Proposition}
\newtheorem{Claim}{Claim}

\newtheorem{Goal}{Goal}

\newtheorem{Definition}{Definition}

\newtheorem{Corollary}{Corollary}
\newtheorem{Remark}{Remark}

\newtheorem{Asm}{Assumption} %{Assumptions}
\newtheorem{Setting}{Setting}
\newtheorem{Example}{Example}

\DeclareMathOperator*{\N}{\mathbb{N}}
\DeclareMathOperator*{\R}{\mathbb{R}}

\DeclareMathOperator{\dom}{range}
\DeclareMathOperator{\alg}{\sigma-alg} %sigma-algebra

\newcommand\ind{\protect\mathpalette{\protect\independenT}{\perp}}
\def\independenT#1#2{\mathrel{\rlap{$#1#2$}\mkern2mu{#1#2}}}

\newcommand{\E}{\mathbb{E}}

\newcommand{\doc}[1]{\mathit{do}(#1)}

\newcommand{\var}{\mathrm{var}}
\newcommand{\range}{\mathrm{range}}

\newcommand{\vn}[1]{\mathit{#1}} %variable names with more than one letter
\newcommand{\an}[1]{{\it{#1}}} %algorithm names with more than one letter

\newcommand{\meas}{\mathcal{M}_1}

\definecolor{light-gray}{gray}{0.75}

\usetikzlibrary{arrows,shapes,plotmarks}

\usetikzlibrary{arrows.meta}

\tikzset{>=stealth'} 
\tikzstyle{graphnode} = [circle,draw=black,minimum size=22pt,text centered,text width=22pt,inner sep=0pt] 
\tikzstyle{var}   =[graphnode,fill=white]
\tikzstyle{const}   =[graphnode,fill=white,draw=none]
\tikzstyle{hid}   =[circle,fill=none,draw=gray,text=black,inner sep=0]
\tikzstyle{phantom}   =[graphnode,fill=white,draw=none,text=white]
\tikzstyle{halfhid}   =[graphnode,fill=light-gray,draw=black,text=white]
\tikzstyle{obs}   =[fill=none,draw=none]
\tikzstyle{des}   =[rectangle,fill=none,draw=none,text=gray,inner sep=0,outer sep=0]
\tikzstyle{mech}   =[rectangle,fill=none,draw=black]
\tikzstyle{sel}   =[rectangle,fill=none,draw=black,minimum width=0.7cm,minimum height=0.7cm]
\tikzstyle{selection}   =[rectangle,draw=none,fill=black,minimum size=5pt,inner sep=0pt,outer sep=0pt]
\tikzstyle{fac}   =[rectangle,draw=black,fill=black!25,minimum size=5pt]
\tikzstyle{facprior} =[rectangle,draw=black,fill=black,text=white,minimum size=5pt]
\tikzstyle{edge}  =[draw=white,double=black,thick,-]
\tikzstyle{prior} =[rectangle, draw=black, fill=black, minimum size=5pt, inner sep=0pt]
\tikzstyle{dirprior} = [circle, draw=black, fill=black, minimum size=5pt, inner sep=0pt]

\contourlength{1.5pt}
\newcommand{\nt}[1]{{\small\color{nt}#1}}  % node description
\newlength\figheight
\newlength\figwidth

\newlength\figureheight
\newlength\figurewidth

\newif\iffinal % introduce a switch for draft vs. final document
\iffinal

\else

\fi
\usepackage{diagbox}

\usepackage{slashbox,pict2e}
\usepackage{multirow}

\usepackage{pgfplots}
\usetikzlibrary{arrows,shapes,plotmarks,calc}

\newenvironment{cenum}{\begin{enumerate}[topsep=0pt, partopsep=0pt, itemsep=2pt, parsep=2pt, wide=\parindent  % this is to remove indent after first line
	]}{\end{enumerate}}
\newenvironment{citem}{\begin{itemize}[topsep=0pt, partopsep=0pt, itemsep=2pt, parsep=2pt,
	wide=\parindent  % this is to remove indent after first line	
	]}{\end{itemize}}
\newenvironment{iitem}{\begin{inparaitem}}{\end{inparaitem}}

\renewcommand{\paragraph}[1]{{\bf #1}}  % to save some space
\newcommand{\parag}[1]{{\bf #1}}  % to save some space

\newcommand{\spar}{\\[2pt]}  % have a small new paragraph within say Example envs

\newcommand{\lpred}{L^{\text{pred}}}
\newcommand{\lpoint}{L^{\text{point}}}%Pred}}}

\newcommand{\lne}{L^{\text{Nash}}}
\newcommand{\lpointt}{L^{t, \text{point}}}%Pred}}}

\newcommand{\lpredt}{L^{t, \text{pred}}}
\newcommand{\lnet}{L^{t, \text{Nash}}}
\newcommand{\lpointpredt}{\lpointt}
\newcommand{\ltildepointt}{\tilde{L}^{t, \text{point}}}

\newcommand{\gat}{G^{\text{small}}}
\newcommand{\gnonat}{G^{\text{large}}}

\newcommand{\sur}{Y} % surrogate for C, aggregate
\newcommand{\game}{G}
\newcommand{\gamegen}{G}

\newcommand{\nonagg}{small-scale setting\xspace}
\newcommand{\aggr}{large-scale setting\xspace}
\newcommand{\Nonagg}{Small-scale setting\xspace}
\newcommand{\Aggr}{Large-scale setting\xspace}

\newcommand{\mr}{M^{\text{dyn}}}

\newcommand{\bg}{benchmark game\xspace}

\newcommand{\Acov}{V}
\newcommand{\Usig}{W}
\newcommand{\Uact}{B}
\newcommand{\acov}{v}
\newcommand{\usig}{w}
\newcommand{\uact}{b}
\newcommand{\slots}{K}
\newcommand{\Slots}{K}
\newcommand{\slot}{k}
\newcommand{\users}{I}
\newcommand{\usersalg}{\mathcal{I}}
\newcommand{\obsmech}{\bar{Y}}

\newcommand{\utb}{\bar{U}} % most basic utility
\newcommand{\uta}{\tilde{U}} % approx

\newcommand{\msmall}{M^{\text{small}}}
\newcommand{\mlarge}{M^{\text{large}}}

\newcommand{\mpara}{x} %{\phi}
\newcommand{\Mpara}{X} %{\Phi}
\renewcommand{\game}{\gnonat}

\newcommand{\tth}{X}
\newcommand{\bth}{x}

\definecolor{rga}{HTML}{731D0C}
\definecolor{lightblue}{HTML}{EEF2F6}  % 05. Cool Blues
\definecolor{philcol1}{HTML}{BCD7DB}
\newcommand{\rga}[1]{\color{black}{#1}}
\renewcommand{\rga}[1]{}
\definecolor{nt}{rgb}{0, 0, 0}
\definecolor{des}{rgb}{0, 0, 0}
\definecolor{base}{rgb}{0, 0, 0}
\definecolor{ntx}{rgb}{0, 0, 0} %{gray}{0.5}

\title{Coordinating users of shared facilities via data-driven predictive assistants and game theory\thanks{\,\, Extended version, incl.\ supplement, of a publication at \textit{35th Conf.\ on Uncertainty in Artificial Intelligence (UAI), 2019}}}

\author{ 
{\bf Philipp Geiger\textsuperscript{a,b}, Michel Besserve\textsuperscript{a,c}, Justus Winkelmann\textsuperscript{d},  Claudius Proissl\textsuperscript{a}, Bernhard Sch\"olkopf\textsuperscript{a}}   \\
philipp.geiger@tuebingen.mpg.de\\ %, michel.besserve@tuebingen.mpg.de, jwinkelmann@uni-bonn.de,\\claudius.proissl@gmail.com, bs@tuebingen.mpg.de\\	
\textsuperscript{a}Max Planck Institute for Intelligent Systems, \textsuperscript{b}Bosch Center for Artificial Intelligence, \\\textsuperscript{c}Max Planck Institute for Biological Cybernetics, \textsuperscript{d}Bonn Graduate School of Economics
}

\newcommand{\extrefs}[1]{\ref{#1}}

\newcommand{\extref}[1]{\ref{#1}} % in \citep{geiger2019arxiv}}

\begin{document}

\maketitle

%\twocolumn[
%
%
%
%\aistatsauthor{ Author 1 \And Author 2 \And  Author 3 }
%
%\aistatsaddress{ Institution 1 \And  Institution 2 \And Institution 3 } ]

\begin{abstract}
We study data-driven assistants that provide congestion forecasts to users of shared facilities (roads, cafeterias, etc.), to support coordination between them, and increase efficiency of such collective systems. 
Key questions are: (1) when and how much can (accurate) predictions help for coordination, and (2) which assistant algorithms reach optimal predictions?

First we lay conceptual ground for this setting where user preferences are a priori unknown and predictions influence outcomes.
Addressing~(1), we establish conditions under which self-fulfilling prophecies, i.e., ``perfect'' (probabilistic) predictions of what will happen, solve the coordination problem in the game-theoretic sense of selecting a Bayesian Nash equilibrium (BNE). 
Next we prove that such prophecies exist even in large-scale settings where only aggregated statistics about users are available. This entails a new (nonatomic) BNE existence result. 
Addressing~(2), we propose two assistant algorithms that sequentially learn from users' reactions, together with optimality/convergence guarantees. We validate one of them in a large real-world experiment.

\end{abstract}

\sectionc{Introduction}
\label{sec:intro}

%\paragraph{Motivation and scope of the study:} 
%Data-driven interventions are on the rise to improve efficiency of social/economic systems, but the complexity of interactions and uncertainty over agents' preferences raise non-trivial research questions.
%Data-driven mechanisms are on the rise to improve efficiency of social/economic systems, but uncertainty over agents' preferences and how they are affected by interventions raise non-trivial research questions.
%Data-driven interventions on social/economic systems are on the rise, but understanding to what extent they actually improve efficiency in terms of the utility for people remains a challenge.
%Data-driven interventions on social and economic systems are on the rise, but understanding when and how they can increase efficiency in terms of peoples' actual utilities remains a challenge.
%Data-driven interventions on social and economic systems are on the rise, but understanding when and how they can increase efficiency w.r.t.\ peoples' actual utilities remains a challenge.
%Data-driven interventions on social and economic systems are on the rise, but it remains a challenge to understand when and how they can increase efficiency and be of actual utility for people. % remains a challenge.
Data-driven interventions on social/economic systems are on the rise, but it remains a challenge to understand when and how they can improve such systems in terms of peoples' actual utilities and overall resource efficiency. % remains a challenge.
Here we consider central \emph{predictive coordination assistants}, that, % for improving coordination between humans. % that support coordination between humans. % for improving coordination between humans.
%, allowing a more efficient use of resources. 
in the simplest case, work as follows: The assistant provides a congestion forecast $A$ to users of some facility, based on past observations.
The users trust $A$ to be a good forecast, and individually optimize their facility use based on it, e.g., their arrival time slot, to coordinate and avoid crowds. Thereby they generate an observable outcome $Y$, which $A$ is a forecast for. 
In particular, forecast $A$ \emph{influences} outcome $Y$. % there is an \emph{influence} from forecast $A$ to outcome $Y$. %\todo{Figure \ref{fig:predvis}}
Versions of such assistants exist for roads, trains, swimming pools, etc.\ \citep{googlepopulartimes,franceforecasts,dbtravel}% %,dbtravel}
, or, in our experiment, a cafeteria, see Figure \ref{fig:predvis}. 
%\todo{with the potential to use such shared resources more efficiently} \todo{This may help using existing resources more efficiently ...} \todo{multi-agent systems}
%\todo{maybe move the real existing assistants to related work? also add bahn.de?}
% taken from our cafeteria experiment (Section \ref{sec:exp}).
\maybe{Versions of such assistants are, e.g., ``Google Maps Popular Times'' \citep{googlepopulartimes} for facilities such as swimming pools%
%(not explicitly calling it a forecast though)
, or traffic forecasts \citep{franceforecasts}. }
%(Note that we do not have knowledge about the specific methods these services are based on.)
%\todo{add website citations}
%But related problems occur also in traffic and beyond.
%The illustrative problem instance we have in mind throughout this paper is a cafeteria where employees try to coordinate their time of using the cafeteria, trying to avoid congestion, say, in the sense of long queues. %However, our investigation applies beyond this specific example, e.g., to recent technology such as \emph{Google's ``Popular times''}, which displays the current and general congestion for many public facilities. % \todo{or feedback loops see google paper}.

%We belief that, while potential improvements due to such assistants in \emph{individual} facilities may be limited, small improvements multiplied by the \emph{scale of the application domain} do become significant.
%\todo{\stress{Main questins, goals and our contributions}}
%
%\todo{Various questions arise at this point: can such a system be useful? Does prediction accuracy as objective make sense -- feedback loops?}
%
%
%\todo{To understand As a key step, we establish a relation between ... and NE}

%\begin{figure}[t]
%	\centering
%	%\fbox{
%	\includegraphics[width=0.6\columnwidth]{incl/app.png}
%	%}
%	\caption{Mission 1}
%	\label{integration:fig::frame_exp1}
%\end{figure}

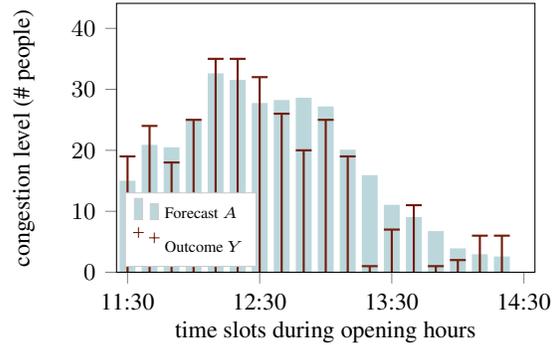
\begin{figure}[t]
\centering
\setlength{\figwidth}{0.9\columnwidth}
\setlength{\figheight}{0.65\columnwidth}
% Modified by phil. This file was created by matplotlib2tikz v0.6.14.
% http://pgfplots.sourceforge.net/gallery.html
\begin{tikzpicture}
\pgfplotsset{every tick label/.append style={font=\small}}
\begin{axis}[
axis on top, % https://tex.stackexchange.com/questions/125823/showing-grid-on-top-of-plot-fill-with-pgfplots
xlabel={time slots during cafeteria's opening hours},
xlabel style={font=\small},
ylabel={congestion level (\# people)},
ylabel style={font=\small},
xmin=35, xmax=73, %36.75, %73.5,
ymin=0, ymax=44.1,
width=\figwidth,
height=\figheight,
xtick={36,48,60,72}, %{0,12,24,36}, %,48,60},
xticklabels={11:30,12:30,13:30,14:30}, %,12:30,13:30},
tick style={font=\small},
tick align=outside,
tick pos=left,
x grid style={lightgray!92.02614379084967!black},
y grid style={lightgray!92.02614379084967!black},
legend style={draw=white!80.0!black,font=\tiny,at={(0.02,0.02)},anchor=south west},
legend entries={{Forecast $A$}, {Outcome $Y$}},
legend cell align={left},
ybar, % otherwise the legendimage doesnt work
bar shift = 0,
bar width=6
]
\addlegendimage{ybar legend, no markers, fill=philcol1, draw=none}
\addlegendimage{only marks, rga, mark=+} %lightgray!66.92810457516339!black}

\addplot [ybar, bar width=6, fill=philcol1, draw=none] %+[ycomb, rga, mark=-, mark size=3, thick] plot %[only marks, mark mid] %only marks, mark=+, rga]
table {%
36 15.014359489998
%37 22.3416180147784
38 20.892448421147
%39 17.1723138048616
40 20.5011196463734
%41 20.3706728662843
42 24.9444404594
%43 26.5025631724055
44 32.6210763466449
%45 32.4436123163952
46 31.5565919887034
%47 30.4477198621725
48 27.7526261615357
%49 25.2754881829375
50 28.2412338517731
%51 27.2752059684237
52 28.6304650711178
%53 29.4735063398951
54 27.1907309029328
%55 23.7169037831813
56 20.1022557322115
%57 17.2655993739509
58 15.9256948309117
%59 12.8975453876438
60 11.0628723721185
%61 10.9851194664073
62 9.04827248991681
%63 8.3191365942165
64 6.75085715910207
%65 4.74615819469087
66 3.91672984287326
%67 3.72724047495002
68 2.95118235627112
%69 2.92519048008868
70 2.58537888064649
};

\addplot +[ycomb, rga, mark=-, mark size=3, thick]%[ybar, bar width=6, fill=lightgray, draw=none] %[ybar] plot % [semithick, lightgray!66.92810457516339!black]
table {%
	36 19
	%37 24
	38 24
	%39 19
	40 18
	%41 19
	42 25
	%43 34
	44 35
	%45 34
	46 35
	%47 32
	48 32
	%49 36
	50 26
	%51 16
	52 20
	%53 31
	54 25
	%55 18
	56 19
	%57 12
	58 1
	%59 16
	60 7
	%61 0
	62 11
	%63 12
	64 1
	%65 7
	66 2
	%67 7
	68 6
	%69 2
	70 6
	%0 14
	%%1 18
	%2 19
	%%3 3
	%4 28
	%%5 31
	%6 42
	%%7 38
	%8 41
	%%9 36
	%10 30
	%%11 34
	%12 37
	%%13 25
	%14 20
	%%15 21
	%16 11
	%%17 22
	%18 21
	%%19 16
	%20 7
	%%21 18
	%22 22
	%%23 13
	%24 4
	%%25 11
	%26 16
	%%27 8
	%28 10
	%%29 5
	%30 5
	%%31 2
	%32 0
	%%33 7
	%34 0
};

%\addplot [semithick, rga]
%table {%
%0 15.1199922290789
%1 22.7937427372917
%2 21.0895230404179
%3 18.6481814248557
%4 19.7202058977425
%5 19.2637625935528
%6 23.1683191842545
%7 25.3052502311021
%8 31.7485171950018
%9 32.0732594228135
%10 31.7186913992373
%11 30.077794717644
%12 26.7896288834562
%13 25.3041768003705
%14 29.0994543483062
%15 27.9286894733129
%16 30.4664555016271
%17 30.2517777052305
%18 27.8354173958752
%19 24.5205219204188
%20 21.4666903086898
%21 17.1891208173078
%22 15.2931326087892
%23 12.8868760328032
%24 11.7983813921478
%25 10.9835698466345
%26 8.32433780455621
%27 8.35237064281608
%28 6.412499938281
%29 4.71972377429806
%30 3.80392092750329
%31 3.90711077164409
%32 3.25851075305164
%33 2.50085049503854
%34 2.85461346209175
%};
\end{axis}

\node at ({$(current bounding box.south west)!0.5!(current bounding box.south east)$}|-{$(current bounding box.south west)!0.98!(current bounding box.north west)$})[
  scale=0.6,
  anchor=north,
  text=black,
  rotate=0.0
]{ };
\end{tikzpicture} %eval_intro_oneday_drop_outcome_for_pres.tex}
%	\vspace{-0.5cm}
\caption{Example of $A \in \R^{18}$, the assistant's (point) forecast, and outcome $Y \in \R^{18}$, in our cafeteria experiment ($A$ updated between but not within days, Section~\ref{sec:exp}).}
\label{fig:predvis}
\end{figure}

%\paragraph{Main goals and contributions:} %\todo{update this based on rearrangement} %, ordered by decreasing significance:
%%In this paper, we take the following approach: we fix \emph{prediction accuracy} $L=\| A - Y \|_2^2$, or versions of it, as the loss function for the assistant. Then we make the following contributions:
%%Our main contributions consist in first answers to the following questions:
%%Our approach consists of two parts
%%
%%
%%Considering classical prediction accuracy as objective main objective for assistants
%%
%%We split the problem into two parts: First we focus on defining a feasible objective for such assistants and analyzing it w.r.t.\ the overall goal of coordination. We pick
%%
%%, and (2) designing assistant algorithms that pursue this objectives.
%%
%%
%We split the problem into two parts: 
%%\begin{citem}
%%\item[(1)] 
%(1) analyzing to what extent classical prediction accuracy (``$\|A - Y \|$'') is a sensible objective for the assistant in light of the overal goal of coordination (\nameref{quest:obj}), and
%%\item[(2)] 
%(2) designing and evaluating assistant algorithms that pursue this objective (\nameref{quest:algo}).
%%\end{citem}
%
%Our main contributions are:
%
%For the ``natural'' choice of prediction accuracy as the assistant's objective, 
%our first main goal
\parag{Main goals and contributions:}
%Towards
We aim at
(1)~understanding to what extent optimally accurate assistant predictions can help coordination between users (Goal \ref{quest:obj}), and 
(2)~designing sequential assistant algorithms that achieve optimal predictions (Goal \ref{quest:algo}). 
Our contributions:
\begin{citem}
\item 
%We lay conceptual ground and define new concepts We lay conceptual ground and give conditions under which 
Introducing new concepts for this setting, we analyze when the assistant achieving a ``perfect'' (probabilistic) prediction $A$ of $Y$, i.e., a ``self-fulfilling prophecy'', is equivalent to ``solving'' coordination in the sense of selecting a Bayesian Nash equilibrium (BNE) (Theorem~\ref{prop:predne}).
%\nameref{prop:predne} (Theorem \ref{prop:predne}): we establish a general equivalence between self-fulfilling prophecies ()general conditions under which optimizing prediction accuracy is equivalent to ``solving'' the coordination problem in an idealized game-theoretic sense -- selecting a certain Bayesian Nash equilibrium (BNE). \maybe{This applies to both, small-scale and large-scale settings (\socalled{nonatomic games})}
%This allows to analyze the utilities experienced by users of the assistant.
\item 
%We establish conditions under which such a self-fulfilling prophecy \maybe{and thus an assistant-based ``solution'' }exists even in large-scale settings with only population-level aggregated user data (Theorem~\ref{thm:ex}), based on the Leray-Schauder-Tychonoff fixed point theorem. This entails a new BNE existence result (Corollary~\ref{cor:ex}) for a new type of \termtech{nonatomic} game.
We establish conditions under which such a prophecy exists even in large-scale settings with only population-level aggregated user data (Theorem~\ref{thm:ex}), using the Leray-Schauder-Tychonoff fixed point theorem. This entails a new nonatomic game BNE existence result (Corollary~\ref{cor:ex}). % for a new type of \termtech{nonatomic} game.
%\end{citem}
%
%Second, we provide data-driven coordination assistant algorithms for dynamic settings (Section~\ref{sec:alg}): % that optimize prediction accuracy in dynamic settings. %Here we contribute:
%%Here our main contributions are:
%\begin{citem}
\item 
%\socalled{Expodamp} (Algorithm \ref{alg:expodamp}): 
We propose learning assistant Algorithms~\ref{alg:expodamp} and \ref{alg:partpredsk} (controllers), for large-/small-scale settings, with optimality/convergence guarantees (Propositions~\ref{prop:expod} and \ref{prop:consistresp}).
%, and, for the smale-scale setting, Algorithm~\ref{alg:partpredsk} (and Proposition~\ref{prop:consistresp}) (Section \ref{sec:alg}).
\item We report positive evaluation of Algorithm~\ref{alg:expodamp} in a large-scale real-world cafeteria experiment (Section \ref{sec:exp}).
\end{citem}
\maybe{The remainder is organized as follows: Sections~\ref{sec:problem} and \ref{sec:dyndef} contain basic model definitions for one-stage and dynamic case, respectively, and Sections~\ref{sec:rem} and \ref{sec:conc} discussion of our results and conclusions, respectively.
\michelcmt{organization of remaining parts not super instructive, would remove.}}

\paragraph{Closest related research overview:}
%
%
%As mentioned above, various such assistants do exist in reality.
%However, we are not aware of existing \emph{research} on user coordination in facilities that is based on predictive ``assistants'' informed by past behavioral data, and where the analysis takes into account that users are individual utility-directed agents.
Within \socalled{game theory}, % \citep{osborne1994course}, %in \socalled{learning in games} \citep{roughgarden2016twenty%
%,leslie2006generalised
%} 
dynamics/equilibria of multiple agents are studied that \socalled{learn} about each other by repeatedly interacting, but without central assistant \citep{shoham2008multiagent}.
Besides this, the following game-theoretic work usually assumes that agents reason fully rationally based on their own a priori given beliefs about other agents,
%use own a priori given beliefs and perform \emph{fully rational} inference about other agents, 
instead of using a predictive assistant informed by past behavioral data:
\socalled{Congestion games} \citep{nisan2007algorithmic} formalize coordination in certain shared facilities.
\maybe{Based on game theory, \socalled{mechanism design} \citep{nisan2007algorithmic,kearns2014mechanism} aims at designing (allocation) mechanisms that maximize social welfare (which is defined in terms of agent's a priori unknown preferences), in spite of agents being self-interested, by using incentives.}
%Based on game theory, 
\emph{(Allocation) mechanisms} are designed \citep{nisan2007algorithmic} that maximize social welfare (which is defined in terms of agent's a priori unknown preferences), in spite of agents being self-interested, by using incentives.
%Interestingly, the social welfare depends on the agent's -- a priori unknown -- preferences.
%Our assistant remotely resembles such mechanisms, 
Unlike our assistant, these mechanisms often fully control the outcome.
%And maximizing social welfare would mean a different conception of ``solving coordination'' than our game-theoretic one.
And we consider ``solving coordination'' in game-theoretic (equilibrium selection) rather than in social welfare terms.
Beyond game theory, certain \socalled{smart cities} research \citep{marevcek2015signalling} uses a \emph{control-theoretic} approach for congested facilities, but they fix an objective that does not in general account for users' individual, a priori unknown preferences. % (Section \ref{sec:...}).
For further related work, see Section~\ref{sec:addrelatedwork}, and Section~\extref{sec:supprelated}. % (the extended version of this paper).

\sectionc{Preliminaries and setting} % and objectives}
\label{sec:problem}
\label{sec:scope}

\label{sec:model}
\label{sec:inducedgame}
%\todo{to elaborate the description form the introduction ...}

%\subsection{Scope of study, model and definitions}

\paragraph{Notation:}
\label{sec:pre}
%\paragraph{Notation:}
%\maybe{MAYBE DROP: 
%}
For a vector $\uact$, 
$\uact_i$ or $[\uact]_i$ is the $i$-th component, 
$\uact_{-i}$ means dropping $b_i$, and %the $i$-th component, and
$(\uact_i, \uact_{-i})$ reads $\uact$.
%\michelcmt{"For a vector $\uact$, $\uact_i$ or $[\uact]_i$ is the $i$-th component, $\uact_{-i}$ means dropping $\uact_i$, and $(\uact_i, \uact_{-i})$ reads $\uact$."}
For a variable $Z$, $\dom_Z$ denotes the (implicitly given) range.
%
%\todo{MAYBEDROP: Given a finite set $S$, $\Delta(S)$ denotes the set of all probability measures on $S$.}
%Keep in mind that, due to space constraints, in what follows we will interleave succinct definitions of basic concepts with problem and model description.
%
\maybe{causal model commented below}
%\defi{Causal DAG} means that a variable $V$ in the DAG is directly influenced only by its \defi{parents} $\PA_V$ in the DAG, and \emph{causal model} means additionally specifying the specific influence mechanisms via conditional distributions. A causal model can predict the outcomes of interventions on variables and mechanisms by performing the very intervention in the model itself: e.g., $M_v$ (or $M_{\dc{V=v}}$), $P^M_v$, $Z_v$ mean model $M$, distribution $P^M$ (under $M$) and random variable $Z$ (in $M$), respectively, after dropping $V$'s generating mechanism from the model and replacing $V$ by constant $v$ in the remaining mechanisms. Similarly, e.g., by $M_\pi$ we denote fixing the assistant's policy to $\pi$. For further details on causal models see \citep{Pearl2000, Spirtes2000, peters2017elements}.

\subsection{General setting and assistant-based system} %of the users' decision problem} % (inference-assistable):}

%\todo{to makre it more readable, maybe introduce X and V after each user i. ("Additional to that, we assume that each user get's the public $V$ and that there is a latent state $X$.)}
%In this and the next paragraph, we introduce the general setting for this paper as well as three additional key assumptions.
Let us first introduce %our general %multiagent %\citep{shoham2008multiagent} 
the general users' decision problem. %, and two additional key assumptions.
We leave it fairly abstract so that later on we can consider different forms of decision making scenarios based on it.

\newcommand{\genset}{general setting}
\newcommand{\Genset}{General setting}

\begin{Setting}[\sdefi{General (one-stage) setting}] %[\sdefi{\Genset}]
\label{set:gen}
There is a finite set $\Slots = \{0, \ldots, |K|-1\}$ of \defi{slots}\footnote{$K$ can be e.g., several facilities, or time slots in one facility.}, \maybe{``Slot'' can mean one time slot, or one of several facilities.}
and a set $\users$, interpreted as \defi{users} (here and in Section \ref{sec:atom}) or \defi{types of users} (in Section \ref{sec:nonat}), respectively. %\todo{DROP: , equipped with $\sigma$-algebra $\usersalg$}. %\michelcmt{is sigma algebra necessary here?} 
%In the \nonagg in Section \ref{sec:atom}, and for simplicity also in the current section, we will interpret $N$ as set of \emph{users} though.
%
%
%\stress{User variables:}
Each user $i \in \users$:
\begin{citem}
	\item receives a (private) \emph{signal} $\Usig_i$, %consisting of a private $\Usig_i$ plus the public $\Acov$,
	\item as (private) \defi{action $\Uact_i$} chooses a slot in $\Slots$, and
	\item experiences (private) \defi{utility} $U_i$ he wants to maximize. % \maybe{\todo{OLD:} $U_i = \bar{U}_i(X, \Uact)$, which he wants to maximize, with $\bar{U}_i$ his \defi{utility function}.}
\end{citem}
Let  $\Usig = (\Usig_i)_{i \in \users}$,  $\Uact = (\Uact_i)_{i \in \users}$ and $U = (U_i)_{i \in \users}$.
%Besides these ``private'' variables, assume there is 
%Furthermore, assume there is some \defi{(latent) state} $X$ and 
Besides the private signals, there is a \defi{publicly available signal} $\Acov$, and some underlying \defi{(latent) state} $X$.
And there is a \defi{publicly observable outcome} $Y = \obsmech(X, \Uact)$, for some function $\obsmech$.%. $Y$ can coincide with $B$, in , which is either $B$ its %of the users' actions $B$ and the latent state $X$ 
\footnote{This models the fact that actions $B$ may no be observed publicly, but just, say, some stochastic aggregation of them.} %\michelcmt{in the one stage system it is unclear what we mean by "Y is observed by the assistant", I am getting more convinced that introducing multistage directly would be better}
We assume there is a ``true'' distribution $P(X, \Acov, \Usig)$.
If not stated otherwise, we assume that all users $i$ are \sdefi{inference-assistable}, i.e.,
%for all users $i$ we assume that
\begin{align}
U_i %= \uta_i(\Usig_i, \Uact_i, h_i(Y)) 
= \uta_i(\Usig_i, \Uact_i, h_i(\obsmech(X, \Uact))), \label{eqn:uta}
\end{align}
for (continuous) functions $\tilde{U}_i, h_i$ such that $h_i(\obsmech(x, (\uact_i, \uact_{-i})))$ does not depend on $\uact_i$, for all $\uact, x$.%
\footnote{The intuition behind this constraint on the utility functions %(which may be modulated in future work) 
	is that users' decision making can be discerned into (1) an optimization performed by the users and (2) the task of predicting $Y$ which can be ``outsourced'' to an assistant.} %, which 
And let all users $i$ be \sdefi{assistant-separable}, i.e., $h_i(Y) \ind \Usig_{i} | \Acov$ (for any possible mechanism that generates $B$ from $V, W$).\footnote{This means, roughly, that the users do not know more about each other than is contained in the public $\Acov$.}
%$V$that $i$'s private $\Usig_{i}$ does not add information about $h_i(Y)$ beyond what is contained in the public $\Acov$.
%
%\parag{Inference-assistable:} 
%\todo{TODO: rename this Inference-assistable in rest of the paper} 
%Let us define certain forms of utility functions and behavior that makes users amenable to coordination assistance.
%\todo{Give intuitions for both definitions}
%
%\parag{Inference-assistable:} 
%Regarding utilities, we have to make a fundamental assumption ... \defi{inference-assistable} ... to be able to separate the users decision making problem into a prediction task -- that can be outsourced -- and an optimiation task ...
%\todo{
\end{Setting}

\begin{figure}[t]
\contourlength{1pt}
\centering 
\tikzstyle{stdarrow}=[-{Latex[scale=1.2]}]
\tikzstyle{darrow}=[-{Latex[scale=1.2]}] %{Latex reversed[scale=1.2]}
\tikzstyle{oarrow}=[-] %{Latex reversed[scale=1.2]}
\begin{tikzpicture}[scale=1]

%\color{black}

%\begin{scope}

\node[obs, align=center] at (0, 0) (X) {\nt{state}\\$X$};
%\node[des, below=0.2cm of X, anchor=center] () {\ntx{state}};

\node[obs, align=center] at (-2, 0) (W) {\nt{pub.\ signal}\\$\Acov$};

\node[obs, align=center] at (2, 0) (Theta) {\nt{priv.\ signals}\\$(\Usig_i)_{i \in \users}$}; % {=} (\Usig_i)_{i \in N}$};

\node[obs, align=center] at (-2, -2.6) (A) {\nt{forecast}\\$A$};

\node[obs, align=center] at (2, -2.6) (C) {\nt{priv. actions}\\$(\Uact_i)_{i \in \users}$};

\node[obs, align=center] at (0, -3.9) (Y) {\nt{pub.\ outcome}\\$Y$}; %-3.9

\node[sel,align=center,minimum height=1cm,minimum width=1.6cm,fill=lightblue]  at (-2, -1.3) (pi) {\nt{policy}\\$\pi$};
\node[sel,align=center,minimum height=1cm,minimum width=1.6cm,fill=lightblue]  at (2, -1.3) (sigma) {\nt{``best resp-}\\\nt{onse to $A$''}}; %``opt.\ act.\\under $Y$''};  %\nt{policies}\\$(\sigma_i)_{i \in N}$};

\draw[stdarrow] (X) to (W);
\draw[stdarrow] (X) to (Theta);
\draw[stdarrow] (W) to (pi);
\draw[stdarrow] (Theta) to (sigma);
\draw[stdarrow] (pi) to (A);
\draw[stdarrow] (A) to (1, -2.6) to (1, -1.3) to (sigma);
\draw[stdarrow] (sigma) to (C);
\draw[stdarrow] (C) to (2, -3.9) to (Y); %-3.9
\draw[stdarrow] (X) to (Y);

\draw[darrow,color=gray,dashed] (Y) to (-3, -3.9) to (-3, -1.3) to (pi); %-3.9
\draw[darrow,color=gray,dashed] (-3.6, -1.3) to (pi);
%\draw[->,color=gray,dashed] (-.3, 1.3) to (0, 1.3) to (X);
\draw[darrow,color=gray,dashed] (0, 1) to (X);
\draw[oarrow,color=gray,dashed] (A) to (-3, -2.6);
\draw[oarrow,color=gray,dashed] (W) to (-3, 0) to (-3, -1.3);

%\draw [decorate,decoration={brace},xshift=0pt,yshift=0pt,color=des]
%(-3,1.2) -- (-1,1.2) node [midway,yshift=0.3cm] 
%{\nt{assistant}};
%\draw [decorate,decoration={brace},xshift=0pt,yshift=0pt,color=des]
%(1,1.2) -- (3,1.2) node [midway,yshift=0.3cm] 
%{\nt{users}};

\draw [xshift=0pt,yshift=0pt,draw=none] %color=des]
(-3.3,0.8) -- (-0.7,0.8) node [midway,yshift=0.3cm] 
{\nt{\bf{\color{rga}assistant}}}; % 0.8
\draw [xshift=0pt,yshift=0pt,draw=none] %color=des]
(3.3,0.8) -- (0.7,0.8) node [midway,yshift=0.3cm] 
{\nt{\color{rga}\bf{users}}};
%\draw [decorate,decoration={brace},xshift=0pt,yshift=0pt,fill=none,draw=none] %color=des]
%(1,1.2) -- (3,1.2) node [midway,yshift=0.3cm] 
%{\nt{users}};

\begin{scope}[on background layer]

\draw[fill=lightblue,draw=none,path fading=south] (-3.3, 1.5)  % 1.5
rectangle (-0.7, -4.4);
\draw[fill=lightblue,draw=none,path fading=south] (3.3, 1.5)
rectangle (0.7, -4.4);

\end{scope}

%\vspace{-.5cm}

%\node[sel,align=center,minimum width=3cm,fill=lightblue,draw=none]  at (-2, 1.3) () {\nt{assistant:}};
%\node[sel,align=center,minimum width=3cm,fill=lightblue,draw=none]  at (2, 1.3) () {\nt{users (or types):}};

\end{tikzpicture}
\caption{(Causal) diagram of the \emph{assistant-based system~$M$} (without $U$). The dashed gray arrows indicate the dynamic extension $\mr$ we will introduce in Section \ref{sec:rep}.} % Variables $U, L$ are not depicted.}
\label{fig:ts}
\label{fig:all_DAGs}
\label{fig:dag}
\end{figure}
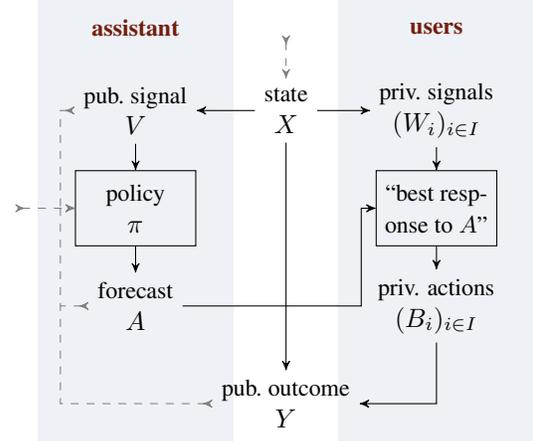

Setting \ref{set:gen} leaves open how users reason/decide. % on their actions.
%Our main object of study is a system that enriches this setting by a central assistant, where user $i$ picks the action that maximizes her expected utility, given $Y$ is distributed according to the assistant's forecast.
Our main object of study is a system that enriches this setting: user $i$ chooses $B_i$ that maximizes her expected utility, given $Y$ is distributed according to a central assistant's forecast: %the forecast of a central assistant: %'s forecast.
% the expected optimal action  whose forecast informs the users, and they pick the optimal action under the forecast:

\begin{Definition}[\sdefi{Assistant-based %(one-stage) 
	system $M$}]
\label{def:m}
Based on Setting \ref{set:gen}, or any restricted version, let the \defi{assistant-based (one-stage) 
	system $M$} be defined by the following objects and assumptions additional to Setting \ref{set:gen}, as depicted by the (causal) Bayes net \citep{Pearl2000} in Figure \ref{fig:dag}:
There is an assistant that takes public signal $\Acov$ as input and outputs $A$, a probabilistic \emph{forecast} for the public outcome $Y$,
%\end{align}
based on \defi{policy} $\pi$, i.e.,  $A= \pi(\Acov)$. That is, $A$ is a distribution over $Y$ (later we also consider point forecasts).
%
%\todo{CHANGE TO THIS TO DROP $\sigma$?:} 
User $i \in \users$ takes forecast $A$ (besides her private signal $W_i$) as input, and acts \sdefi{assistant-best-respondingly}, i.e., %choses the ``optimal action under $A$'', i.e., %``best-responds'' to forecast $A$, i.e., choses action 
\begin{align}
\label{eqn:mbr}
%&C_i \in \arg\max_{c_i} \E_{Y' \sim \pi(w)}   \left( \uta_i(\Usig_i, c_i, h_i(Y'))\right),  \\
%&\todo{or} C_i \in \arg\max_{c_i} \E_{Y' \sim A}   \left( \uta_i(\Usig_i, c_i, h_i(Y'))\right),  \\
B_i \in \arg\max_{\uact_i'} \E_{Y' \sim A}   ( \uta_i(\Usig_i, \uact_i', h_i(Y')) 
%&C_i \in \arg\max_{c_i} \E_{C_{-i} \sim P_M(C_{-i}|Y=a)}\left(U_i(\usig_i, (c_i, C_{-i}), \Phi)\right), \label{eqn:uniquebrass}
\end{align}
(breaking ties via $K$).
A joint $P_M(X, \Acov, \Usig, A, \Uact, Y, U)$ is induced by all the above (measurable) equations and $P(X, \Acov, \Usig)$.
%We may write $P_{M, \pi}(\ldots)$ and $\E_{M, \pi}(\ldots)$ to make the dependence on $\pi$ explicit.
We may write $P_{M, \pi}$ and $\E_{M, \pi}$ to make the dependence on $\pi$ explicit.
$Y$ is observed, but the specific $P(X, \Acov, \Usig)$ and $\tilde{U}_i$'s are a priori unknown, and $X, \Usig, B$ and utilities $U$ % (and $\tilde{U}_i$'s) 
are unobserved by the assistant. % (but $Y$ is observed).
%$Y$ is observed, but
%$P(X, \Acov, \Usig)$, $X, \Usig$, actions $B$ and utilities $U$ (and $\tilde{U}_i$'s) 
%are \emph{hidden} from the assistant.
%The latter distribution and the values of $X, \Usig, \Uact, U$ are hidden from the assistant.
%
%
%$P(X, \Acov, \Usig)$ together with all the above (measurable) equations induces a joint $P_M(X, \Acov, \Usig, A, \Uact, Y, U)$. %If we want to make the dependence of $P_M$ on, say, $\pi$ explicit, 
%This distribution and the values of $X, \Usig, \Uact, U$ are hidden from the assistant. %, except that users $i$ observes $\Usig_i, \Uact_i, U_i$.
%%The assistant knows nothing about this distribution (in particular not the users utilities) a priori, nor does it observe $X, .
%%

\end{Definition}

%\subsection{Model of the assistant-based system:}

\subsection{Game-theoretic tools to characterize efficiency} % as a tool to characterize the user preference-respecting efficiency of assistant-based systems}
\label{sec:pregame}

We want to analyze the degree of efficiency that assistant-based coordination can achieve.
For this, we now define what a \emph{``solution''} of the coordination problem would be (a BNE), accounting for users' preferences.
This is based on an \emph{idealized}, assistant-free version of Setting \ref{set:gen} (a Bayesian game), where users $i$ have \emph{informed priors and unlimited inference abilities themselves}, using $W_i$ and $V$ as input.
Then, for any user behavior that arises in the assistant-based system, we can check if it is %efficient in the sense of being 
(or rather: corresponds to) such a solution.
For background on game theory and the Bayesian game definition we use, see Section \extref{sec:gameback}.
%
%we introduce a scenario (a game) where the users have knowlegde and reasoning abilities, and what the established solution of this scenario is.
%Then
%
%In order to assess the efficiency that an assistant-based system can achieve, 
%
%Let us now introduce concepts from game theory in order to be able to analyze the degree of efficiency, in terms of users' utilities, that an assistant-based system can achieve.
%
%
%%\parag{Assistant-free game (as comparison):}
%Let us now use game theory as one way to formalize what an ideal ``solution'' of the users' decision problem would be, giving us one way to benchmark assistants.

\begin{Definition}[\sdefi{Benchmark (assistant-free) game $G$}]
\label{def:g}
Based on Setting~\ref{set:gen}, or any restriction of it, let the \defi{\bg $G$} be defined as the Bayesian game canonically associated to this setting:
%\footnote{This game is very general: it allows $\users$ to be a continuum, similar as in \citep{kim1997existence}.} 
%\michel{The set of players (or types of players) is $\users$, and player $i \in \users$ has:}{
Each user $i \in \users$ is a player who has:
\begin{iitem}
	%\item private signal $\Psi_i := (\Usig_i, W) \in \dom_{\Usig_i} \times \dom_{W}$,
	\item[] signal $(\Usig_i, \Acov) \in \dom_{(\Usig_i, \Acov)}$,
	%$\Psi = (\Usig_i, W) \in \dom_\Psi = \dom_{\Usig_i} \times \dom_W$
	%\maybe{\footnote{Instead of just specifying the domains of the components of a game, we also directly introduce the variable names we will typically associate with it (slightly overloading them since e.g. $\Uact$ can have different underlying generating mechanisms in $M$ and $G$). Once the game and a strategy profile is fixed, these variables can rigorously be interpreted as random variables.}}
	\item[] (measurable) utility function $\dom_{(X, W_i, \Uact)}  \to \R$ given by Eq.~\ref{eqn:uta}, %$, (x, c) \mapsto f_{U_i}(X, C)$, measurable w.r.t.\ the product $\sigma$-algebra
	%defined by $(\usig_i, c) \mapsto U_i(q_i, c)$ \todo{or rather: $\dom_X \times \mathcal{C}  \to \R$}
	\item[] and action $\Uact_i \in \Slots$.
	The utility functions are common knowledge and $P(X, \Acov, \Usig)$ is the common prior. %(``objective'') prior.
\end{iitem}
%
%	
%Ignoring the assistant-based system for a moment and starting from an instance of the above general setting of the users' decision problem again, imagine an \emph{idealized} scenario, where we additionally assume the utility functions $\tilde{U}_i$ (incl. $h_i$), $i \in \users$ to be \defi{common knowledge} among the users, and $P(X, \Acov, \Usig)$ as their \defi{common (``objective'') prior}. This induces a Bayesian game \citep{osborne1994course} -- with the users as players -- which we refer to as \defi{\bg $\game$}.
%As an \emph{idealized} game-theoretic vmodel of the above decision problem and how users solve it, let us first take a game-theoretic view: we 
%First, additional to the above basic decision problem formulation, assume that the utility functions $\Usig_i, i \in N$ are \defi{common knowledge} among the users, and $P(X, \Acov, \Usig)$ is their \defi{common (``objective'') prior}. This turns the above description into a Bayesian game -- with the users as players -- which we refer to as \defi{\bg $\game$}.
%As an \emph{idealized}, fully rational scenario where a user's decision making is based on her \emph{own inference} let us consider the \defi{assistant-free (Bayesian) game $\game$} defined by the above decision making problem, with users as players, utility functions $\Usig_i, i \in N$ as \defi{common knowledge}, and $P(X, \Acov, \Usig)$ as \defi{common (``objective'') prior}.
\end{Definition}

As usual, a \sdefi{(pure) strategy profile} for $G$ is a tuple $s=(s_i)_{i \in \users}$ of \defi{(measurable, pure) \textbf{strategies}} $s_i : \range_{(\Usig_i, \Acov)} \to K$\footnote{I.e., the strategy $s_i$ maps player $i$'s signal to her action}, $i \in \users$. %$s_i, i \in \users$, i.e., mappings from signal space $\range_{(V, W_i)}$ to action space $K$. % \todo{$s_i : \range_{(\Usig_i, \Acov)} \to K$} $s_i(\Usig_i, \Acov), i \in \users$. % \acov_i \in \dom_{\Acov}, \usig_i \in \dom_{\Usig_i}, i \in N$.
%$s_i : \dom_{(\Usig_i, \Acov)} \to \Slots, i \in N$. % such that $(i, \usig_i, \acov) \mapsto s_i(\usig_i, \acov)$ is measurable.
%
%Given a strategy profile $s$, let us denote by 
%Let $P_{G, s}$ denote the distribution over all variables induced by ``plugging'' $s$ into $G$. %the induced distribution over the variables of $G$.
%Let us now introduce a classical game-theoretic concept for ``solving'' the idealized decision problem that the game $G$ poses. 
%As usual, the idealized game-theoretic solution
%It will help us characterize the actual utilities that users can gain from using an assistant.
%The most common game-theoretic ``solution'' to this idealized scenario is the following
%
%
%\parag{Def. Nash equilibrium:}
%\todo{say that we restrict to pure? or simply allow mixed as well which simplifies also the expectation issue}
%\todo{Commonly in game theory, the notion of players ``solving'' the game $G$ is formalized by them selecting a strategy profile $s$ that is ....}
A strategy profile $s$ is a \textbf{\defi{Bayesian Nash equilibrium (BNE)}} of $\game$, if %, for (almost) all $i, \usig_i, \acov$, 
\begin{align}
%\E_{G, (s_{i}, s_{-i})}(U_i|\Psi=\psi) \geq \E_{G, (s_{i}', s_{-i})}(U_i|\Psi=\psi) \\
%\todo{or } s_i(\psi) \in \arg\max_{c'} \E(\bar{U}_i(X, (c', (s_i(\Psi_i))_{i \in N \setminus \{i\} } )|\Psi=\psi) \label{eqn:nd} \\
%
%s_i(\usig_i, w) \in \arg\max_{c_i} \E(\bar{U}_i(X, (c_i, s_{-j}(\Usig, W) ))|\usig_i, w) ,
%
%s_i(\usig_i, \! \acov) {\in} \arg\!\max_{\uact_i} \E(\bar{U}_i(X,\! (\uact_i,\! s_{-j}(\Usig,\! \Acov) ))|\usig_i, \! \acov),\!
s_i(\usig_i, \acov) \in \arg \max_{\uact_i} \E_{G, (\uact_i, s_{-i})}(U_i | \usig_i, \acov),
%s_i(\usig_i, \acov) \in \arg\max_{\uact_i'} \E \left( \uta_i(\Usig_i, \uact_i', h_i(\bar{Y}(\uact_i', s_{-i}( \Usig_{-i}, \Acov)))) \mid \usig_i, \acov \right),  
\end{align}
for (almost) all $i, \usig_i, \acov$;
%\michelcmt{I would use square bracket for conditional expectation}
with $U_i$ as in Eq.~\ref{eqn:uta}, and $\E_{G, (\uact_i, s_{-i})}$ the expectation under $P_{G, (\uact_i, s_{-i})}(X, \ldots, U)$ obtained by ``plugging'' strategy profile $(\uact_i, s_{-i})$ into game $G$ ($\uact_i$ here means the constant strategy).%
\footnote{I.e., each player's strategy is a best response to the others.} %can improve her utility by unilaterally deviating from her strategy $s_i$, when leaving the other players' strategies $s_{-i}$ fixed.}
\maybe{for this to be defined we need that $U_i$ is guaranteed to be a RV}
We call the BNE \defi{strict} if the argmax is unique. % for (almost) all $\usig_i, \acov$.
%
%
%\parag{Correspondence:} % between assistant-based and assistant-free:}
%\todo{To use these game-theoretic concepts to evaluate an assistant policy $\pi$ of the assistant-based system $M$, let us introduce a way to consider outcomes/behavior of the assistant-based system $M$ ..... The idea is that the concatenation of the assistants policy and the users' response can be seen as a strategy profile of the assisatne-free game $G$, and then we can check whether it ``solves'' this game.}
%To use these game-theoretic concepts to evaluate an assistant policy $\pi$ of the assistant-based system $M$, % with user policies $\sigma=(\sigma_i)_{i \in N}$, 
%note that the induced user behavior $P_{M, \pi, \sigma}(B|V, W)$ can be seen as a stratgy profile of $G$ for which we can check whether it is a ``solution'' of $G$.
%%\todo{TODO: ``Solutions'' to the game $G$ are given in the form of strategy profiles, disribution of $B$ given $(W, V_i)$
%%Now to be able to check }
%%To make assistant-based and assistant-free outcomes comparable, 
%Therefore let us formally define: Given $M$ with assistant policy $\pi$ and users' behavior $\sigma$, 
%

To relate $M$ to $G$, given an assistant policy $\pi$ of the assistant-based system $M$,
we define the \textbf{\defi{corresponding strategy profile}} $s_\pi$ by the composition of $\pi$ and users' subsequent (deterministic) ``best-response'' action, i.e.,
%\todo{ $s_\pi(w_i,v) = (\sigma_1(w_1, \pi(v)), \ldots, \sigma_n(w_n, \pi(v)))$}
\begin{align}
[s_\pi]_i(\usig_i, \acov) := \E_{M, \pi}(B_i|\usig_i, \acov),\ \text{ for all $i, \usig_i,  \acov$}. \label{eqn:s}
%[s_\pi]_i(\usig_i, \acov) := \sigma_i(\usig_i, \pi(\acov)), \text{ for all $i, \usig_i,  \acov$,} \label{eqn:s}
%[s_\pi]_i : (\usig_i, \acov) \mapsto \sigma(\usig_i, \pi(\acov))%
%:= s_i(\uact_i|\usig_i, \acov) := P_{M, \pi}(\uact_i|\usig_i, \acov)
%, i \in N , 
%\label{eqn:s}
\end{align} 
%\footnote{\todo{and then simply check whether $s_\pi$ is a BNE of $G$.}}
%
%
%player $i$'s strategy resembles the concatenation of assistant policy and user $i$'s response to it in $M$.
\maybe{(if measurable).} Conversely, given a strategy profile $s$ of the \bg $G$, we define the \textbf{\defi{corresponding assistant policy}}%
\footnote{I.e., the assistant as forecast takes the distribution of outcome $\sur = \obsmech(X, \Uact)$ given $\Acov=\acov$, under $G,s$.}%
\begin{align}
\pi_s(\acov) := P_{G, s}(\sur|\acov),\ \text{ for all $\acov\maybe{ \in \dom_\Acov}$.} \label{eqn:cp}
\end{align}

\subsection{Objective functions}
\label{sec:preobjectives}

%
%\paragraph{Algorithm objective.}
%
%Roughly speaking, we aim at designing an assistant policy $\pi$ that support coordination between the users $i \in N$. Since this task itself is hard to make rigorous, in particular due to the multi-agent nature of the problem, we consider the following two rigorous, quantitative objectives as \emph{``proxy'' objective} as a function of $\pi$: 
%
%
%Since the overall goal -- coordination between users -- is hard to make rigorous, we consider the following two rigorous ``proxy'' objectives for the assistant's policy $\pi$, and analyze their relation to the overall goal in Section \ref{sec:obj}:
%
%
We consider the following two objective functions for the assistant's policy $\pi$, where,
as we will see, the former can be seen as a directly measurable ``proxy'' to the latter: \maybe{ more principled objective which is defined in terms of the hidden preferences:}
%(see also remarks  %as reasonable, quantitative ``proxies'' for the overall coordination goal 
%in Section \ref{sec:social_welfare_remark}):
%
%
%
%
%\todo{based on data} (which also depends on teh users strategy tuple $\sigma$ which we leave unspecified for the moment):
%On the one hand, the objective we consider in this paper for designing the assistant's policy $\pi$ is \emph{accuracy of the (influential) predictions} of the assistant. 
%On the other hand, we consider the objective of selecting a (Bayesian) Nash equilibrium.
%We call assistants that ``aim'' at this objective \defi{predictive (coordination) assisants}. 
%We consider two versions of this objective -- two loss functions -- depending on the form of the prediction:
\begin{citem}
\item \textbf{\defi{(probabilistic) prediction accuracy objective (loss)}}:
%\begin{align}
\begin{align}
\lpred_\pi := \E \left( d\left(P_{M, \pi}\left(Y|\Acov\right), \pi\left(\Acov\right)\right) \right) \text{, for all $\pi$,}
\end{align}
%i.e., \defi{(probabilisitc) predicition accuracy},
%\end{align}
with $d(\cdot,\cdot)$ some arbitrary but fixed statistical distance which is 0 iff both distributions coincide;
%total variation or wasserstein or any statistical distance -- we get more specific later
%\todo{OLD: problwm with KL: only defined when derivative exists}
%$L^{\text{ProbPred}} = L^{\text{prob}}_\pi := \kl(P_\pi(Y|W) \| \pi(W) )$ \todo{or rather $\| \|$} \footnote{which is not directly observable and serves for the more conceptual part of the paper} (Sections \ref{sec:pred-sol} and \ref{sec:consistresp});
%\item prediction accuracy for the case $A \in \mathcal{C}$, i.e., point forecasts: $\lpoint = \lpoint_\pi := \E_\pi(\|A-Y\|_2^2)$ (Section \ref{sec:expodamp});
\item \textbf{\defi{equilibrium selection objective%
		\footnote{Equilibrium selection is a game-theoretic formulation of solving coordination \citep{nisan2007algorithmic}. Clearly, equilibria can still be inefficient in terms of social welfare (see Section \ref{sec:rem}).}}}: $\pi$ is optimal iff the corresponding strategy profile $s_\pi$ (Eq.~\ref{eqn:s}) is a BNE of the \bg $G$.
%
%\todo{maybe drop the function-based formulation} \emph{achieving an outcome comparable to a BNE of the \bg $G$\maybe{where players are fully rational and perform inference themselves}}:
%\begin{align}
%\lne_\pi := \left \{ \begin{array}{ll} 0, & \text{  if $s_\pi$ is a BNE of $G$ } \\ 1, & \text{ else, } \end{array} \right.
%\end{align} 
%with $s_\pi$ defined by
%\begin{align}
%[s_\pi]_i  := s_i(c_i|\usig_i, w) := P_{M, \pi}(c_i|\usig_i, w), i \in N ,\label{eqn:s}
%\end{align} 
%(selection of a (B)NE is a well-established formulation of the coordination problem in game theory \citep{nisan2007algorithmic}). 
%($[\cdot]$ are the Iverson brackets -- )
\maybe{or maybe a more quantitative version? like $\sum_i (U_i)_\pi - (U_i)_{closest BNE}$.}
\maybe{and maybe use *partial loss funcitons* instead, e.g., the partial prediction accuracy just meaning that self-fulfilling prophecies are preferred over everything else, but not the preference relations amonf the other outcomes}
%
%\item \todo{Pareto partial ordering}
\end{citem}
We call a policy $\pi$ that tries to optimize $\lpred_\pi$ loosely a \defi{(predictive coordination) assistant}, and a $\pi$ that achieves $\lpred_\pi=0$ a \textbf{\defi{(formally) self-fulfilling prophecy (policy)}}.

\sectionc{The utility of predictions for coordination -- analysis} %ve predictions} %accuracy objective $\lpred$}
\label{sec:pred-sol}
\label{sec:obj}
\label{sec:rel_epistemic}

%As mentioned in the introduction, the overall ``qualitative'' goal of this work is to design assistants that support coordination between the users $i \in N$ of the facilities. Since this goal itself is hard to make rigorous (we elaborate on this in Section \ref{sec:obj_disc}), in particular due to the multi-agent nature of the problem, we introduced prediction accuracy $\lpred_{\pi}$ as rigorous ``proxy'' objective. In this section we want to understand to what extent this is a reasonable approximation: % in light of the overall (but vague) goal being coordination.
%%More specifically, we aim at the following:
%
%In this section we want to understand to what extent prediction accuracy as objective function is a reasonable approximation to the overall but vague goal of coordination.

%The overall qualitative goal -- assisting coordination -- is hard to make rigorous, also due to the multi-agent nature of the problem (we elaborate on this in Section \ref{sec:obj_add}). So let us try to understand the potential of prediction accuracy as a quantitative proxy objective:

%Let us try to understand what predictive assistants (i.e., assistants that that try to optimize prediction accuracy) can at most achieve, in light of the overall goal of coordination as measured by the utilities that are experience by the users:

In this section, we pursue the following goal,
for which the one-stage setting we introduced in Section~\ref{sec:pre} is sufficient (we will introduce a \emph{repeated} version in Section~\ref{sec:defdyn}).

\begin{Goal} %[Objective Clarification Goal] %{1}
\label{quest:pred_acc}
\label{quest:obj}
%Answer the question: under what conditions and to what extent (in terms of users' utilities $(U_i)_{i \in N}$) can predictive assistants support coordination (given they achieve an optimal policy $\pi$ under the prediction accuracy objective $\lpred$, i.e., $\pi \in \arg\max_{\pi'} \lpred_{\pi'}$).
%Answer the question: Under what conditions do predictive assistants (given they achieve $\pi \in \arg\max_{\pi'} \lpred_{\pi'}$) help solve the users' coordination problem (in the sense of selecting an equilibrium)?
Understand the conditions when, and the degree to which, assistants, that achieve $\pi \in \arg\max_{\pi'} \lpred_{\pi'}$, help solve the problem of coordination between users of facilities (here: in terms of equilibrium selection). %' coordination problem
\end{Goal}

\subsection{Characterization step in general setting}
\label{sec:chargen}
\label{sec:social_welfare_remark}

\begin{Theorem}[Self-Fulfilling Prophecy Characterization]
\label{prop:predne}
\label{thm:predne}
\label{prop:argmin_ne}
We have, in the general setting (Setting \ref{set:gen}, with all users being inference-assistable and assistant-separable): % and assistant-best-responding):
\begin{citem}
	\item If the assistant policy $\pi$ in the assistant-based system $M$ (where all users are assistant-best-responding) is a self-fulfilling prophecy (i.e., $\lpred_{\pi} = 0$)% prediction accuracy $\lpred_\pi=0$
	, then the corresponding strategy profile $s_\pi$ \maybe{(Eq.\ \ref{eqn:s})} is a Bayesian Nash equilibrium (BNE) of the \bg $\game$.
	\maybe{If $\lpred_\pi=0$, then $\lne_\pi=0$.}
	%In words: If the assistant reaches an optimal prediction policy (``self-fulfilling prophecy''), the resulting user behavior is comparable to a BNE in the \bg.
	\item Conversely, \maybe{in the case of finite ranges, }if the strategy profile $s$ is a strict BNE of the \bg $\game$, then the corresponding assistant policy $\pi_s$ is a self-fulfilling prophecy. % (i.e., $\lpred_{\pi_s}{=}0$). % for the  $\pi_s$\maybe{ (Eq.~\ref{eqn:cp})}.
	%In words: If $s$ is a strict BNE of the \bg, then the assistant reaches an optimal prediction policy by predicting the user behavior under $s$.
	\maybe{\item In particular, if a strict BNE exists in $G$, then: if $argmin .... L_\pi=0$, then $s_\pi$ is essentially a BNE of $\game$.}
\end{citem}
\end{Theorem}

The proof %, which is rather straight forward, is based on the (common, informal) intuition of the Nash equilibrium as a self-fulfilling prophecy, 
is in Section~\extref{sec:pr_argmin_ne}.
Since in Section \ref{sec:pre} we were very brief regarding some of the (measurability) assumptions and definitions underlying the theorem,
we give a \emph{detailed elaboration} of these assumptions and definitions, and their soundness, in Section \extref{supp:sec:addnas}.
%and some justification of the assumptions in Section~\ref{sec:nerem}.
For a \emph{justification} of some of the theorem's assumptions see Section~\ref{sec:nerem}.
%Note that Regarding \emph{generality} of Setting \ref{set:gen}, note that it is formulated pretty abstractly. In particular, 
Note that Setting \ref{set:gen} is formulated pretty \emph{generally}: the slots $K$ can be any set of options the users have. $K$ can be time slots in one shared facility, like a road section; or $K$ can be several facilities that provide the same service, say citizen centers in a city; or it can be a combination, i.e., time slots in several facilities. (Our ``slot'' is similar to ``facility'', or, to some extent, ``feasible combination of facilities'', in congestion games.) The main limitation may be seen in the assumptions of inference-assistability and assistant-best-responding, saying that users can meaningfully evaluate the utility of their choices based on only $Y$, which $A$ is a forecast for.

\subsection{Existence step in \nonagg}
%\section{\Nonagg: objective clarifiaction and dynamic algorithm} %: an algorithm}
\label{sec:atom}
\label{sec:exsmall}

%\subsection{Proof-of-concept algorithm for dynamic system}
%\label{sec:dyn}
%\label{sec:select_dyn}

\renewcommand{\game}{\gat}

%We address Goal \ref{quest:obj} by first establishing a general equivalence between reaching a self-fulfilling prophecy and selecting an equilibrium (\emph{``characterization step''}, Section \ref{sec:chargen}), and second providing conditions under which a self-fulfilling prophecy \emph{exists}, for a small-scale and a large-scale setting (\emph{``existence step''}, Sections \ref{sec:exsmall}, \ref{sec:exlarge}).

In Theorem~\ref{prop:predne}, we characterized the type of solution (a BNE of game $G$) that is implemented by the assistant-based system (with significantly lower requirements on users' knowledge/inference capacities than in the game $G$),
%
% we did a first step towards Goal \ref{quest:obj} by characterizing the users' utilities that are implied 
\emph{if} the assistant reaches a self-fulfilling prophecy policy (``characterization step''). 
We established this result for the general setting (Setting~\ref{set:gen}).
As second step towards Goal \ref{quest:obj}, it remains to understand when such a self-fulfilling prophecy \emph{exists} (``existence step'').
%
%This step turns out to be more instructive when performing it %very difficult in the most general setting due to measure-theoretic issues, so we look at two more restrictive but still general subsettings of the general setting.
This second step we perform separately for two instructive subsettings of the general setting, which are each still reasonably general.
As a warm-up exercise, we start in a setting where we can easily build on game-theoretic results -- because it corresponds to a classical finite Bayesian game.
\maybe{ for the corresponding game.}

%\todo{
%Let us now illustrate 
%For illustration purposes let us now consider a simple, more concrete sett ......
%In this section we address the \nameref{quest:obj} for the following \defi{(incomplete-information  non-aggregated) \nonagg} (a special case of the general setting in Section \ref{sec:scope}):
%}

\subsubsection{Introducing the setting}

%Here we provide a corollary to Proposition \ref{prop:argmin_ne} and a proof-of-concept assistant algorithm for the following \defi{atomic setting}
%In this section we consider another special case of the general setting (Section \ref{sec:scope}), calling it \defi{atomic setting}, providing a corollary to Proposition \ref{prop:argmin_ne} (addressing the \nameref{quest:obj}) and a proof-of-concept assistant algorithm (addressing the \nameref{quest:algo}) for it: All (random) variables have finite range.
\begin{Setting}[\sdefi{\Nonagg}]
\label{set:nonagg}
As a restricted form of Setting \ref{set:gen}, consider the following \defi{\nonagg}:
$\users = \{1, \ldots, n\}$ is finite and we interpret its elements here as \emph{users} (not types), and the individual actions of the users are directly publicly observable, i.e., $Y=\Uact$, % = (\Uact_i)_{i \in N}$.
and $h_i(B) = B_{-i}$ in Eq.~\ref{eqn:uta}.
$X, V, W$ all have finite range. % (and are equipped with the discrete $\sigma$-algebras).
\end{Setting}

We may write $\msmall$ and $\gat$ to denote assistant-based system (Definition \ref{def:m}) and \bg (Definition \ref{def:g}), respectively, canonically associated to this particular \nonagg.
\maybe{Let $\gat$ denote the induced \bg $G$ (Definition \ref{def:g}) for this setting.} % (as defined in Section \ref{sec:inducedgame}). 
%\todo{Maybe also introduce $M^{\text{small}}$?}
%
%For the sake of completeness, 
%Let us formally state a version of Theorem \ref{prop:predne} for this setting which follows immediately. %\todo{maybe drop all the corollaries and only state it informally; in particular to avoid the impression that all our results are trivial}
\begin{Corollary}
\label{cor:at}
Setting~\ref{set:nonagg} is a special case of Setting \ref{set:gen}.
In particular, it satisfies the conditions of Theorem \ref{prop:predne} and hence the theorem's implications hold for $M = \msmall$ and $\gamegen = \gat$.
%
%%We have (in the general setting of Section \ref{sec:model}, with all users $i \in N$ being inference-assistable, assistant-seperable and assistant-best-responding):
%\todo{Setting~\ref{set:nonagg} is a special case of ... (with $h_i(B) = B_{-i}$ in Eq.~\ref{eqn:uta}), and its canonically associated assistant-free system $M$ (Definition ) and \bg (Definition ) are both consistent (w.r.t.\ the measurability assumption).
%	In particular, the conditions of Theorem \ref{prop:predne} are satisfied and so its implications hold for this setting.}
%Setting~\ref{set:nonagg} together with all users $i \in \users$ being inference-assistable and assistant-seperable\maybe{ and assistant-best-responding} (with $h_i(B) = B_{-i}$ in Eq.~\ref{eqn:uta}) satisfies the conditions of Theorem \ref{prop:predne}, and therefore its implications hold for $M = \msmall$ and $\gamegen = \gat$.
%%Therefore the implications of Theorem \ref{prop:predne} hold for $\gamegen = \gat$.
\end{Corollary}

\subsubsection{Self-fulfilling prophecy existence}

To understand the conditions under which a self-fulfilling prophecy exists, based on the second part of Theorem \ref{thm:predne} (or rather: Corollary \ref{cor:at}) it is enough to understand when a strict BNE of $G$ exists.
But in the current \nonagg, $\gamegen = \gat$ is the classical finite (Bayesian) game, which is well understood. 
For example, \citet[Theorems 3, 4]{harsanyi1973games} showed that when assuming that the players' utilities are contaminated by a small additive noise, then there exists a strict equilibrium with probability one.
Furthermore, \citet{bilancini2016strict} establish conditions under which all BNE are (essentially) strict, which entails existence of such strict BNE when combined with general BNE existence results.
%In particular, at least for the complete-information case, for instance it has been shown that small perturbations of a given game often have strict Nash equilibria \citep{harsanyi1973games, govindan2003short}.

\maybe{the ``likelihood'' of existence of a strict Nash equilibrium of a game has been studied \citep{harsanyi1973games, govindan2003short}. \todo{In one sentence, the results state that ...} }
\maybe{So for the current \nonagg the second step of our approach towards Goal \ref{quest:obj} (as outlined at the beginning of the current Section \ref{sec:obj}) is performed by the mentioned work (together with the second part of Theorem \ref{thm:predne}).}

%For the small-scale, i.e., finite-player, game $\gnonat$, conditions for the existence of strict Nash equilibria \cite{harsanyi1973games}, and to some extent also strict Bayesian Nash equilibria \cite{...} \cite{govindan2003short, govindan2003short} have been studied. \todo{for instance, ...}
%%
%So, using the second part of Theorem \ref{thm:predne}, we know that under these conditions (together with the conditions of Theorem \ref{thm:predne}), we know that a self-fulfilling prophecy exists (i.e., an assistant policy $\pi$ with $\lpred_\pi = 0$).
%This completes the second step of our approach towards Goal \ref{quest:obj} for the current \nonagg (as outlined at the beginning of the current Section \ref{sec:obj}).

%This answers the second question of Goal \ref{quest:obj}.
%Altogether, for the current \nonagg we now understand -- to some extent -- the implications of an assistant pursuing the prediction objective $\lpred$: under the mentioned conditions, and when it reaches its optimum, the utilities yielded for the users coincides with those of a Bayesian Nash equilibrium of $\gnonat$.

\subsection{Existence step in \aggr}
\label{sec:nonatom}
\label{sec:nonat}
\label{sec:exlarge}

While the above \nonagg is easy to understand, it has significant limitations: first, the users' actions $\Uact$ have to be fully observable for the (loss $\lpred$ of the) assistant, which is often impossible due to data privacy regulations; and second, there has to be a fixed set of unique users, while in practice the set of users may change of course. %we assume an invariant set of individual users.
Therefore we perform the second step towards Goal \ref{quest:obj} %-- analysis of existence of a self-fulfilling prophecy policy -- 
also for the following \aggr (again a subsetting of the general setting of Section \ref{sec:scope}, different from the \nonagg). It corresponds to \termtech{nonatomic games} \citep{schmeidler1973equilibrium},
%
% (again a special case of the general setting in Section \ref{sec:scope}). % (again a special case of the general setting in Section \ref{sec:scope}). 
and is mathematically more involved, but abstracts away from individual users and in particular only requires a cross-user aggregate of actions to be publicly observed.

\subsubsection{Introducing the setting}

%In what follows, for a random variable like $X$, %let $\dom_X$ denote its range and 
%let $\alg_X$ denote the $\sigma$-algebra on $\dom_X$.

%Here we address both, the \nameref{quest:obj} and \nameref{quest:algo}, for the following \defi{(incomplete-information nonatomic aggregated) \aggr} (a special case of the general setting in Section \ref{sec:scope}).

%\subsection{Setting}
\label{sec:aset}

%\todo{check that i changed asm1 to set3 everywhere}

\begin{Setting}[\sdefi{\Aggr}]
\label{set:agg}
As a restricted form of Setting \ref{set:gen}, consider the following (aggregated) \defi{\aggr}:
%
%\todo{makes sense to merge Assumption 1 in here and simply also put the label for this assumption here just in case i miss any ref (STRONG PRO: Corr3 is about the G induced by this extended setting anyway! PRO: reduces the abstractness of the overall formulation. CON: for instance for Glarge to be a nonatomic game, I guess it already qualifies that $I = [0,1]$. Q: how restrictive are the Aussmptio1 assumptions actually?):}
\begin{citem}
	\item For simplicity, we assume there are only two slots, $K=\{0,1\}$, and that $\Acov, \Usig$ are constant.\footnote{We will prove the main results, Theorem \ref{thm:ex}, for an \emph{arbitrary number $|K|$ of slots} though. The extension to stochastic $\Acov, \Usig$ is less obvious due to measure-theoretic issues.}
	We assume $\users = [0, 1]$ with the Borel sets as $\sigma$-algebra $\usersalg$, and interpret $i \in \users$ as a \emph{type} of user with a certain form of utility function and private signal (similar as \citet{kim1997existence}). 
	Let $\dom_\Uact$ (i.e., the set of possible joint user actions $(B_i)_{\in \users}$) be the set of $\{0, 1\}$-valued Lebesgue-measurable functions on $\users$\maybe{ and $\alg_B$ the Borel $\sigma$-algebra for $L^2(\users)$. OR: the $\sigma$-algebra generated by the singletons.}.
	%\todo{DROP: Let $\dom_X$ be some topological space with $\alg_X$\footnote{$\alg_X$ denotes the $\sigma$-algebra on $\dom_X$ we consider for a random variable $X$.} the corresponding Borel $\sigma$-algebra. }
	%Furthermore, the individual actions of the users are directly observed, $Y=C = (C_i)_{i \in N}$ and, 
	%\item \emph{Regarding domains and $\sigma$-algebras},  (this is based on considering $Y$ as one-dimensional, as we will argue below), 
	\item Let $\dom_{Y_1} = [0,1]$ and $Y_1 := \int \Uact_i r(i|X) d i $, for $(r(\cdot|\mpara))_{\mpara \in \dom_\Mpara}$ a family of continuous (Lebesgue) densities on $(\users, \usersalg)$, continuous also in $x$.
	And let $Y_0 := 1 - Y_1$. The interpretation is that $Y_1$ is the fraction of users that choose slot $1$, i.e., a (stochastic) aggregate of $B$, and $Y_0$ is the remaining amount of users, that choose slot $0$.
	Since $Y=(Y_0, Y_1)$ is fully parameterized by $Y_1$, from now on we consider $Y$ to be 1-dimensional and stand for $Y_1$.
	%
	%\footnote{Accordingly, in the corresponding assistant-based system, let $\dom_A$ be the set of Borel measures on $[0, 1]$, interpreted as forecast for $Y_1$ and thus implicitly also for $Y_0 = 1 - Y_1$.
	%}
	%
	%$Y_1$, with $\dom_{Y_1} = [0,1]$, in this setting is interpreted as the fraction of users choosing slot $1$, i.e., a (stochastic) aggregate of $B$.
	%Specifically, let $Y_1 := \int \Uact_i r(i|X) d i $, for $(r(\cdot|\mpara))_{\mpara \in \dom_\Mpara}$ a family of continuous (Lebesgue) densities on $(\users, \usersalg)$, continuous also in $x$.
	%And $Y_0 := 1 - Y_1$, interpreted as the remaining users, that choose slot $0$.
	%
	%$Y_0$ is 
	%..... \emph{Regarding the outcome $Y$ generating mechanism}, let there be a family $(r(\cdot|\mpara))_{\mpara \in \dom_\Mpara}$ of continuous (Lebesgue) densities on $(\users, \usersalg)$, continuous also in $x$, and let $Y_1 := \int \Uact_i r(i|X) d i $ and accordingly $Y_0 = 1 - Y_1$.
	
	\item
%	Regarding users $i \in \users$ and utilities, let $h_i$ (Eq.~\ref{eqn:uta}) be the identity, %, $\bar{U}_i(X, \Uact) = \tilde{U}_i(\Usig, \Uact_i, Y)$
%	%all $i \in N$ in the assistant-free system $M$ be inference-assistable (Eq.~\ref{eqn:uta}), with $h_i$ simply the identity, 
%	and let $(i, y) \mapsto \tilde{U}_i(k, y)$ be a polynomial in $i, y$, for all $\slot \in \Slots$ (since $W_i$ is constant and $h_i$ the identity, we can simply drop them from the general formulation of $\tilde{U}_i$, Eq.~\ref{eqn:uta}).
%	This means, in particular, that the utilities only depend on the \emph{amount} of users at the various slots, not on their \emph{identities}.
%	%\footnote{In the corresponding assistant-based system, as a tie-break rule for the assistant-best-responding Eq.~\ref{eqn:mbr}, assume users pick the slot with the lowest number if several yield the same utility.}
Regarding users $i \in \users$ and utilities, let $h_i$ (Eq.~\ref{eqn:uta}) be the identity, %, $\bar{U}_i(X, \Uact) = \tilde{U}_i(\Usig, \Uact_i, Y)$
%all $i \in N$ in the assistant-free system $M$ be inference-assistable (Eq.~\ref{eqn:uta}), with $h_i$ simply the identity, 
and let $(i, y) \mapsto \tilde{U}_i(k, y)$ be a polynomial in $i, y$, for all $\slot \in \Slots$
(we dropped $W_i, h_i$ from the general $\tilde{U}_i$, Eq.~\ref{eqn:uta}). %\todo{, and let $\tilde{U}_i(k, y) - \tilde{U}_i(l, y)$ be non-constant in $i$, for any $k, l, y$} 
This means, in particular, that the utilities only depend on the \emph{amount} of users at the various slots, not on their \emph{identities}.
For any $k \neq l \in \slots$, let $\tilde{U}_i(k, y) - \tilde{U}_i(l, y) = \sum_m i^m q_m(y)$ be such that, for at least one $m \geq 1$, $q_m(y)$ 
is nonzero and constant in $y$.
\end{citem}
\end{Setting}
%
%\todo{in $\mlarge$, %$i$ be 
%	assistant-best-responding, picking (w.l.o.g.) \todo{the slot with the lowest number if several yield the same utility (i.e., as tie-breaking rule)} slot $0$ if both yield the same utility (i.e., as tie-breaking rule).}
%
Note that, while in practice of course the set of (simultaneous) users and thus also (simultaneous) types of users is finite, having $\users=[0, 1]$ can be seen as an approximation with nice theoretical properties to real settings with \emph{many} users.
We may write $\mlarge$ and $\gnonat$ to denote assistant-based system (Definition \ref{def:m}) and \bg (Definition \ref{def:g}), respectively, for this particular \aggr.%
\footnote{In the assistant-based system $\mlarge$, let $\dom_A$ be the set of Borel measures on $[0, 1]$, since $A$ is a probabilistic forecast for $Y$ ($=Y_1$).
\maybe{MAYBE DROP: Furthermore, in $\mlarge$, as a tie-break rule for the assistant-best-responding Eq.~\ref{eqn:mbr}, assume users pick the lowest (w.r.t.\ the ordering on $K$) slot $k$ if several yield the same utility.}}
\maybe{Let $\gnonat$ denote the induced \bg $G$ (Section \ref{sec:pre}) for this setting.}
$\gnonat$ can be seen as an incomplete-information \emph{nonatomic game}, related to \citep{kim1997existence} but different in that our state can have uncountable range, see also Section~\extref{sec:supprelated}.%, in particular having both, types \emph{and} private signals (for each type).
\footnote{The distribution over the types (which is not to be interpreted as a probability -- rather as one actual realization) is \emph{random}, turning it into an incomplete-information setting.
The name ``nonatomic'' comes from the fact that one considers \socalled{nonatomic measures} on the type space $\users=[0, 1]$.}
For the sake of completeness, let us formally state a version of Theorem \ref{prop:predne} for this setting, proved in Section \extref{sec:pr_cn}. %\todo{maybe drop all the corollaries and only state it informally; in particular to avoid the impression that all our results are trivial}
%\begin{Corollary}
%\label{cor:nonat}
%In the \aggr specified above, let Assumption \ref{asm:aggr} hold true.
%Then the implications of Theorem \ref{prop:predne} hold for $\gamegen = \gnonat$.
%\end{Corollary}
\begin{Corollary}
\label{cor:nonat}
Setting~\ref{set:agg} is a special case of Setting \ref{set:gen}.
In particular, it satisfies the conditions of Theorem \ref{prop:predne} and hence the theorem's implications hold for $M = \mlarge$ and $\gamegen = \gnonat$.
%Setting~\ref{set:agg} \maybe{together with Assumption \ref{asm:aggr}} satisfies the conditions of Theorem \ref{prop:predne}, and therefore its implications hold for $M = \mlarge$ and $\gamegen = \gnonat$.
%Setting~\ref{set:agg} together with Assumption \ref{asm:aggr} forms a special case of the general setting (Section \ref{sec:model}).
%Therefore the implications of Theorem \ref{prop:predne} hold for $\gamegen = \gnonat$.
\end{Corollary}

%\subsubsection{Characterization}

%The \nameref{prop:argmin_ne} entails the following:
%
%
%
%%For simplicity of notation, let 
%%\todo{say that $Y$ in this section is considered as $Y_1$, defining $Y_0 = 1 - Y_1$? also relevant for expodamp section}
%
%%\todo{MAYBE INSTEAD OF COROLLAY: It is easy to see that the \aggr specified above together with Assumptions \ref{asm:aggr} satisfy the conditions of the \nameref{prop:predne} and thus this result particularly holds for this setting.}
%
%
%%It is easy to show that this setting fulfills the requirements of Proposition \ref{prop:predne} and so we have the following corollary to it.
%\begin{Corollary} %[Self-fulfilling Prophecy Characterization for \Aggr]
%\label{cor:nonat}
%In the \aggr specified above, let Assumption \ref{asm:aggr} hold true.
%%Assume the assistant-free system $M$ obeys the Non-Aggregated Setting specified above
%%In the current setting, assume $C_{-i} \ind \Usig_i | W, i \in N$ %(in this discrete setting, $C$ has the natural $\sigma$-algebra and pushforward measure) 
%%and let assumption A3 hold true. Then:
%Then the implications of Theorem \ref{prop:predne} hold for $\gamegen = \gnonat$.
%%\begin{iitem}
%%\item If $\lpred_\pi=0$, then $s_\pi$ is a BNE of $\game$.
%%\item Conversely, if $s$ is a strict BNE of $\game$, then $\lpred_{\pi_s}=0$. %, where $\pi_s(w) := P_{G, s}(\sur|w)$ for all $w$.
%%\end{iitem}
%\end{Corollary}

\maybe{
- linearize AROUND EQUILIBRIUM POINT because then optimal prediction nicely relates to NE again
- big question: does this really work on aggregate level? to what exactly do users best-respond?
- maybe it works if we assume that we linearize around the aggregate of an actual profile and somehow users can reconstruct their position form the aggregate only? on the other hand, this would contain ``recommendation'' parts again and not just ``prediction of others'' parts.}

\subsubsection{Self-fulfilling prophecy existence}
\label{sec:existence}

In contrast to the \nonagg, for the \aggr and the corresponding \bg $\gnonat$ there is less established work that helps to understand existence of a self-fulfilling prophecy policy.
Intuitively, a key question in this \aggr is: can a forecast that only forecasts an \emph{aggregate} of the users' actions (the $Y$ of Setting \ref{set:agg}) actually be a self-fulfilling prophecy and thus help for coordination? For instance, as observed by \citet{marevcek2016r}, if the population of users is completely homogeneous, they will all respond in the same way upon receiving the same input, making coordination difficult.
Here is our answer for this question -- the second of our two main theoretical results. %, proved in Section \extref{sec:expr}:

\begin{Theorem}[Large-Scale Self-Fulfilling Prophecy Existence]
\label{thm:ex}
%In the \aggr specified above, let Assumptions \ref{asm:aggr} hold true.
%In Setting~\ref{set:agg}, %let Assumption \ref{asm:aggr} hold true.
%Then there exists an assistant policy $\pi$ such that prediction accuracy $\lpred_\pi = 0$.
%Then 
There exists a self-fulfilling prophecy policy $\pi$ in the assistant-based system $\mlarge$ (in Setting~\ref{set:agg}).
\end{Theorem}

This implies, based on Corollary \ref{cor:nonat}:

\begin{Corollary}[Large-Scale Bayesian Nash Equilibrium Existence]
\label{cor:ex}
%In Setting~\ref{set:agg}, 
%Let Assumption \ref{asm:aggr} hold true.
The \bg $\gnonat$ (for Setting~\ref{set:agg}) has a Bayesian Nash equilibrium (BNE).
\end{Corollary}

%\paragraph{Interpretation: self-fulfilling prophecy characterization step towards Goal \ref{quest:obj}.}
\paragraph{Proof idea and interpretation:} % of Theorem \ref{thm:ex} as second step towards Goal \ref{quest:obj}:} % and Corollary \ref{cor:ex}:}
The proof of Theorem \ref{thm:ex}, which is given in Section \extref{sec:expr} for an arbitrary number of slots $K$, is based on the Leray-Schauder-Tychonoff fixed point theorem, harnessing the compactness of the set of Borel measures, $\dom_A$, under a weak topology. The most important implication of the theorem is that $\min_\pi \lpred_\pi = 0$. And therefore, together with the first step in the form of Theorem~\ref{prop:argmin_ne}, it shows that an assistant that only forecasts an aggregate can nonetheless, when it achieves its optimum, help ``solve'' the coordination problem -- select a BNE. The intuition behind the assumptions is that types and their utility functions have to be \emph{diverse}. Corollary \ref{cor:ex} can be seen as stand-alone, purely game-theoretic result for $\gnonat$.

%Therefore, here we provide insight into the conditions that are suffcient such that coordination can work based on a (probabilistic) forecast of the \emph{aggregated, anonymous} $Y$ alone, by giving conditions under which $\pi$ with $\lpred_\pi = 0$ exists.
%
% 
%In words: there exists a self-fulfilling prohecy, i.e., a probabilistic forecast that, upon announcing it, coincides with the actual distribution of outcomes.

%\todo{why and where did i write that we assume that $i$s utility just depends on the slot where $i$ goes}

\subsubsection{An instructive linear special case}
\label{sec:linagg}
\renewcommand{\phi}{\varphi}
\newcommand{\expe}{\E_{Y' \sim a}(Y')}

%\todo{two slots}
%
Let us consider a simple special case of the \aggr (which is not central to understand the rest of the paper and can be skipped). On the one hand, this helps to get an intuition for Theorem \ref{thm:ex}, on the other hand this will justify assumptions we will make in the analysis of our algorithm in Section \ref{sec:expodamp}. Assume the utility $\tilde{U}_i(k, y)$ of Setting~\ref{set:agg} \maybe{ (we dropped $\Usig_i$ because it is constant) } is \emph{linear} in $i, y$ for all $k \in K$ (making the users ``risk-neutral''). So $\tilde{U}_i(1, y) - \tilde{U}_i(0, y) = i + \phi y + \chi$, for $\phi, \chi \in \R$.
Let $r(i|x) := \frac{1}{2 \delta} [x - \delta \leq i \leq x + \delta], i \in \users, x \in \dom_X$, with $[\cdot]$ the Iverson bracket (i.e., density of the uniform on $[x-\delta, x+\delta]$), and
let $P_X$ be the uniform on $[\delta, 1-\delta]$.
%$r(i|x) := \frac{1}{\mpara} [i \leq \mpara]$
Then the value of $Y$ as a function of $A=a, X=x$ is, for $H$ the Heaviside function, given by
\begin{align}
&\int H\left( \int \tilde{U}_i( 1, y) - \tilde{U}_i( 0, y) d a(y) \right) r(i|\mpara) di \\
%&\int H\left( \int ( i + \phi y) d a(y) \right) r(i|\mpara) di \\
&= \int H\left( i + \phi  \E_{Y' \sim a}(Y') + \chi \right) r(i|\mpara) di \\
%&= \int [ i \geq - \phi \expe - \chi ]  r(i|\mpara) di \\
&= \frac{\phi }{2 \delta} \expe + \frac{1}{2 \delta} \mpara + \frac{\delta + \chi}{2 \delta},
\end{align}
for $x - \delta \leq - \phi \expe - \chi \leq x + \delta$, and 0 or 1, respectively, otherwise -- a \emph{piece-wise linear function} in $\expe, x$.

First, this shows that under the mentioned assumptions, $Y$ (and its distribution) only depends on the mean $\expe$, but no other properties of $a$.
In particular,  $\lpred_{\pi} = 0$ iff $\lpoint_{\pi'} = 0$, for $\pi$ an appropriate probabilistic extension of $\pi'$, and
\begin{align}
\label{eqn:ptf}
\lpoint_{\pi'} := \E_{M, \pi'} \left( \| A - \E( Y | \Acov ) \|_2^2 \right)
\end{align}
a \textbf{\emph{point prediction version}} of the probabilistic \textbf{\emph{prediction accuracy}} loss $\lpred_{\pi}$.%
\footnote{The reason why we take this definition of $\lpoint_{\pi}$ instead of, say, some form of ``$\E( ( A^{t} - Y^{t} )^2 | \Acov )$'', is because the (distribution of) $Y$ depends on $A$. And it may happen that the latter quantity, which is some form of ``variance'' that depends on the distribution of $Y$, is lower for a non-fixed point $A$ than for a fixed point $A$, which would hurt the relation to equilibrium selection. A similar reason underlies our definition of $\lpred_\pi$.}
This justifies for the assistant to provide \emph{point forecasts} under the above assumptions.
Second, this justifies a (locally) linear model for $Y$ in $(\expe, x)$ and noise.
Note that Theorem \ref{thm:ex} restricted to this simple linear case is immediate based on the intuitive fact that a generic linear function has a fixed point.

\sectionc{Setting for algorithm part -- control dynamics} % assistant-based system}
\label{sec:interdef}
\label{sec:dyndef}
\label{sec:defdyn}
\label{sec:rep}

%\todo{MAYBE PUT THE FOLLOWING INTO INTERMEDIATE DEFINITIONS} By \defi{aggregation/distribution} \todo{do i even still need this?} of $c$ we mean $y=(y_k)_{k=1}^d$ with $y_k = \sum_{i : c_i = k} 1, k= 1,\ldots,d$, i.e., the histogram of values in $c$.

To prepare the algorithm part of the paper, let us extend the general one-stage setting (Setting \ref{set:gen}) and the assistant-based one-stage system (Definition \ref{def:m}) to a \sdefi{general dynamic setting} and an \sdefi{assistant-based dynamic system} $\mr$, respectively, in the following ``natural'' way.
This directly implies also dynamic extensions of small-scale and large-scale setting (Settings \ref{set:nonagg} and \ref{set:agg}) and the corresponding assistant-based systems (we do not introduce explicit symbols for them though).

%
%, as well as the assistant-based systems that is based on it,  are based on them  introduce a \emph{dynamic/repeated} extension of the one-stage setting that
%
%
%, small-scale and large-scale   \ref{set:nonagg} and \ref{set:agg}
%
% we introduced in Section \ref{sec:scope}. 
%%Since we will use this repeated version also in the \nonagg later on, here we introduce it in its general form, not restricting to the \aggr.
%%Let $t \in \N$ denote the repetition. 
%
%
%%\stress{Variables} We call an individual instance of this scenario \defi{(one-)stage system} while the dynamic, repeatad (e.g., day-wise) version of this system we call \defi{dynamic/repeated system}. We denote a variable, say $A$, in the $t$-th repetition by $A^t$, $t \in \N$.
%
%
%%\paragraph{Dynamic assistant-based system:}
%We define a \textbf{\defi{model $\mr$ for the assistant-based dynamical system}} as follows: 
The dynamic extensions consists of $\N$ copies of the one-stage versions, called \emph{stages/repetitions}. % assistant-based system $M$ (from Section \ref{sec:model})\maybe{(over an implicit unified measurable space .....)} which we call \defi{repetitions} or \defi{stages}.
We denote variables, say $A$, in the $t$-th repetition by $A^t$, $t \in \N$. 
%Regarding the \emph{repeated system} we assume its structure to be the r.h.s\ of Figure \ref{fig:all_DAGs}.
Furthermore, %$\mr$ 
the dynamic extensions contains the following equations that replace/extend the ones of repetition $t$  -- think of it as a form of  feedback control model, a partially observable Markov decision process (POMDP) \citep{sutton1998reinforcement} (from the perspective of the assistant): 
\begin{align}
X^t &= \bar{X}(X^{t-1}, E^t) \\
A^t &= \pi(\Acov^{0:t}, A^{0:t-1}, Y^{0:t-1}) , \label{eqn:ye}
\end{align}
with $E^t$ independent stochastic error terms, (measurable) function $\bar{X}$, and (measurable) \defi{dynamic assistant policy $\pi$}.\footnote{$A^{0:t-1}$ means $A^0, \ldots, A^{t-1}$; similarly for other variables.} The gray, dashed arrows of Figure \ref{fig:dag} indicates this dynamic extension.
Regarding the assistant's objectives, let $\lpredt, \lpointt$ be defined similarly as $\lpred, \lpoint$, but additionally conditioning on the observed past: %, in particular, for all $\pi$:
\begin{align}
\lpredt_\pi \! = \!\E_{M, \pi} \!\left( d \!\left(P_{M, \pi} \!\left(Y^t|A^{0:t-1}\!, V^{0:t}\!, Y^{0:t-1} \right)\!, A^t \right) \right) \!,
%\lpredt_\pi = \E_{M, \pi} \left( \| A^t - \E_{M, \pi}\left(Y^t|(A, Y)^{0:t-1}, V^{0:1}, Y \right) \|^2_2 \right)
\end{align}
for all $\pi$, and similarly $\lpointt_\pi$.
%
%, but w.r.t.\ the variables in stage $t \in \N$ and additionally conditioning on past $\Acov^{1:t-1}, A^{1:t-1}, Y^{1:t-1}$.
Remember: \emph{stage $t$} must not be confused with \emph{(time) slot $k$ within} one stage.
%\todo{maybe simplify this statement:} 
%Note that analogously to the \nonagg (Setting \ref{set:nonagg}) and \aggr (Setting \ref{set:agg}) we defined above, %as special cases of the general setting (Section \ref{sec:model}), 
%one can define \emph{dynamic} \nonagg and \aggr as special cases of the above dynamic assistant-based system.
%
%
%%\paragraph{Objectives}
%
%%\todo{How to handle autocorrelations in the Kalman case? Put past observations into $W$?}
%
%
%
%%\todo{Modified goal: prediction accuracy for the case $A \in \mathcal{C}$, i.e., point forecasts: $\lpoint = \lpoint_\pi := \E_\pi(\|A-Y\|_2^2)$ (Section \ref{sec:expodamp}). However, this is just a special case of the probabilistic forecast goal once say the variance is fixed and assumed to be broadcasted as well.}
%
%
%
%
%
%
%
%
%
To motivate the algorithmic part below, let us give two examples of naive dynamic assistant policies that fail.
%an example to demonstrate what can go wrong in the scenario we consider. %, and then present an assistant algorithm that avoids such poor dynamics.
\begin{Example}[Naive assistant yields oscillation]
% \todo{where to put this? doesn't fit here because it's atomic. put it at the beginning -- inflates the non-technical part; putting it into the atomic part -- may be too late to give an intution there}]
\label{ex:oscil}
\label{expl:oscil}
Consider a toy scenario of two users, $i=1,2$, two slots, $\Slots = \{0, 1\}$, $\Usig, \Acov, X$ constant, and $(0,1)$ and $(1,0)$ the (pure) Nash equilibria of the induced complete-information \bg. For simplicity, let $\Uact$ be directly observed ($Y=\Uact$), let $A$ be a point forecast. As usual, assume each day $t$ both users best-respond to $A^t$.
The assistant starts with, say, $A^0=(0,0)$ and then, naively, each day takes yesterday's outcome $\Uact^{t-1}$ as forecast for today, $A^t$.
It is easy to see that this will lead to an \emph{overshooting and oscillating} system $\Uact^0 = (1,1), \Uact^1 = (0,0), \Uact^2 = (1,1), \ldots$ (called \socalled{flapping} by \citet{marevcek2015signalling}).
\end{Example}

\begin{Example}[``I.i.d.'' assistant is sub-optimal]
%Recall the oscillation Example \ref{expl:oscil} where we assumed $C$ to be observed. %, where simultaneous best-responses of users and a naive assistant policy lead to an oscillation $c^1 = (0,0), c^2 = (1,1), c^3 = (0,0), \ldots$.
%This example also nicely illustrates how the optimal policy cannot always be identiefied from data under a different policy.
Classical forecasting applied to the sequence $\Uact^1, \Uact^2, ...$ from Example \ref{expl:oscil} would yield the empirical distribution $P(\Uact=\uact) = \frac{1}{2} ( \delta_{(0,0), \uact} + \delta_{(1,1), \uact} )$, with $\delta$ the Dirac delta, as optimal probabilistic forecast $A^t$ -- under some stationarity  assumption.
But the actual best forecast would be a Dirac delta on one of the two Nash equilibria $(0,1)$ and $(1,0)$ (Theorem \ref{prop:argmin_ne}; we ignore mixed equilibria here).
%In the next section we present assistant algorithms that do make the users converge to a Nash equilibrium.
\end{Example}

\sectionc{Predictive assistant algorithms with guarantees}
\label{sec:alg}

In the first part, we analyzed conditions under which predictive assistants help coordination (in terms of the equilibrium selection objective, Section~\ref{sec:preobjectives}), \emph{if} they manage to optimize prediction accuracy, leaving open the \emph{``how''}.
Therefore, as second part of the paper, we address:

\begin{Goal} %[Algorithm Design Goal]
\label{quest:algo}
Design algorithms for the assistant policy $\pi$ in the dynamic assistant-based system $\mr$ that optimize prediction accuracy $\lpredt$ (and asymptotically select an equilibrium, if possible), learning from past interactions. %\footnote{Recall that the users utilities (and utility functions) are hidden .... form user reactions ... }
% (while keeping in mind that $\lpred$ is just a proxy objective).
%
%, at least asymptotically, while also keeping in mind that it is just a proxy and other objectives like $\lne$ are relevant as well (i.e., design predictive coordination assistants). Try to minimize the requirements on users (utilities, explicit information provided, ...) because they are costly, instead try to learn as much as possible from (behavioral) data.
%
\end{Goal}

We will consider dynamic versions of the two settings for which we established in Section \ref{sec:obj} that predictions can help coordination: \aggr (in Section \ref{sec:expodamp}) and \nonagg (in Section \ref{sec:consistresp}).
%
%
%Similar as in Section \ref{sec:obj}, we will consider two settings: the (dynamic) \aggr (in Section \ref{sec:expodamp}) and the \nonagg (in Section \ref{sec:consistresp}). % for the dynamic system $\mr$ just introduced in Section \ref{sec:defdyn}.
For each setting, we propose an assistant algorithm $\pi$, and provide a theoretical analysis of its dynamics/convergence.
A unifying idea behind both algorithms is that they mitigate certain bad user behavior, e.g., ``overshooting'' due to too many users jumping to the same purportedly ``good'' slot, helping convergence to a Nash equilibrium (of the \emph{stage} benchmark game).
Recall that \emph{users' utilities (functions) are hidden} from the assistant (Definition~\ref{def:m}), so the assistant's inference (about the equilibrium) is mainly based on behavioral data of how users react to forecasts.

%\todo{unifying ideas: overshooting; harness assistant to get convergence guarantees (compared to classical best resp)}

\subsection{Expodamp for large-scale setting}
%Exponential smooting method and identifiability from machine learning and linear dynamic systems view}
\label{sec:learn}
\label{sec:general_dyn}
\label{sec:meth}
\label{sec:expodamp}
\label{sec:stab}

%Here we propose an algorithm for the dynamic system $\mr$ (Section \ref{sec:rep}) in the \aggr (Setting \ref{set:agg}).
%
%\todo{maybe drop obs noise altogether; merge the two propos}
%
%Let us now turn towards the \nameref{quest:algo} for the repeated system $\mr$ (Section \ref{sec:rep}) in the \aggr. 
Consider the dynamic \aggr\footnote{In particular, $Y$ is considered 1-dimensional (since $Y_1$ determines $Y_0$). The extension to more slots is straight forward.}
(Section \ref{sec:defdyn}) and let $A$ be a \emph{point forecast} for $Y$, i.e., $\dom_A = \dom_Y$, and consider $\lpointt_\pi$ as loss (dynamic version of Eq.~\ref{eqn:ptf}, as described in Section~\ref{sec:dyndef}). Recall that in Section~\ref{sec:linagg} we gave conditions that justify this point prediction approach.

%For simplicity, let us assume that only the mean of $Y$ is of relevance to the users (for a justification of this assumption see below).
%In this sense, let $A$ be a \emph{point forecast} for $Y$, i.e., $\dom_A = \dom_Y$. Accordingly, as loss function we consider a ``point forecast version'' of prediction accuracy, given $\pi$:
%\begin{align*}
%\lpointt_\pi := \E_\pi \left( ( A^{t} - \E( Y^{t} | A^{0:t-1}, Y^{0:t-1} ) )^2 \right) .
%\end{align*}

We propose \defi{Expodamp} as described in Algorithm \ref{alg:expodamp} as the assistant's dynamic policy $\pi$. 
The intuition behind Expodamp is that this formula can dampen oscillations due to ``overshooting'' user behavior (Example \ref{ex:oscil}) but it can also accommodate for non-stationarities in user preferences. These intuitions will be made rigorous in the proposition below.%
\footnote{The formula in Algorithm \ref{alg:expodamp} is a case of a so-called exponential smoothing method \citep{hyndman2008forecasting}. However, so far (to the best of our knowledge) it has only been applied to classical forecasts that do not influence the outcome. In a sense, we generalize the established method to this new setting.}

%\begin{Question}
%What is the overall dynamics implied by Expodamp?
%One specific question: does this relate to dampening again?
%If we do not assume a specific ``multi-agent model'' with specific (say best response) behaviour as we did in the last section, but rather take a ``machine learning'' approach of performing inference based on a very abstract model (which, by the way, comprises the setting where the assistant has no infuence as special case): which assistant policy $\pi$ makes sense and what guarantees do we get? Let us first give an example of what can go wrong when naively approaching this forecasting problem.
%\end{Question}

%We now come closer to practice by first proposing a specific method -- which we will later also apply in the real experimental setting -- and then analyzing it.

%\begin{wrapfigure}{r}{0.5\textwidth}
%\begin{minipage}{.5\textwidth}
\begin{algorithm}[t] %[H] %[tb]
% for custom keywords etc., see 9.5.1 in algorithm2e.pdf
\caption{Expodamp (\aggr)}
\label{alg:expodamp}
%\KwSty{test command block} 
%\KwIn{ parameter: $\alpha$ } %A^{0}, 
\KwSty{Input:} parameter: $\alpha$\\
\For{each stage $t\geq 1$}{
	%\KwIn{$A^{t-1}, Y^{t-1}$}
	%\KwOut{$A^t := A^{t-1} + \alpha(Y^{t-1} - A^{t-1})$}
	\KwSty{Input:} $A^{t-1}, Y^{t-1}$\\
	\KwSty{Output:} $A^t := A^{t-1} + \alpha(Y^{t-1} - A^{t-1})$\\ \maybe{Maybe define it in the general way as in the proof: $A^t := A^{t-1} + \alpha_t (Y^{t-1} - A^{t-1})$ for sequence of parameters $(\alpha_t)_{t \in \N}$}
}
\end{algorithm}
%\end{minipage}
%\end{wrapfigure}

%Let us now present some optimality and convergence gurantees for Expodamp. The proofs can be found in Section \extref{sec:pr_stab}.
%For this analysis, let us assume $Y^t$ is univariate.%
%\footnote{The univariate case includes the case of bivariate $Y_t$ if the total number of users is roughly fixed, and seems to be easily extendable to the multivariate case.} %, because then the outcome can be parametrized by just the outcome of, say, at the first slot.}

\begin{Asm}
\label{asm:exposm}
%Besides the basic model assumptions in Section \ref{sec:model}, from now on, for this section assume, that only $Y^t$ is observed by the assistant at each stage $t$, which is an aggregate (Section \ref{sec:pre}) of the action profile $C^t$.
%Furthermore, assume that dynamics of $X^t$ and users' strategy profile $\sigma$ is such that:
Let the following equations hold for the dynamic assistant-based system $\mr$, $t\geq 1$:
\begin{align}
%A_{t+1} &= A_t + \alpha (Y^t - A^t) , \quad
\tth^t &= \tth^{t-1} + E_{\tth}^t, \label{eqn:x} \\
Y^{t} &= \beta A^t + \gamma \tth^t + E_{Y}^t, \label{eqn:gamma} %\todo{+ \beta Y_{t-1}} . \label{eqn:gamma}
%Y^{t} &= A^t + \gamma (\tth^t - A^t) + E_{Y}^t, \label{eqn:gamma} %\todo{+ \beta Y_{t-1}} . \label{eqn:gamma}
\end{align}
with $E_{\tth}^t, E_{Y}^t$ noise terms that are independent of the past and each other. (This is a \socalled{state-space model} known from the \socalled{Kalman filter} \citep{Luetkepohl2006}.)
%. \maybe{maybe drop $E_{Y}^t$ altogether}
%where $\tth^t$ denotes some form of aggregated equilibrium under private signal profile $\Usig^t$, i.e., is a function of  $\Usig^t$ \todo{hä?}. % which gives us the relevant part of $\Usig$.
\end{Asm}

Recall that in Section \ref{sec:linagg} we gave conditions, in the large-scale setting, that justify the linearity in Assumption \ref{asm:exposm} (note that the $X$ in Assumption \ref{asm:exposm} would correspond to a parameter of the distribution of $X$ rather than to $X$ itself in Section~\ref{sec:linagg}, but we neglect this detail for simplicity of notation).
Also note that Assumption \ref{asm:exposm} is a linear approximation which facilitates the theoretical analysis but comes at the cost of a mismatch to the actual setting: $(Y^t)_{t \in \N}$ in Assumption \ref{asm:exposm} can leave $[0,1]$ in the long run, so the model should rather be seen as a local approximation. Note that, due to convexity, Expodamp will always output $A^t \in [0,1]$ upon $Y^{t-j} \in [0, 1], j \geq 1$ though.
Keep in mind that the fixed point (self-fulfilling prophecy) of the linear function $a \mapsto \beta a + \gamma \bth$ (ignoring the noise term) is $\gamma (1-\beta)^{-1} \bth$ (exists whenever $\beta \neq 1$). In particular, if $\beta = (1-\gamma)$, then the fixed point (corresponding to the self-fulfilling prophecy/BNE) is $\bth$.
We can give the following guarantees%\maybe{which are a restricted version of Proposition \extref{prop:expodfull} which is}
, for which we prove a generalization\footnote{It is formulated slightly cleaner, using the do-operator.}%
in Section \extref{sec:pr_stab}.

\begin{Proposition}[Optimality and Convergence Rate of Expodamp]
\label{prop:expod}
%\todo{maybe shorten this whole section by completely dropping observation noise}
In the dynamic \aggr (Section \ref{sec:defdyn}), let Assumption \ref{asm:exposm} hold true.
Let the assistant's policy $\pi$ be Expodamp (Algorithm \ref{alg:expodamp}).
%Let $\pi$ be defined by 
%\begin{align}
%A_\pi^0 &:= \gamma (1-\beta)^{-1} \E(\tth^0),\\
%A_\pi^{t+1} &:= A_\pi^t +  \gamma (1-\beta)^{-1} Q_t  ( Y^t - A_\pi^t ) , %\text{  for all $t\geq 0$,}
%\end{align}
%for all $t\geq 0$
%with $Q_t$ some function of the covariance structure as detailed in the proof of this proposition. 
%\spar
%In particular, $Q_t$ is such that, if there is no observation noise in the latent-state model, i.e., $\var(E_{Y}^0) = 0$, then $\pi$ coincides with Expodamp (Algorithm \ref{alg:expodamp}) when setting $\alpha:=(1-\beta)^{-1}$. \spar
\begin{citem}
	\item \emph{Stochastic case:} %\todo{move it here?:} 
	In Expodamp, let $\alpha:=(1-\beta)^{-1}$ for the true $\beta$ of Eq.~\ref{eqn:gamma}. \maybe{let $\alpha_t := \gamma (1-\beta)^{-1} Q_t, t \in \N$, with $Q_t$ a quantity that depends on $\gamma$ and the noise variances as defined in \extref{supp:eqn:Q}}
	Assume $E_{Y}^t = 0, t \geq 1$. Then, at each stage $t$, $\lpointpredt = 0$ and
	\begin{align}
	A^{t} {=} \arg\min_{a'} \E\left( \| A^{t} {-} Y^{t}\|_2^2 \mid A^{t}{=}a', A^{0:t-1}\!, Y^{0:t-1} \right).  \label{eqn:sa}
	\end{align}
	%and, in particular, 
	%.
	\item \emph{Deterministic case:} Assume that $\tth^t = \bth$ is constant, that $\beta = (1-\gamma)$ and that $E_{\tth}^t = E_{Y}^t = 0$.
	Then 
	\begin{align*}
	Y^{t} &= \bth +(1-\gamma)(A^0-\bth)(1-\alpha\gamma)^t , \text{  for all $t\geq 0$}. %\\
	%\todo{L^{t}} &= 
	\end{align*}
	That is, $Y^{t}$ converges exponentially with rate $\gamma\alpha$ towards the ``optimum''/fixed point $\bth$ (and thus also $A^t$ converges to $x$ based on Expodamp's formula) if
	$0<\gamma\alpha<2$.
\end{citem}
\end{Proposition}

When applying Algorithm \ref{alg:expodamp} in practice, often one does not know the parameter $\alpha$ a priori and has to infer it.
As a first approximation, it may be learned by naively fitting Algorithm \ref{alg:expodamp} to past observational data %\maybe{i.e., where no assistant forecast was given to users} 
as if it were a classical (non-influential) forecasting method \citep{hyndman2008forecasting}.
In principle however, without going into detail, $\alpha$ rather has to be learned like a \emph{control policy}, based on how the environment responds to it.
%
%
%the parameter $\alpha$ (in the (first part) of Proposition \ref{prop:expod} we sort of assume it as given).
%%
%\footnote{The (first part 
%	\todo{Note that clearly, in practice $\alpha$ cannot be known a priori and has to be estimated from data. Our analysis here applies to the part of the problem where alpha is already estimated.} Let us mention that $\alpha$ cannot in general be trained in the classical way of fitting it to the past distribution \citep{hyndman2008forecasting}, but rather has to be trained like a controller.
%	\todo{something else i wanted to mention ... maybe regarding the difference between det and stoch case in terms of constants}}
%
%\maybe{reformulate the second part in terms of $Y-A \to 0$}

%\todo{problem when goingg from individual-action to aggregate level: users may not know their spot and then first part is not applicable anymore etc.}

%\paragraph{Justification of the point forecast. \todo{maybe move up}}
%Making point forecasts instead of probabilistic ones can be justifies in two ways: If utility functions are linear together with our assumption of risk-neutrality. Or, if we assume a Gaussian model as we will do, together with the fact that under the absence of observational noise, the variance is invariant over time and does not need to be part of the forecast at each step.

\subsection{Partpred for \nonagg}
\label{sec:consistresp}
\label{sec:partpred}

While Expodamp is the main algorithm of this paper, here we also provide a proof-of-concept algorithm for the repeated \emph{small-scale} setting (Section \ref{sec:defdyn}). %, which is not central to the main line of reasoning though.
%In the dynamic assistant-based system $\mr$, 
Assume $X^t$ to be independent of $X^{1:t-1}$, i.e., the special case where the $X^t, t \in N$ are i.i.d.
The algorithm, \defi{Partpred}, is sketched -- for the case that $\Acov$ is constant -- in Algorithm \ref{alg:partpredsk}, and fully described in Section \extref{sec:pr_partpred}.

%\begin{minipage}{1\textwidth}
%\vspace{0pt}  
\begin{algorithm}[tb] %[H] %[H] %[tb]
% for custom keywords etc., see 9.5.1 in algorithm2e.pdf
\caption{Partpred (small-scale; sketch)}
\label{alg:partpredsk}
%\KwIn{parameters: $\bar{A}$, $r$}
\KwSty{Input:} parameters: $\bar{A}$, $r$; initialization: $a \in \bar{A}$\\
%Unless convergence has happened, 
For $r$ steps, output $A=a$ and sample $\Uact$. Let $\hat{P}_{a}^r$ be the resulting empirical distribution of $\Uact$. \label{code:jump}\\
Let $a' := \arg\min_{a'' \in \bar{A}} \| a'' - \hat{P}_{a}^r \|$.
Let $a''$ be obtained by only taking over a subset %(as large as possible) 
of best responses from $a$ and $a'$ such that $a''$ is a correct forecast at least w.r.t.\ some users \\ % (using prior knowledge)\\
%
%\For{each stage $t\geq 0$}{
%%\KwIn{$w:= W^{t}$}
%%Let $I := \{ i : \text{ player $i$ has best-responded to $a_w$ $J$ times} \}$\\
%\eIf{$a_w$ has been announced less than $r$ times or has converged under $W=w$}{
%\KwOut{$A^t := a_w$}
%%\KwIn{$C^t$}
%}{
%Let $\hat{P}_{w, a_w}^r$ be the empirical distribution of $C$ in the $r$ times it was sampled\\ % under $W=w, A=a_w$\\
%Let $a_w' := \arg\min_{a_w'' \in \bar{A}_{w}} \| a_w'' - \hat{P}_{w, a_w}^r \|$\\
%%Let $E := \emptyset$\\
%%\For{$i \in I$ (in random order)}{
%%Let $\tilde{P}_{w, i}$ be the empirical distribution of $C_i$ in these $J$ times\\
%%\If{$\vn{Consist}(E, \{\tilde{P}_{w, j}\}_{j \in E}, i, \tilde{P}_{w, i})$}{
%%$E := E \cup \{i\}$ REPLACE THIS WHOLE THING BY COnist directly returing consistent set
%%}
%%}
%Let $a_w'' := \an{UpdateFunction}(a_w, a_w', w)$\\ %\tcp{defined in supplement}
%%Let $a'_w = a_w$\\
%%\For{$i \in E$}{
%%Let $[a'_w]_i := [c_w]_i$
%%}
\uIf{$a'' = a$}{
	Keep outputting $a$ forever %Remember that convergence happened %\\
	%\KwOut{$A^t := a_w''$}
}
\uElseIf{$a''$ and all other $a' \in \bar{A}$ have been tried $r$ times}{
	Set $a'' := \arg\min_{a'} \| a' - \hat{P}_{a'}^r \|$\\
	Keep outputting $a''$ forever %Remember that convergence happened
	%Let $a_w \in \times_{i \in N} \bar{P}_{w, i}$ be s.t. it has not been announced before, or, if everything has been tried, stick with ...\\
}
\uElseIf{$a''$ has been tried $r$ times}{
	Pick unused $a'' \in \bar{A}$ at random
} \label{code:skrands}
Set $a = a''$ and jump to line \ref{code:jump}   \label{code:skrande} %; %\\
%\KwOut{$A^t := a$}
%}
%}
\end{algorithm}
%\end{minipage}

%by the combination of Algorithms \extref{alg:partpred}, \extrefs{alg:consistupdate} (for the congestion game case) and \extrefs{alg:updategeneral}.
% (for the general stochastic case) of the supplement.
The basic idea is as follows: as long as there is (significant) uncertainty about where the optimum (self-fulfilling prophecy/BNE) would be, the algorithm tries to make a prediction that is at least partially correct (i.e., makes the correct prediction at least w.r.t.\ the behavior of one player). % --  which gives the algorithm its name.
The algorithm combines ideas from \socalled{best-response dynamics} and \socalled{congestion games} \maybe{and \socalled{potential games}} \citep{roughgarden2016twenty} with random exploration whenever the best-response dynamics would cycle.
Let $\bar{A}$ be the (finite) set of all distributions $P_{G, s}(\Uact)$ that arise from (deterministic) strategy profiles $s$ of $\gat$. For simplicity, we assume $\bar{A}$ to be given, but in a next step this could be inferred as well.
We give the following guarantee, sketched for $V$ constant, whose general version is proved in Section \extref{sec:pr_partpred}.

\begin{Proposition}[Convergence of Algorithm \ref{alg:partpredsk} (Sketch)] % -- sketch of Proposition \extref{prop:partpred}]
\label{prop:consistresp}
In the %(dynamic) small-scale 
setting described above, %let all users be inference-assistable and assistant-separable (with $h_i(B) = B_{-i}$ in Eq.~\ref{eqn:uta}).
%\begin{citem}
%\item Let $\Usig$ be fully determined by $W$ and for each value of $W$, let the (complete information) game $\gat$ be a congestion game where all NE are strict. %, which can be seen as our ``prior'' on user utilities.
%%For simplicity, in this deterministic setting, let $A \in \mathcal{C}$ (i.e., an action profile instead of (dirac) distributions over action profiles).
%Let the assistant's policy $\pi_r$ be given by Algorithm \ref{alg:partpredsk} with $\bar{A} = \mathcal{C}$.
%%Furthermore, let Assumption \ref{asm:unibr} hold true.
%Then $\lpredt_\pi, \lnet_\pi \to 0$ for $t \to \infty$ without ever invoking lines \ref{code:skrands} to \ref{code:skrande}, i.e., without needing ``undirected'' search.
%\item 
%Let $\gat$ have a strict BNE.
assume $\gat$ has a strict BNE. % with $\lpred_{\pi}=0$ (a ).
Let the assistant's policies $\pi_r, r \in \N$ be given by Algorithm \ref{alg:partpredsk}, with parameter $\bar{A}$ as defined above.
Then, for any $\varepsilon > 0$, there exists $R, T$ such that for all $r > R, t > T$, it holds that 
%$P( \lpredt_{\pi_r} = 0 ), P( \lnet_{\pi_r} = 0 ) > 1 - \varepsilon$
$P( \lpredt_{\pi_r} = 0 ) > 1 - \varepsilon$ and $P( \text{$s_{\pi_r}$ is a BNE of $\gat$} ) > 1 - \varepsilon$. %\todo{comma notation confusing; either drop BNE or ausschreiben}
%\end{citem}
\end{Proposition}

\sectionc{Experiment} %: real-world evaluation}
%al study and influential assistants}
%: evaluation of our and various existing methods}
\label{sec:exp}

%\todo{at some point, show under what conditions (at least in the linear iid case, potentially also in the state-space case, see prop1) the Lpred can be estimated by the MSE we use below; turn this into a removable section}

Here we empirically evaluate Expodamp (Algorithm \ref{alg:expodamp} for the \aggr) and a baseline.

\paragraph{Experimental setup:} We conducted our experiment in a real-world congested campus cafeteria with around 400 users per day. Here, observation $[Y^t]_k$ is (a proxy to) the number of people in the queue at time $k$ of day $t$.\footnote{While our general considerations allow $Y$ to be queue length, in our large-scale setting the components of $Y$ are the slots, of which queue length is rather something like an integral.} % Clearly, our large-scale model holds only approximately for queue length as $Y$ which means users could be in several ``slots'', but our general model is general enough.}
The coordination assistant in this experiment is a web app which provides the daily forecast (i.e., the forecast is updated once per day, in the morning -- more dynamic versions are future work) to the cafeteria users, to inform their decisions in terms of when to go to the cafeteria.
The web app is used by between 15 and 45 users per day but may influence more (slightly deviating from our model). %
%\footnote{One aspect about our experimental system, the strong external effects of current on future queue lengths, is neither captured by our model nor central to our analysis.}
Besides Expodamp (with parameter $\alpha$ tuned based on a previous observational sample), we evaluate the baseline method \defi{Average} defined by
\begin{align}
a^{t+1} := \frac{1}{t} \sum_{s=1}^t y^s, t \geq 2 \label{eqn:av}
\end{align}
(i.e., treating $y^{1:t}$ as purely observational i.i.d.\ sample).
Expodamp and Average are run as the policy that generates the forecast (which is then provided via the web app to the users of the cafeteria), each for a period of $T=35$ days. See Figure \ref{table:protocol} for an illustration of the experimental protocol. As metric, we use the mean squared error%
\footnote{We use $\ltildepointt$ as a sample-level proxy for the population-level loss $\lpoint$ of Eq.~\ref{eqn:ptf}. We conjecture that, under appropriate assumptions related to Assumption \ref{asm:exposm}, it can be shown that the policy that is optimal under the former loss converges (say, in probability) to a policy which is also optimal under the latter loss. The argument may build on the equivalence between $\lpointpredt = 0$ and Eq.~\ref{eqn:sa} in Proposition \ref{prop:expod}.}
%\todo{suggestion to handle the issue in the following simple way:} We use $\tilde{L}^{t, \text{PointPred}}$ as a sample-level proxy for the population-level loss $\lpoint$ of Eq.~\ref{eqn:ptf}. This can easily be justified in simplistic settings: for instance if $A^t$ is constant over time and we assume Assumption \ref{asm:exposm} as data generating model with $X^t$ i.i.d., then the former converges against the latter almost surely based on the law of large numbers (as detailed in \todo{briefly write down in supplement}). A further analysis is beyond the scope of this paper. \todo{alternative/better arguments could be: (1) $A - Y$ is in fact stationary under the Expodamp model(?). (2) based more on control perspective or on coincidence between population and proxy at least in the optimal case.} }
\begin{align}
\ltildepointt := \frac{1}{T} \sum_{t=1}^T \| a^t -y^t \|^2_2.  \label{eqn:ev}
\end{align}

%\todo{\stress{Relation between metric and $L$.}
%We take the above metric $\tilde{L}^{t, \text{PointPred}}$ as a proxy for $\lpointt$. Note that under various model assumptions, the former in fact converges against the latter. The simplest one would be ... iid .....\\

\begin{figure}[t]
%\usetikzlibrary{arrows.meta}
\contourlength{1pt}
\centering 
\begin{tikzpicture}[scale=1]

\usetikzlibrary{decorations.markings}

\def\MarkLt{4pt}
\def\MarkSep{2pt}
\tikzset{
	TwoMarks/.style={
		postaction={decorate,
			decoration={
				markings,
				mark=at position #1 with
				{
					\begin{scope}[xslant=0.2]
					\draw[line width=\MarkSep,white,-] (0pt,-\MarkLt) -- (0pt,\MarkLt) ;
					\draw[-] (-0.5*\MarkSep,-\MarkLt) -- (-0.5*\MarkSep,\MarkLt) ;
					\draw[-] (0.5*\MarkSep,-\MarkLt) -- (0.5*\MarkSep,\MarkLt) ;
					\end{scope}
				}
			}
		}
	},
	TwoMarks/.default={0.5},
}

\draw[-{Latex[scale=1.2]},TwoMarks=0.2,TwoMarks=0.6,color=black] (0, 0) to (7.5, 0);

\draw[decorate,decoration={brace},xshift=0pt,yshift=0pt,align=center,anchor=north,color=rga] (1.4,-.3) -- (0.1,-.3) node [midway,yshift=-.2cm] {\small{no fore-}\\\small{cast pub-}\\\small{lished}};
\draw[decorate,decoration={brace},xshift=0pt,yshift=0pt,align=center,anchor=north,color=rga] (4.4,-.3) -- (1.6,-.3) node [midway,yshift=-.2cm] {\small{web app announces}\\\small{Expodamp's forecast}\\\small{to campus for 35d}};
\draw[decorate,decoration={brace},xshift=0pt,yshift=0pt,align=center,anchor=north,color=rga] (7.4,-.3) -- (4.6,-.3) node [midway,yshift=-.2cm] {\small{web app announces}\\\small{Average's forecast}\\\small{to campus for 35d}};

\node[obs, align=center,color=rga] at (1.5, .5) (a) {\small 1st intervention};
\node[obs, align=center,color=rga] at (4.5, .5) (b) {\small 2nd intervention};

\begin{scope}[on background layer]

\draw[fill=lightblue,draw=none,path fading=south] (-.2, .9) rectangle (7.6, -2);

\end{scope}

\end{tikzpicture}
\caption{Protocol of our real-world interventional experiment in a large campus cafeteria; steps along Y-axis.}
\label{table:protocol}
\end{figure}
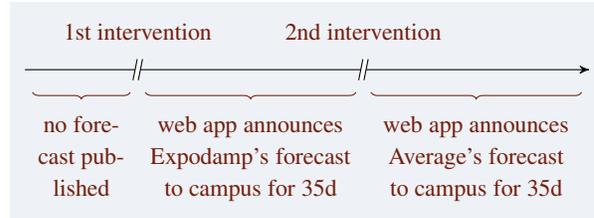

%\begin{table}[t]
%	\begin{tabularx}{\linewidth}{|l|X|}
%		\hline 
%		Stage & Intervention \\ 
%		\hhline{|==|}
%		0 & no
%		\\
%		\hline
%		1, 7 weeks & campus-wide web app publishes Expodamp's output as forecast
%		\\
%		\hline
%	\end{tabularx}
%	%\vspace{0.3cm}
%	\caption{Protocol of the real-world interventional experiment in a large cafeteria.}
%	\label{table:protocol}
%\end{table}

%\begin{figure}[!t]  % so that it is at top of page
%\begin{minipage}[t]{.49\textwidth}
%\vspace{0pt}  
\begin{table}[t]
\begin{tabularx}{\linewidth}{|X|l|}
	\hline 
	Method & $\ltildepointt$ (MSE; Eq.~\ref{eqn:ev}) \\ 
	%\hhline{|==|}
	%\multicolumn{2}{|l|}{\it{Observational plus informative period}}\\
	%\hline
	%Expodamp & 67.5
	%
	% \\ \hline 
	%Average  & 71.7
	% \\ \hline 
	%Kalman & 69.4
	%% \\ \hline 
	%%HoltWinters incl.\ intra-day & 55.0
	%%\\ \hline 
	%%ETS(A,N,A) incl.\ Intra-day & 56.5
	%\\ 
	\hhline{|==|}
	%\multicolumn{2}{|l|}{\it{Influential experimental period}}\\
	%\hline 
	%Expodamp & 77.418  % <- this is already the "final" thing, i.e., eval. on complete perdiod where expodamp was deployed (until 30.7. when seasonalav started) \\
	%\\
	%\hline
	%Average & 92.5
	% \\
	%\hline 
	Expodamp (Algorithm \ref{alg:expodamp}) & 69.56  % <- this is already the "final" thing, i.e., eval. on complete perdiod where expodamp was deployed (until 30.7. when seasonalav started) \\
	\\
	\hline
	Average (baseline; Eq.~\ref{eqn:av}) & 74.25
	\\
	\hline
\end{tabularx}
%\vspace{0.3cm}
\caption{Evaluation shows that Expodamp has higher prediction accuracy.}
\label{table:predeval}
\end{table}
%\end{minipage}%
%\hspace{.02\textwidth}%
%\begin{minipage}[t]{.49\textwidth}
%\vspace{0pt}
%
%\begin{figure}[t]
%	\centering
%	\setlength{\figwidth}{\columnwidth}
%	\setlength{\figheight}{0.635\columnwidth}
%	\input{incl/eval.tex} %raw_plots/eval.tex}
%	\vspace{-0.5cm}
%	\caption{Two-day sample of Expodamp and truth.}
%	\label{fig:predvis}
%\end{figure}
%\end{minipage}
%\end{figure}

\stress{The outcome} is in Table \ref{table:predeval}, showing that Expodamp outperforms Average in this experiment. For illustration, we also show a sample of Expodamp's output $A^t$ and actual outcome $Y^t$, for one day $t$, in Figure \ref{fig:predvis}.

\sectionc{Remarks and further related work}
\label{sec:rem}
\label{sec:alg_add}
\label{sec:obj_add}
\label{sec:addrelatedwork}
\label{sec:nerem}

This section\maybe{, which is not central to understand the main line of reasoning,} discusses additional aspects of the main results and further related work.

%\subsection{For Sections ... model, theory}

%
%
%Within game theory, also note the following related work:
%
%\todo{E.g., \socalled{correlated equilibria} that require a correlated signal were studied \citep{osborne1994course}.}
%
%\todo{}
%
%
%
%\parag{Relation to existing work on BNE as self-fulfilling prophecy.}
%Note that interpretations of the Nash equilibrium as self-fulfilling prophecy have been discussed, often informally, in \socalled{epistemic game theory} \citep{sep-epistemic-game,spohn1982how}. But they do not give a rigorous analysis of the specific conditions on information/utility/response structure for a concrete setting where the prophecy comes from an ``external'' agent. \maybe{, i.e., where an assistant is the origin of the users' beliefs.}
%\todo{The Nash equilibrium was informally interpreted as ``self-fulfilling prophecy'' \citep{spohn1982how} but without a rigorous study of the necessary conditions in an assistive setting.}
%
%
%Influential forecasts have been studied, also using fixed-point formulations but for election predictions, by \citet{simon1954bandwagon}. 
%

\parag{Why prediction accuracy / equilibrium selection as objective.}
Alternative to our approach in this paper, one could start from some (somehow legitimized) \emph{social welfare} \citep{nisan2007algorithmic} as a function of users' preferences, and design assistants that try to optimize it.
This would be somewhat more in line with the economic notion of optimizing efficiency.
%Here, we rather follow a heuristic approach of starting with the ``natural'' prediction accuracy objective, because we could show that it at least helps w.r.t.\ equilibrium selection, and also for the following reasons:
Here, we rather follow a heuristic approach of starting with the ``natural'' prediction accuracy objective, because it compares well to the benchmark of equilibrium selection (Theorem~\ref{thm:predne}), and for the following reasons:
First, prediction accuracy can be directly measured, while social welfare seems hard to infer/identify from the incomplete information contained in the behavioral data available in our setting.
Second, it is easy to interpret for users and leads to a form of ``incentive compatibility'' of users' assistant-best-response (see remark below). %Section~\ref{sec:obj_add}).
Third, we feel that in our coordinative setting, equilibrium outcomes can be quite efficient in terms of social welfare.
Generally, social welfare functions of course are hard to pick and impose in the first place. %, due to the multi-agent nature of the setting.
%
%\parag{The price of anarchy.}
Nonetheless, equilibrium outcomes can of course be significantly inefficient, which has extensively been studied under the name of \socalled{price of anarchy} \citep{roughgarden2005selfish,nisan2007algorithmic}.
But even in this regard, Theorem \ref{prop:argmin_ne} can be helpful in that it makes predictive assistant-based settings amenable to such studies.

\parag{Remarks on our model assumptions:} %Theorem \ref{prop:predne}:} 
%\begin{iitem}
%\item[] %\parag{Remark on assumptions.}
To justify our assumption of users ``blindly'' best-responding to the assistant's forecast (Definition \ref{def:m}) %-- which is our key restrictive assumption for the assistant-based system $M$ -- 
observe that it can be seen as consistent with (instrumental) rationality%
\footnote{In this work, we adopt the game-theoretic view of humans in social situations as ``selfish'' agents maximizing exogenously given individual utility functions. We feel this is appropriate for our simple setting of facility use. But overall, decision making in social systems has many more aspects of course.} 
in the following sense: if only considering the \emph{asymptotic} utility (once the assistant converged), then deviating from this behavior means deviating from a BNE, %Bayesian Nash equilibrium, % BNE, 
based on Theorem \ref{prop:argmin_ne}.\footnote{\citet{nisan2011best} studied rationality of \termtech{best-response dyn}.}
Furthermore, all users best-responding simultaneously can sometimes be a too strong assumption, but we feel that it is \emph{a} situation that can happen (more or less) at least \emph{sometimes}, and therefore is worth analyzing.
This being said, the assistant-best-responding assumption should be seen as a pragmatic \emph{first step} that can be refined in future work.
%\item[] 
Generally, Theorem \ref{thm:predne} shows that assistant-based systems can achieve coordination comparable to the \bg (additionally, it serves as a mechanism for equilibrium \emph{selection} if there are several) -- \emph{but at a significantly lower cost}, since the inference task is \emph{centrally} done by the assistant. (Obviously, it is only cheaper when inference comes at a cost -- otherwise raw data $\Acov$ could simply be provided to users directly.)

\parag{Further general related work:}
Let us mention that for the various versions of assistants we mentioned in Section \ref{sec:intro} that are publicly available \citep{googlepopulartimes,dbtravel,franceforecasts}%
, we could not find out what algorithms or theory they rely on.%
\footnote{Also note that some of them do not explicitly call the service an ``assistant'' or a ``forecast''.}
Research-wise, in \emph{mechanism design\footnote{The analogy between our assistant and a mechanism is that they are both ``institutions'' added to the set of agents to solve some collective decision making problem.}}, a related direction has been emerging that studies how to design the \emph{information structure} \citep{taneva2015information,bergemann2017information} instead of the allocation/payment structure.
Furthermore, \emph{data-driven approaches to mechanism design} have gained momentum \citep{balcan2016sample,dutting2019optimal,tang2017reinforcement,kearns2014mechanism}. But these lines of research differ from ours -- often additionally to what we already mentioned in Section~\ref{sec:intro} (bounded rationality of our users and limited power of our mechanism) as follows: either they assume that agents input their (true, if ``incentive compatible'') preferences explicitly (instead of behavioral data), or they neglect, to some extent, agent's actual preferences (which can be appropriate for revenue maximization of course).

\sectionc{Conclusions}
\label{sec:conc}

%\todo{we combined ... machine learning/linear dynamical systems/time series, game theory to quantify efficiency}

In this work, we studied when and how parts of the coordination process of users of shared resources can be ``outsourced'' to a central data-driven predictive assistant.
%First, we layed conceptual foundations that account for the ``social'' context the assistant operates in.
Our theoretical analysis showed that such assistants can help solve this multi-agent coordination problem in a game-theoretic sense, but non-trivial conditions have to be met: in terms of the information and preference structure of users, %their ``trust'' in the assistant, 
and stochasticity of their preferences in case only large-scale aggregated information is available to the assistant.
Based on this analysis, we proposed two machine learning coordination assistant algorithms on behavioral data. We used linear dynamical systems models to prove their optimality/convergence, accounting for the fact that there is a feedback loop from predictions to outcomes.
%We proved optimality/convergence for them.
And we conducted a large-scale interventional experiment in a real campus cafeteria that provided empirical hints for the validity of our main algorithm.

%A potential next step would be an approach starting from a social welfare functions instead of game-theoretic solutions concepts.
\maybe{As a potential next step, assistants whose owners impose own, say commercial, objectives (besides coordination) could be studied.}
Generally, the mentioned related work and our work indicate that there is a plethora of possible computational mechanisms for collective decision making, in terms of inputs (high-level information, behavioral data, and beyond) and influences (full control over the outcome, money incentives, pure information/predictions, and beyond), many of which may still be unexplored.
%
%selfish assistants
%mechanisms beyond informational

\parag{Acknowledgments.}
The authors thank Carl-Johann Simon-Gabriel, Jonathan Williams and Sebastian Stark for insightful discussions and engineering support.

\maybe{
\paragraph{Potential next steps.}
%We performed inference in a state-space model and harnessed prior knowledge to direct the search in part of the \nonagg. But 
More sophisticated models of users' utilities may help to generalize more from past data beyond our state-space/congestion game model -- but extensive search and exploration still seems unavoidable. Our assumption of users always perfectly best responding to the assistant -- as key source of information about their utilities -- may be relaxed.
Generally, more sophisticated user inputs, assistant outputs, intra-day assistant-user interactions and, e.g. ticket systems, could further help coordinate users.
\todo{selfish assistants- see previous versions}
}

\appendix

\onecolumn

%\title{Title 2}

\section*{Appendix}

%\includepdf[pages=-]{"cafco_supp"}

% FOR SUPP:

\newcommand{\asssep}{(``assistant-separability'')\xspace}

\renewcommand{\extref}[1]{\ref{#1} of the main paper}  % actually, shouldnt we drop the ``of the main paper'' for tis extended version?
\renewcommand{\extrefs}[1]{\ref{#1}}

\renewcommand{\todo}[1]{}

%% THEOREM/ALGO NUMBERING AND CITING IN SUPPLEMENT:

\newtheorem{SuppProposition}{Proposition}
\renewcommand*{\theSuppProposition}{S\arabic{SuppProposition}}

\newtheorem{SuppLemma}{Lemma}
\renewcommand*{\theSuppLemma}{S\arabic{SuppLemma}}

\newtheorem{CiteTheorem}{Theorem}
\newtheorem{CiteCorollary}{Corollary}
\newenvironment{CTheorem}[1]{\renewcommand*{\theCiteTheorem}{#1 of the paper}\begin{CiteTheorem}}{\end{CiteTheorem}}
\newenvironment{CCorollary}[1]{\renewcommand*{\theCiteCorollary}{#1 of the paper}\begin{CiteCorollary}}{\end{CiteCorollary}}

\newcommand{\suppciteref}{\extref}

\renewcommand{\thealgocf}{S\arabic{algocf}}

%\todo{update title}

\section{Background on game theory}
\label{sec:gameback}

Game theory \citep{osborne1994course,shoham2008multiagent} models the \emph{interaction between strategic\footnote{``Strategic'' means that they have goals/objectives/preferences/interests/utilities and take the best possible means to achieve them, accounting for the (multi-agent) context; a common alternative expressions are ``(instrumentally) rational'', ``self-minded'' or ``self-interested''.} agents}, that is, settings with several such agents and where the utility of any one of them is influenced by the actions of one or several of the others.
Since each agent's utility depends on the other agents' actions, each agent has to reason about how the others act when deciding on its own action.

The modeling in game theory is usually split into two parts:
First, a \emph{game} formalizes, in a sense, the decision making \emph{problem}, that the agents (also called ``players'') are facing.\footnote{The name ``game'' likely comes from games of parlor being a special case of such games, but also based on them being a \emph{metaphor} for general multi-agent situations, a metaphor that helps for the formal abstraction.}
Note that, in a sense, there are two problems, but often they are treated simultaneously: the \emph{descriptive} problem of predicting which strategies the players will chose when facing the game, and the \emph{prescriptive} problem that each player faces -- choosing the strategy that best serves her objective.
Second, game theory considers \defi{solution concepts} \citep{shoham2008multiagent} that formalize how the agents will (in the descriptive interpretation) or should (in the prescriptive interpretation) approach this problem (game).

Now essentially, a game, as used by game theory, represents each agent by (1) a \emph{utility function}, that models her interest/preferences/goals, and (2) a set of possible \emph{actions} that she can take and has full control over.\footnote{Alternative formulations use preference relations instead of utility functions.}
In the simplest case, called a \emph{(complete-information) normal-form game} \citep{shoham2008multiagent}, this is essentially already the full model.
In this case it is assumed that all utility functions are \emph{fixed} and each agent \emph{knows} the utility function of all other agents.

\begin{table}[h!]
\begin{center}
	\setlength{\extrarowheight}{2pt}
	\begin{tabular}{cc|c|c|}
		& \multicolumn{1}{c}{} & \multicolumn{2}{c}{Player 2}\\
		& \multicolumn{1}{c}{} & \multicolumn{1}{c}{12noon}  & \multicolumn{1}{c}{1pm} \\\cline{3-4}
		\multirow{2}*{Player 1}  & 12noon & $(0,0)$ & $(2,1)$ \\\cline{3-4}
		& 1pm & $(1,2)$ & $(0,0)$ \\\cline{3-4}
	\end{tabular}
\end{center}
\caption{Example of a two-player, two-action, complete-information coordination game. Note that our treatment in the main paper is for the more general class of \emph{incomplete-information} Bayesian games to account for the significant uncertainty about (other) agents' preferences that usually occurs in such settings.}
\label{table:coord}
\end{table}

Let us give an example of such a complete-information normal-form game with two players and two actions each. We consider a simple coordination problem where the players can chose between going at 12noon or at 1pm to a cafeteria\footnote{This is a toy version of the setting of our cafeteria experiment, Section \extref{sec:exp}.}, and aim to avoid each others, say to avoid queuing. Additionally, assume that going early is favored by both. Specifically, let the game be given by the \defi{payoff matrix} in Table \ref{table:coord} (payoff matrix is just another term for ``utility matrix''). This representation has to be read as follows: 
player $l$'s utility, in case player 1 chooses action $i$ and player 2 chooses action $j$, is given by the $l$-th entry of the tuple at column $i$, row $j$ of the matrix. For instance, if player 1 goes at 12noon and player 2 goes at 1pm, then player 1 has utility 2 and player one has utility 1.

Given such a game, each player can chose an action, and we can also consider jointly the actions of the players. We formalize such joint actions by \defi{action profiles}, i.e., tuples of actions, one for each player.
For any such action profile, we can ask if it is a \emph{solution} to the game, according to some solution concept, as mentioned above.
%call such a joint action once considers all possible actions that all players can take in this game
The most common solution concept for complete-information normal-form games is the \defi{(pure) Nash equilibrium}. In the two-player, two-action case, it is defined as any action profile $(i, j)$, such that no player can improve her utility by unilaterally deviating to an action $i'$, i.e., every player chooses her optimal action (``best-responds'') given the other action in the tuple is fixed. For instance, in Table \ref{table:coord}, the action profile $(\text{12noon, 1pm})$ is a (pure) Nash equilibrium.

Clearly, ``complete-information'' is a strong assumption and therefore a generalization of these complete-information games has been proposed \citep{harsanyi1967games}, to model the case where the utility functions (more specifically: the precise influence of the joint action on the utilities of the agents) are a priori unknown to the agents (so in a sense: \emph{the game is a priori unknown}).
Intuitively, this absence of knowledge comes from agents' preferences as well as relevant external events not being determined/known a priori.
Instead, it is assumed that each agent, before choosing its action, observes a (private) \emph{signal} and then uses a Bayesian \emph{prior distribution} over this signal and the other relevant variables for its inference, according to Bayes rule.
%
%But what 
%
%
%
%It models settings with several agents, where the utility of agent $i$ 
%The basic idea is to model each agent by an own (exogenously given) goal w.r.t. some outcome and has full control over some action that influences the outcome.
%So since the outcome depends on all agent's actions, for one agent to decide on its actions To get as close as possible to their goal, each agent has to take into account how the other.... raitnality ....
%
%
%Beyond the most simple scenarios, incomplete-information

Let us give a formal definition, based on \citep{osborne1994course} but adapted for our purposes (Remark \ref{rem:bg}).
\begin{Definition}
A \defi{Bayesian game} consists of a set $\users$ of players, a state of the world $X$, and for each player $i \in \users$:
\begin{cenum}
	\item a signal $\Theta_i$ out of a set of possible signals $\dom_{\Theta_i}$,
	\item an action $\Uact_i$ out of a set of possible actions $\dom_{\Uact_i}$,
	\item a utility function $\bar{U}: \dom_{(X, \Theta_i, \Uact)}  \to \R$,
	%defined by $(\usig_i, c) \mapsto U_i(q_i, c)$ \todo{or rather: $\dom_X \times \mathcal{C}  \to \R$}
\end{cenum}
and a common (``objective'') prior distribution $P(X, \Theta_i, i \in \users)$.
\end{Definition}

\begin{Remark}
\label{rem:bg}
We adapted the definition in \citep{osborne1994course} to our purposes in three ways:
First, we dropped the assumption of finite \emph{cardinality} of the sets (somewhat similar to \citep{kim1997existence}).
Second, we formulate it more in the spirit of \emph{random variables} (we event treat actions as random variables, in the sense that choosing a specific action corresponds to \emph{intervening} on the action variable, in the sense of causal models \citep{Pearl2000}), instead of just specifying their \emph{ranges}.
Third, our ``state of the world'' $X$ is a random variable and it does not have to \emph{determine} the value of the other variables. In contrast, in \citep{osborne1994course} the ``state'' is the \emph{outcome}, in the probability theoretic sense (as element of the sample space, typically denoted by $\omega \in \Omega$) \citep{klenke2013probability}, instead of a random variable. Accordingly, in their definition, the utility function does not have to depend on the signal, and the prior is already specified by a distribution over the state.
But a Bayesian game according to our definition can be mapped to one according to the definition in \citep{osborne1994course} in the obvious way (essentially replacing our $X$ by the state in the sense of the outcome), and vice versa.
%
%it towards our formulation in Section \extref{sec:pregame}: we base it more on random variables.
%In particular, the ``state of the world'' usually is the outcome, in the probability theoretic sense (as element of the outcome space) \citep{klenke2013probability}, instead of a random variable.
%Here, we let the ``state of the world'' itself be a random variable and furthermore .... noisy ... but can be cast in the calssical definition by considering the underlying outcome space ....
\end{Remark}

Now, while there is no single one established solution concept (mentioned above) for a Bayesian game, the most common one is the \emph{Bayesian Nash equilibrium} as we define it in Section \extref{sec:pre}.
Note that in the case of the Bayesian game, potential solutions are given in the form of \emph{strategy profiles} (as we also introduce it in Section \extref{sec:pre}), i.e., tuples $(s_i)_{i \in \users}$, where each $s_i$ is a \emph{strategy} -- a mapping from the set $\dom_{\Theta}$ of possible signals of player $i$, to the set of possible actions $\dom_{\Uact_i}$ of player $i$. This generalizes the notion of an action profile introduced above, accounting for the fact that the player's behavior is only fully specified once we determine her action \emph{for each possible observed signal}.

%So this game formalizes, in a sense, the \emph{problem}, that the agents/players are facing. 
%
%Now so far, there is no single one established \emph{``solution''} to this problem. Rather, there are several so-called \defi{solution concepts} \citep{shoham2008multiagent}.
%The most important one is the Bayesian Nash equilibrium as we define it in Section \extref{sec:pre}.

\section{Additional related work and comments}
\label{sec:supprelated}
\label{supp:sec:supprelated}

\parag{Further general related work in game theory:}
Within game theory, note that \socalled{correlated equilibria} were studied \citep{osborne1994course} that require a correlated (i.e., central) signal.
Inference of preferences from behavioral data has been studied \citep{ling2018game}, but they do not feed the results back into the multi-agent system.
Interpretations of the Nash equilibrium as self-fulfilling prophecy have been discussed, often informally, in \socalled{epistemic game theory} \citep{sep-epistemic-game,spohn1982how}. But they do not give a rigorous analysis of the specific conditions on information/utility/response structure for a concrete setting where the prophecy comes from an ``external'' agent.
Influential forecasts have been studied, also using fixed-point formulations but for election predictions, by \citet{simon1954bandwagon}. 
Bayesian games with discrete actions but continuous states and signals have been studied by \citet{hellman2017bayesian}.

\parag{Further related work within smart cities research:}
Besides the work already discussed in Sections \extref{sec:intro} and \extref{sec:existence}, also the following work in the area of smart cities and control is on congestion/coordination in shared facilities (often with some form of central assistant or signal):
\citep{wirth2019nonhomogeneous} consider agents that share a constraint resource and receive a central capacity signal. 
%There is an overall objective function given by the sum of the individual agents' utilities. 
For the case that the agents behave according to a (randomized) so-called additive-increase multiplicative decrease (AIMD) algorithm (alongside additional assumptions), their theoretical analysis shows convergence against the optimum under an overall objective function given by the sum of the individual agents' utilities.
\citep{haeusler2014closed} also discuss the problem of flapping for the case of coordinated (balanced) routing of cars in road networks, and present a randomized approach to it.
\citep{schlote2014closed} consider users of bike sharing stations and their decision making in terms of which station to go to for renting/returning a bike. They present an approach that combines providing users with occupancy data and a random assignment based on it, for the sake of balancing.
The main differences between these works and ours are that (1) we focus on game-theoretic Bayesian Nash equilibrium solutions to the coordination/congestion problem (and the conditions under which it exists, in Theorem \extref{thm:ex}), and (2) our results focus more on the conditions under which the assistant can solve certain inference/prediction tasks (like assistant-separability in Theorem \extref{prop:predne}).
%
%
%with objective .... They  and behave according to a so-called additive-increase multiplicative decrease (AIMD) algorithm
%central has a form of a capacity signal in case the overall resource constraint has been exceeded.
%agents have a memory about their past actions
%optimize the (utiliterian, i.e., sum of individual agent's utilities) welfare .
%agents' behavior has the form of a additive-increase multiplicative decrease(AIMD) algorithm, for which they provide a theoretical convergence analysis.

%The parameter gamma is broadcast

%Closed-loop flow regulation with balanced routing : actually i dont think this has any kind of assistant
%
%We wish to solve this problemusing a stochastic approach which leads to heterogeneousroutes for vehicles. We suggest alternative routes to driversin a manner that balances any kind of feature (e.g. pollutionor congestion) along these multiple routes

%On closed-loop bicycle availability prediction: 

%\citep{marevcek2015signalling,marevcek2016r}

%Recovering Markov Models from Closed-Loop Data

\parag{Further related work for Theorem \ref{thm:ex}:}
%
%
%\todo{settings with many individually negligible agents have been studied under the name of \socalled{nonatomic games} \citep{schmeidler1973equilibrium} (Section~\ref{sec:nonat})}
%
%We already mentioned that 
Our setting relates to \socalled{nonatomic games} \citep{schmeidler1973equilibrium} studied in game theory.
However, we are only aware of two lines of work that study the \emph{incomplete-information} case in the setting of a nonatomic continuum of types: \citet{sabourian1990anonymous}, but they do not focus on Bayesian Nash equilibrium existence in the stage game itself. More closely related is \citep{kim1997existence}: they study existence of a Bayesian Nash equilibrium in incomplete-information nonatomic games in quite general terms. But they do not cover our case where the ``state of the world'' ($X$) has an uncountable range. Furthermore, existence of a self-fulfilling aggregate prophecy is not entailed by their results (using our equivalence in Theorem \ref{prop:argmin_ne}), due to potential non-strictness of their Bayesian Nash equilibrium. In this sense, our Corollary \ref{cor:ex} may also be of value for the game-theoretic side. %, although we have to mention again that for now we assumed that the types have no private signals $(W, \Usig$ constant).
\maybe{Note that our $\gnonat$, similar as in \citep{kim1997existence}, is very general: it allows $\users$ to be a continuum}
Let us also mention \citet{rath1992direct}, who, in one part of their proof of their Theorem 1, also reduces the Nash equilibrium existence problem to existence of a form of self-fulfilling prophecy on the aggregate level (without considering it as such). But they restrict to the complete-information case.
From the smart cities research side, we already mentioned \citet{marevcek2016r,marevcek2015signalling} above.
%\todo{put this here??: A main difference is that they consider a (single) cost as a function of the joint user action, given a priori, 
%	%that they optimize 
%	%(see Remark 2 in \citep{marevcek2015signalling}), 
%	while our objective function (prediction accuracy) allows to take into account (individual) user preferences not known to the assistant a priori (Section \ref{sec:pred-sol}).}
They essentially propose two solutions: either sending different signals to different agents (which we, in a different sense, also do in the \nonagg, Setting \ref{set:nonagg}) or the population of agents has to be heterogeneous, which relates to our assumption of random types. But their heterogeneity is rather in the behavior, not in the form of individually differing utility functions, as in our case.

\paragraph{Remarks on Algorithm \ref{alg:partpredsk} and Proposition \ref{prop:consistresp}.}
Algorithm \ref{alg:partpredsk} is mainly a proof-of-concept to illustrate several points: 
%\begin{citem}
%\item 
An assistant can handle simultaneous/imperfectly orchestrated user responses.
%\item 
And while most assistant-free dynamics of ``learning in games'' (mentioned in Section~\ref{sec:intro}), such as \socalled{best-response dynamics} or \socalled{ficticious play} \citep{shoham2008multiagent}, only converge in special cases, an assistant can help to overcome cycling and oscillations, use exploration, and make the system always converge (with high probability; in the finite setting under consideration). Furthermore, as the extension of Proposition~\ref{prop:consistresp} in Section~\extref{sec:pr_partpred} will make more clear, the assistant can use prior knowledge of the utility functions, e.g. that they form a congestion game \citep{nisan2007algorithmic}, to speed up convergence. 
Generally, the area of ``learning in games'' is related to ours in that they also consider the case where agents do not know the preferences of others.
Note that in ``learning in games'', often agents first build a model/belief about the other agents behavior and then optimize their decision under it -- referred to as \socalled{model-based decision making} \citep{shoham2008multiagent}.
Our setting can be seen as a version of such model-based decision making, where the data-driven modeling task is \emph{``outsourced''} to the central assistant.
%\end{citem}

\paragraph{Additional remarks:} % on Theorem \ref{prop:predne}:}
Besides the price of anarchy  \citep{roughgarden2005selfish,nisan2007algorithmic}, further limitations can occur when extending the setting: for instance it could happen that the assistant would figure out that making users not use the assistant (e.g., by deliberately providing poor forecasts for some time) could yield more predictable outcomes than other strategies (although, based on our results, never as good ones as Nash equilibria) -- possibly yielding completely undesired assistant behavior.
Note that, instead of making a statement about the reasonableness of $\lpred$ \emph{in isolation}, which is impossible, rather here we analyzed the \emph{combination} $(\lpred, \sigma)$, for a certain joint user behavior $\sigma = (\sigma_i)_{i \in \users}$.
Also note that so far we only consider classical (Bayesian) Nash equilibria, but the results may be extendable to harness the assistant to also announce \socalled{correlated equilibria} \citep{osborne1994course}.

\parag{Additional related work for algorithmic aspects:} %Related work for Algorithm \ref{alg:expodamp}:} 
%To further embed the results of this section into the context of previous work, let us mention two papers that apply time series forecasting, even latent-state models, to congestion or general multi-agent modeling:
Regarding Algorithm \extref{alg:expodamp},
\citet{zhang2013wait} apply various machine learning methods to wait time prediction, similar as we do, but not considering influential predictions or non-stationarities. \citet{smyrnakis2010dynamic} model multi-agent dynamics using latent-state models, but from the view of one of the players and in a non-aggregate setting.
%A practical problem with only announcing aggregated forecasts/recommendations can be that, without any further information, coordination is not possible when the users don't know where there place in the forecast is \todo{see privacy section in previous version of this paper}.
\maybe{
Abstracting away from the specific congestion problem, and only addressing the predictive assistant aspect while ignoring that the environment consists of multiple utility-directed agents, our problem can on an abstract level (if we fix user behavior) be seen as a specific instance of the \socalled{multi-armed bandit/reinforcement learning/control} problem \citep{sutton1998reinforcement} \todo{drop this/move below: if the assistant has a significant influence, or, if the influence is negligible, as a classical \socalled{time series forecasting} problem \citep{zhang2013wait} often addressed via \socalled{exponential smoothing and Kalman filtering} methods \citep{hyndman2008forecasting} (Section \ref{sec:expodamp})}. 
%\todo{We will treat it as a control problem in Section ..., expanding exponential smoothing methods from classical non-influential time-series analysis to our setting.}
% \item 
}
%
%
%
%

%\parag{Extension of class of solution concepts considered.}

\section{Some additional notation and details on measurability assumptions etc.\ in Section \extref{sec:pre}}
\label{supp:sec:addnas}

\subsection{Additional notation}
\label{supp:sec:addn}

For the following sections, let us introduce explicit names for certain mechanisms that are part of our basic model of Section \extref{sec:pre}, for which we have not given explicit names there:
\begin{citem}
\item We use $\sigma_i$ to denote \emph{user $i$'s behavior in the assistant-based system}, i.e., $i$'s policy that generates her action $B_i$ from the input ($W_i, A$) in the assistant-based system, obeying Eq.~\extref{eqn:mbr}, i.e., the ``best response'' to forecast $A$ (it is uniquely defined by Eq.~\ref{eqn:mbr} together with the tie-breaking rule we give in Section \ref{supp:sec:addass}). And we let $\sigma = (\sigma_i)_{i \in \users}$.
\item We denote by $\bar{U}_i$ the \emph{generic utility function}, the mechanism that generates user $i$'s utility $U_i$, i.e., $U_i = \bar{U}_i(X, B)$. (Recall that in the main paper we generally assume inference-assistability and thus solely use the restricted form of the mechanism in the form of the function $\tilde{U}$. In particular, $\utb_i(X, (\Uact_i, \Uact_{-i})) = \uta_i(\Usig_i, \Uact_i, h_i(Y))$ under the assumption of inference-assistability.) 
\end{citem}

Generally, keep in mind that, for random variables $Z_1, Z_2$, $P(Z_1|z_2)$ is shorthand for the (regular) conditional distribution $P(Z_1|Z_2=z_2)$, for $z_2$ a value of $Z_2$.

\subsection{Details on measurability assumptions etc.\ and soundness of definitions in Section \extref{sec:pre}}
\label{supp:sec:addass}

We were very brief in Section \extref{sec:pre} regarding measurability assumptions etc. Here we explicate the assumptions we meant there in detail.\footnote{These assumptions are more general than, but still somewhat tailored to, the two main settings we consider (small-scale and large-scale). For our purposes this is enough. We do believe that Theorem \extref{prop:argmin_ne} holds more generally; but the measurability side of things becomes rather involved.}

We will have somewhat different assumptions regarding ranges, $\sigma$-algebras, measurability etc.\ of variables for two different cases: (1) finitely many users and (2) infinitely many.
In the proofs, we will treat both cases simultaneously where we can, and treat them separately where we have to.
%\footnote{To see that the various assumptions we make, with some of them being somewhat indirect and involved, are consistent, see the more specific Settings \extref{set:nonagg} and \extref{set:nonagg} for examples of models that satisfy these assumptions.} % .... there, let us be more specific here.

%
%	\item the ranges $\dom_Z$ of the variables $Z \in \{X, V, W_i, Y, h_i(Y), U_i : i \in \users \}$ are equipped with a respective $\sigma$-algebra $\alg_Z$ (and cartesian products of ranges are equipped with the respective product $\sigma$-algebras), in particular $\dom_{U_i}, i \in \users$ is $\R$, equipped with the Borel sets 
%\item $\dom_{B_i}, i \in \users$ is $K$   and they are equipped with and discrete $\sigma$-algebra, respectively,
\subsubsection{Case of finitely many users}
\label{supp:sec:addassfin}

In case the set $\users$ (the set of users) is finite, we assume that (these assumptions are, in a sense, a generalization of Setting \extref{set:nonagg}):
\begin{itemize}
\item there is some underlying probability space $(\Omega, \mathcal{F}, P)$,
\item regarding ranges $\dom_Z$ of the variables $Z \in \{X, V, W_i: i \in \users \}$, we assume that they are either \emph{all} finite or \emph{all} continuous (meaning compact subsets of a Euclidean space), %\todo{what about $Y, h_i(Y) $ it kind of doesnt make sense for it to be continuous, even if the rest is continuous}
\item the ranges $\dom_Z$ of all the variables $Z \in \{X, V, W_i, B_i, Y, h_i(Y), U_i : i \in \users \}$ are equipped with a respective $\sigma$-algebra denoted by $\alg_Z$ (and Cartesian products of ranges are equipped with the respective product $\sigma$-algebras); in particular, $\dom_{B_i}, i \in \users$ is $K$ and is equipped discrete $\sigma$-algebra, and $\dom_{U_i}, i \in \users$ is $\R$ and is equipped with the Borel sets, and the others are either, in the case of discrete ranges, equipped with the discrete $\sigma$-algebras or, in the case of continuous ranges, with the Euclidean topology and the Borel sets,
\item $X, V, W$ are random variables on $(\Omega, \mathcal{F})$,
%	as well as tuples of variables $Z \in \{(X, B), W, B, Y, U\}$ are equipped with a respective $\sigma$-algebras $\alg_Z$
%\item the ranges $\dom_Z$ of all the mentioned (random) variables $Z \in \{X, V, W, B, Y, U\}$ are equipped with a respective $\sigma$-algebra $\alg_Z$
\item the variable $A$ as range $\dom_A$ has $\meas = \meas^{\dom_Y}$, the set of Borel measures\footnote{In the case of finitely many users, $Y$ is assumed to be finite and then the Borel measures are just the usual simplex in the Euclidean space. However, with the formulation in terms of Borel measures we can simultaneously cover the case of infinitely many users, where $Y$ is potentially continuous.} \cite{klenke2013probability} on $\dom_Y$, and it is equipped with $\tau_A$ the weak topology \citep[Remark 13.14(ii)]{klenke2013probability} and $\alg_A$ the Borel sigma algebra induced by the weak topology, % has the \todo{what range? all measures on rangeY? need to restrict rangeY for this?}
\item the range of the public outcome variable $Y$ is finite,
\item the function $\bar{Y}$ is measurable w.r.t.\ the respective (product) $\sigma$-algebras,
\item the functions $y \mapsto h_i(y), i \in \users$, and $(i, w_0, y) \mapsto \tilde{U}_i(w_0, k, h_i(y))$ are continuous and bounded in all arguments with continuous ranges\footnote{Again this is a formulation to cover the case of finitely and infinitely many users simultaneously.}, for all $k \in \slots$, %\todo{are there any issues w.r.t. $h_i$?}

\item the assistant policy $\pi$ is measurable w.r.t.\ $\alg_V$ and $\alg_A$, %\todo{DROP because should be derived: $\sigma_i$ is measurable w.r.t.\ the product $\sigma$-algebra $\alg_{W_i} \otimes \alg_A$ and $\alg_{B_i}$, $i \in \users$,}
%\end{itemize}
\item we assume that users \emph{break ties} by preferring the slot $k \in K \subset \N$  with the lower number (i.e., take the natural ordering of $\N$ as the tie breaking preference ordering whenever two slots yield the same (expected) utility $U_i$ for them), this together with Eq.~\ref{eqn:mbr} uniquely determines the user behavior $\sigma$,

%For a strategy profile $s$ (Section \extref{sec:pregame}) for the \bg $G$ (Definition \extref{def:g}) we assume regarding measurability that:
%\begin{itemize}
\item for a strategy profile $s$ (and in particular a BNE $s$), we assume that $s$ is measurable w.r.t.\ the product $\sigma$-algebra $\alg_W \otimes \alg_V$ to $\alg_A$.
%\item for fixed $w, v$, the mapping $i \mapsto  s_i(w_i, v)$ is measurable w.r.t.\ $\usersalg$ and $2^\slots$.
\end{itemize}

\subsubsection{Case of infinitely many users}

In case the set $\users$ is infinite (interpreted as \emph{types} of users in this case), we make the following assumptions  (these assumptions are, in a sense, a generalization of Setting \extref{set:agg}), as modifications of those for the case of finite $\users$ stated above (Section \ref{supp:sec:addassfin}):
\begin{itemize}
\item we now assume $V, W$ to be constant \reasonwhy{meaning constant random variable; reason why i dont want to drop them being random variables is that for instance then the assistant-seperability satetement does not make sense anymore for this case} (corresponding to Setting \extref{set:agg}), with $W_i$ also being constant in $i$, %\todo{rather: variables but not random variables? or maybe this is rather what has to be stated w..r.t A, B?}
\item we let $I = [0,1]$ (in this case interpreted as types of users), and equip with the Borel sets as $\sigma$-algebra $\usersalg$, % \todo{OLD too general: the set $I$ (in this case interpreted as types of users) is equipped with a $\sigma$-algebra $\usersalg$, \todo{drop}}
\item we consider the variables $A, B$ not to be random variables but only variables, in particular, only have a range but not a $\sigma$-algebra, % (in the sense) to not depend on $V, W$ and instead be constant \todo{not a random variable, not measurable, but rather a variable which depends on sigma ...} (still depending on the choice of $\pi$ and $\sigma$ though\todo{, where $\pi$ has no argument anymore and $\sigma$ only has argument $A$}),
\item $\pi$ takes the values of $V$ as argument, but we do not assume measurability in this argument anymore (alternatively one can consider $\pi$ not to take any argument),
\item $\sigma_i, i \in \users$ takes the values of $W_i, A$ as arguments, but we do not consider it as a potentially measurable function anymore (alternatively one can consider $\sigma_i$ to only have argument $A$),
%	\item \todo{$\pi$ takes the value that $V$ is fixed to as argument, but we do not assume measurability in this argument anymore,}
\item for the variable $B$, we introduce an explicit range $\dom_B$, which is a subset of $\times_{i \in \users} \dom_{B_i}$, and can be a \emph{proper} subset (in particular, this constraints the range of $\sigma=(\sigma_i)_{i \in \users}$),
%	\item \todo{drop? since we prove it: for fixed $w, a$, the mapping $i \mapsto  \sigma_i(w_i, a)$ is measurable w.r.t.\ $\usersalg$ and $2^\slots$,}
\item the range of the public outcome variable $Y$ can be continuous,
\item we consider $\bar{Y}$ to be a function with domain $\dom_X \times \dom_B$, for which we do not require measurability in all arguments but only that, for fixed $b$, $x \mapsto \bar{Y}(x, b)$ is measurable,
% we consider $\tilde{U}_i$ to be a function with domain $\dom_{B_i} \times h_i(\dom_{Y})$ (not $\dom_{W_i} \times \dom_{B_i} \times h_i(\dom_{Y})$), for which we do not require measurability in all arguments but only that, for fixed $b_i, y$, the mapping $i \mapsto \tilde{U}_i(b_i, h_i(y))$ is measurable w.r.t.\ the respective $\sigma$-algebras, \todo{$(i, k, y) \mapsto \tilde{U}_i(k, h_i(y))$  is measurable w.r.t.\ the respective (product) $\sigma$-algebras. BECAUSE IN SECOND PART OF THM 1 WE NEED IT}
\item for a strategy profile (and in particular a BNE) $s$, (which we still consider to take the values of $V, W$ as input, but not to be measurable in them anymore) the mapping $i \mapsto  s_i$ is measurable w.r.t.\ $\usersalg$ and $2^\slots$.
%	\item \todo{old: in case $V, W$ are constant, we consider $A, B$ to not depend on them and instead be constant random variables (parametrized by $\pi$ and $\sigma$).}
\end{itemize}

%
%We assume the users ``joint policy'', specified by Eq.~... which we gave the name $\sigma$ in Section \ref{supp:sec:addn}, to be measurable in two ways:
%for fixed $w, a$, $i \mapsto  \sigma_i(w_i, a)$ ....,
%and $(v, w) \mapsto \sigma(v, w)$ , i.e., from ...... to $\dom_B$ with $\sigma$-algebra  $\alg_B$.
%
%
%
%
%\todo{OLD:}
%Recall that by the .... $M$ we refer to the set of objects and assumptions in Setting ... together with those stated in Definition ...
%Since we were very succinct regarding measurability there, let us be more specific here.
%We assume the set of users / types of users $I$ to be equipped with a $\sigma$-algebra $\usersalg$.
%We assume there is some underlying probability space $(\Gamma, \mathcal{F}, P)$.
%We assume the ranges $\dom_Z$ of all the mentioned (random) variables $Z \in \{X, V, W, A, B, Y, U\}$ to be equipped with a respective $\sigma$-algebra ....
%We assume the users ``joint policy'', specified by Eq.~... which we gave the name $\sigma$ in Section \ref{supp:sec:addn}, to be measurable in two ways:\footnote{To see that the various assumptions we make, with some of them being somewhat indirect and involved, are consistent, see the more specific Settings \extref{set:nonagg} and \extref{set:nonagg} for examples of models that satisfy these assumptions.}
%for fixed $w, a$, $i \mapsto  \sigma_i(w_i, a)$ ....,
%and $(v, w) \mapsto \sigma(v, w)$ , i.e., from ...... to $\dom_B$ with $\sigma$-algebra  $\alg_B$.

\subsection{Well-definedness in terms of measurability etc.}

Above (in the case of finitely many users), we assumed measurability of all of the relevant ``primitive'' mappings that occur in Section \extref{sec:pre}. However, for $\sigma$ (defined in Section \ref{supp:sec:addn} as shorthand for Eq.~\extref{eqn:mbr} and tie-breaking), the user behavior in the assistant-based system (Definition \extref{def:m}), we have to \emph{prove} measurability, because it is not a ``primitive'' mapping, but rather defined based on other mappings.
This, together with measurability of $\pi$, also establishes the soundness of the definition of the corresponding strategy profile (Eq.~\extref{eqn:s}); we will use this in the proof of Theorem \extref{thm:predne}.

%Let us first prove the following lemma, which covers most of the measurability statements that have to be proved for the above theorem.
%Note that this can also be seen as a stand-alone result which shows that the assistant-based system canonically associated to the general setting (or any restriction of it) is \emph{well-defined}, in the sense that in deed the mapping $\sigma$ is measurable (for the other mappings we assumed measurability explicitly).

\begin{SuppLemma}
\label{supp:lem:m}
The following mappings are measurable w.r.t. the respective (product) $\sigma$-algebras (in the general setting and the canonically associated assistant-based system):
in the case of finitely many users,
for fixed $i$,
\begin{align}
(w_i, a) &\mapsto \sigma_i(w_i, a) \label{supp:eqn:fin1}
\end{align}
and, in the case of infinitely many users, for fixed $w, a$,
\begin{align}
i &\mapsto \sigma_i(w_i, a) . \label{supp:eqn:fin2}
\end{align}
%	(The former is relevant for the case of finitely many, the latter for the case of infinitely many users.)
\end{SuppLemma}

\begin{proof}[Proof of Lemma \ref{supp:lem:m}]
Note that, in the case of discrete ranges, everything is measurable, so let us focus on the case of continuous ranges.

We explicate the proof for the case of two slots, i.e., $\slots = \{0,1\}$. The case of more slots works similarly.

\parag{Measurability of the mapping in Eq.~\ref{supp:eqn:fin1}\footnote{We formulate this proof more generally than we would have to: we formulate it for $A$ being a general Borel measure, although we only consider the case of finitely many users where $A$ is actually always a measure over a finite set.}:}

Recall that $\dom_A$ is $\meas = \meas^{\dom_Y}$, the set of Borel measures on $\dom_Y$, and we equip it with $\tau_A$ the weak topology (with the bounded continuous functions as ``test functions'' \citep[Remark 13.14(ii)]{klenke2013probability}) and $\alg_A$ the Borel sigma algebra induced by the weak topology.
%Also recall that the ranges $\dom_Y, \dom_W$ are compact subsets of the Euclidean space.

Let $H$ be the Heaviside step function.
Let $i$ be arbitrary but fixed.
Let 
$$Q(w,y) := \uta_i(\usig, 1, h_i(y)) - \uta_i(\usig, 0, h_i(y)), w \in \dom_{W_i}, y \in \dom_Y$$
(here $w$ is a value of $W_i$, not of $W$, for ease of notation).

We have to show that (keep in mind that values $a$ of $A$ are elements of $\meas$, i.e., measures)
\begin{align}
(w, a) \mapsto \sigma_i (w, a) = H( \int Q_i(w,y) da(y) )
\end{align}
is measurable, as a mapping from the product $\sigma$-algebra $\alg_{W_i} \otimes \alg_A$ to the discrete $\sigma$-algebra on $\slots$.
For this it is enough show that
\begin{align}
(w, a) \mapsto \int Q(w,y) da(y)
\end{align}
is continuous w.r.t.\ the respective (product) topologies, because continuity implies measurability and furthermore, the Heaviside function $H$ is measurable, and so their concatenation is \citep{klenke2013probability}.

So let $(w_m, a_m) \xrightarrow{m \to \infty} (w, a)$ w.r.t.\ the product topology of $\tau_{W_i}$ and $\tau_A$, which implies convergence $w_m \xrightarrow{m \to \infty} w$ and $a_m \xrightarrow{m \to \infty} a$ w.r.t.\ the individual topologies as well.
Let $R_m(y) := Q(w_m, y)$ and $R(y) := Q(w, y)$ for all $y$.

Then (using an argument similar to \citep[Proposition 3.13]{brezis2010functional})
\begin{align}
&\left| \int Q(w_m,y) da_m(y) - \int Q(w,y) da(y) \right| \\
&= \left| \int R_m d a_m - \int R d a \right| \\
&\leq \left| \int ( R_m - R ) d a_m  \right| + \left| \int R d a_m - \int R d a \right| .
\end{align}
The second term converges to zero by definition of the weak convergence (because $R$ is continuous and bounded, thus qualifies as a ``test function'').
The first term can be bounded by 
\begin{align}
\int \| R_m - R \|_\infty d a_m \leq \| R_m - R \|_\infty \int 1 d a_m = \| R_m - R \|_\infty \xrightarrow{m \to \infty} 0,
\end{align}
since $\uta_i$ and thus $Q$ is uniformly continuous (based on continuity and $\dom_{W_i}$ and any other range to being compact) and thus $R_m$ converges uniformly to $R$.  \todo{or need to be more explicit? compactness of ranges? continuity of $U(i, w, y)$ wrt product topology implies continuity also when we fix an arg i guess right}. % we have $\| R_m - R \|_\infty \xrightarrow{m \to \infty} 0$.

%	
%	
%	
%	First, observe that for all $i$,
%	\begin{align}
%	(w, a) \mapsto f_k(i, w, a) := \int Q_i(w,y) dA(y) 
%	\end{align}
%	is measurable since we assumed 
%	\begin{align}
%	(i, y) \mapsto \uta_i(\usig_i, k, h_i(y)) \label{supp:eqn:me1}
%	\end{align}
%	to be measurable (w.r.t.\ the product $\sigma$-algebra) and then we can apply standard arguments involved in Fubini's theorem (more specifically: \citep[Theorem 14.16, Eq.\ 14.6]{klenke2013probability}).
%	
%
%We show that for any values $v, w$ that $V, W$ are fixed to,
%\begin{align}
%i \mapsto [s_\pi]_i(w_i, \acov) = \sigma_i(w_i, \pi(\acov)) 
%\end{align}
%is measurable in $i$ (w.r.t.\ codomain $\slots=\{0, 1\}$ equipped with the power set as $\sigma$-algebra).
%
%
%
%We have
%	
%	
%	\[
%	F^{\Fcomp}_{\mu}(\mpara):=\int\left[\prod_{l\neq k} H_l^k\left(\int \Diffu_l^k(i,\yvec)d\mu(\yvec)\right)\right]r(i|\mpara)di\,,
%	\]
%	with $\Diffu_l^k(i,y)=\Ua_{i}(y,k)-\Ua_{i}(y,l)$.
%	
%	

\parag{Measurability of the mapping in Eq.~\ref{supp:eqn:fin2}:}

Let $a, w_i$ be arbitrary but fixed.
First, observe that for all $k \in \Slots$,
\begin{align}
i \mapsto f_k(i) := \E_{Y' \sim a}   \left( \uta_i(\usig_i, k, h_i(Y'))\right)
\end{align}
is measurable since we assumed 
\begin{align}
(i, y) \mapsto \uta_i(\usig_i, k, h_i(y)) \label{supp:eqn:me1}
\end{align}
to be measurable (w.r.t.\ the product $\sigma$-algebra) and then we can apply standard arguments involved in Fubini's theorem (more specifically: \citep[Theorem 14.16, Eq.\ 14.6]{klenke2013probability}).
%		
%		
%			We have to show that for $\usig \in \dom_\Usig, \acov \in \dom_\Acov$, $\sigma(\usig, \pi(\acov)) = (\sigma_i(\usig_i, \pi(\acov)))_{i \in \users} \in \dom_\Uact$. 
%		Equivalently, we have to show that
%		\begin{align}
%		\{ j \in \users : \sigma_j(\usig_j, \pi(\acov)) = 0 \}
%		\end{align}
%		is measurable.
%		Equivalently, we have to show that %for any $\usig \in \dom_\Usig, w \in \dom_W$, 
%		the mapping
%		\begin{align}
%		j \mapsto \sigma_j(\usig_j, \pi(a)) \label{supp:eqn:measu}
%		\end{align}
%		is measurable (w.r.t.\ codomain $D=\{0, 1\}$ equipped with the power set as $\sigma$-algebra).
%		
Now, observe that for all $i$ %and $\usig, w$ fixed, we have
\begin{align}
\sigma_i(\usig_i, a) = 0
\end{align}
iff
\begin{align}
0 \in \arg\max_{\uact_i} \E_{Y' \sim a}   \left( \uta_i(\usig_i, \uact_i, h_i(Y'))\right)
\end{align}
iff
%\begin{align}
%\E_{Y' \sim \pi(w)}   \left( \uta_i(\usig_0, 0, h_i(Y'))\right) = \max_{c_i} \E_{Y' \sim \pi(w)}   \left( \uta_i(\usig_0, c_i, h_i(Y'))\right)
%\end{align}
%iff
%\begin{align}
%\E_{Y' \sim \pi(w)}   \left( \uta_i(\usig_0, 0, h_i(Y'))\right) - \max_{c_i} \E_{Y' \sim \pi(w)}   \left( \uta_i(\usig_0, c_i, h_i(Y'))\right) = 0
%\end{align}
%iff
\begin{align}
\E_{Y' \sim a}   \left( \uta_i(\usig_i, 0, h_i(Y'))\right) - \max \left( \E_{Y' \sim \pi(\acov)}   \left( \uta_i(\usig_i, 0, h_i(Y'))\right), \E_{Y' \sim \pi(\acov)}   \left( \uta_i(\usig_i, 1, h_i(Y'))\right) \right) = 0
\end{align}
But the l.h.s.\ of the latter equation is a composition ($f_0(i) - \max(f_0(i), f_1(i)) $) of measurable functions (recall that we showed $i \mapsto f_k(i) = \E_{Y' \sim a}   \left( \uta_i(\usig_i, k, h_i(Y'))\right)$ to be measurable) that is measurable again \citep{klenke2013probability}.

\end{proof}

Also keep in mind the following statement, which guarantees measurability of assistant policies induced by strategy profiles.

\begin{SuppLemma}
\label{supp:lem:pis}
In the case of finitely many users,
given a strategy profile $s$, the corresponding assistant policy $\pi_s$ (Eq.~\ref{eqn:cp}) is measurable from $\alg_V$ to $\alg_A$.
%(This is relevant for the case of finitely many  users.)
\end{SuppLemma}

\begin{proof}
We have to show that 	$v \mapsto P_{G, s}( Y | v)$ is measurable, as a mapping from $\range_V$ equipped with $\alg_V$, to $\range_A$ (generally: the set of Borel measures on $Y$), equipped with $\alg_A$ (generally: the Borel sets induced by the weak topology on  $\range_A$).

%Under our assumptions (Section \ref{supp:sec:addass}), $P(X, W, V | v)$ is a regular conditional distribution, i.e., given any measurable set $S$, $P( (X, W, V) \in S | v)$ measurable in $v$ for \citep[Theorem 8.37]{klenke2013probability}.
%Keep in mind that, given any measurable set $S$, we have that $P( (X, W, V) \in S | v)$ is measurable in $v$ (by the definition of conditional expectations/distributions \citep{klenke2013probability}).

Let $f(x, w, v) := \bar{Y}(x, s(v, w))$.

Since we assumed $\dom_Y$ to be finite in the case of finitely many users, $\dom_A$ (generally the Borel measures on $\dom_Y$) is simply a subspace of the Euclidean space, and $\alg_A$ are simply the Borel sets on it. % \todo{(the weak topology on finite spaces is the classical Euclidean topology)}.
So $P( Y | v)$ can be seen as a finite-dimensional vector in the Euclidean space with components $P( Y = y^l | v) = P( (X, W, V) \in f^{-1}(\{y^l\}) | v), l = 1, \ldots, m$, assuming $\dom_Y = \{y^1, \ldots, y^m\}$.

So to show that 	$v \mapsto P_{G, s}( Y | v)$ is measurable, it is enough to show that each of its components $P( (X, W, V) \in f^{-1}(\{y^l\}) | v)$ is measurable in $v$. %\todo{right? since the Euclidean $\sigma$-algebra is the product $\sigma$-algebra?}
But this holds true since  given any measurable set $S$, we have that $P( (X, W, V) \in S | v)$ is measurable in $v$ (by the definition of conditional expectations/distributions \citep{klenke2013probability}).

%based on what we said above (given any measurable set $S$, $P( (X, W, V) \in S | v)$ measurable in $v$).

\end{proof}

Also keep in mind the following observation.

\begin{Remark}
\label{supp:rem:pf}
Recall how we defined the corresponding assistant policy $\pi_s(v)$ in Eq.~\ref{eqn:cp} as $P_{G,s}(Y|V=v)$.
Note that, since we assumed $\dom_Y$ to be equipped with the Borel sets as $\sigma$-algebra, the (regular) conditional distribution $P_{G,s}(Y|V=v)$ is a Borel measure on $Y$.
This implies that the output of $\pi_s(v)$ is guaranteed to be contained in the range $\dom_A$ we assumed for it -- the Borel measures (Section \ref{supp:sec:addass}).
%
%	\todo{dont we also have to prove that its actually a borel measure!?! MAYBE: Note that $P(Y|v)$ is a Borel measure for each $v$, since we equipped $\range_Y$ with the Borel sets and assumed the necessary measurabilities, and $P(Y|v)$ is a (regular) conditional distribution for $Y$. Hence $A$ has the ``correct'' range.}
\end{Remark}

\section{Proofs for Section \extref{sec:obj}}

\todo{in the end, copy the most recent versions of the results here}

\newcommand{\mapp}{J}
\newcommand{\Diffu}{Q}
\newcommand{\diffu}{p}
\newcommand{\yvec}{y}
\newcommand{\Ua}{\tilde{U}}
\newcommand{\newK}{|K|}
\newcommand{\Fcomp}{k}
\newcommand{\newk}{n}

%\end{samepage}

\subsection{Theorem \extref{prop:argmin_ne}} % and an elaboration of the implications}
\label{sec:pr_argmin_ne}

Before proving it, let us restate the result
\footnote{Note that even when not assuming assistant-separability, a BNE may be achieved. However, this would be a BNE w.r.t.\ a different game, where players would not use the full information available to them -- $V, W_i$.}
% can be weakened or even dropped. This is because assistant-separability may already be implied by $\lpred_{\pi} = 0$. But, more broadly speaking, to achieve $\lpred_{\pi} = 0$, it seems necessary.}%
:

\begin{CTheorem}{\extrefs{thm:predne}}
\label{suppcite:thm:predne}
We have, in the general setting (Setting \extref{set:gen}, with all users being inference-assistable and assistant-separable): % and assistant-best-responding):
\begin{citem}
	\item If the assistant policy $\pi$ in the assistant-based system $M$ (where all users are assistant-best-responding) is a self-fulfilling prophecy (i.e., $\lpred_{\pi} = 0$)% prediction accuracy $\lpred_\pi=0$
	, then the corresponding strategy profile $s_\pi$ \maybe{(Eq.\ \ref{eqn:s})} is a Bayesian Nash equilibrium (BNE) of the \bg $\game$.
	\maybe{If $\lpred_\pi=0$, then $\lne_\pi=0$.}
	%In words: If the assistant reaches an optimal prediction policy (``self-fulfilling prophecy''), the resulting user behavior is comparable to a BNE in the \bg.
	\item Conversely, \maybe{in the case of finite ranges, }if the strategy profile $s$ is a strict BNE of the \bg $\game$, then the corresponding assistant policy $\pi_s$ is a self-fulfilling prophecy. % (i.e., $\lpred_{\pi_s}{=}0$). % for the  $\pi_s$\maybe{ (Eq.~\ref{eqn:cp})}.
	%In words: If $s$ is a strict BNE of the \bg, then the assistant reaches an optimal prediction policy by predicting the user behavior under $s$.
	\maybe{\item In particular, if a strict BNE exists in $G$, then: if $argmin .... L_\pi=0$, then $s_\pi$ is essentially a BNE of $\game$.}
\end{citem}
\end{CTheorem}

\begin{proof}[Proof of Theorem \suppciteref{prop:predne}]

\remark{main things i'm not entirely sure about:\\
	- whether i can just fix one of the $c_i$ to a constant and things (in particular $U_i$) are still measurable (might be the case, in case i assume product sigma algebras -- based on extending klenkes prdoct sigma algebra arg) -- but am i actually doing this anywhere ($C$ is not an RV!)\\
	- whether everything i take measure over is actually measurable and well defined, e.g. $\bar{U}$.\\
}

\remark{my impression is that for this proof, $\Uact$ itself does not have to be a random variable. only $\obsmech$ has to be an RV again. whenever $\Uact$ occurrs, I think one can replace RV by just a *function* form $\Omega$ or $\Usig$ or so. similarly for $\pi$ -- maybe in the end only $\sigma(\usig, \pi(w))$ has to be measurable from $\usig, w$}

Let $Z_i := h_i(Y), i \in \users$.

\begin{Claim}
	\label{claim:yomama1}
	If $\lpred_\pi=0$, then $s_\pi$ is a BNE of $\gamegen$.
\end{Claim}

\begin{proof}[Proof of Claim \ref{claim:yomama1}]

	\newcommand{\Om}{X}
	\newcommand{\om}{x}
	%		
	%		First let us show that $s_\pi$, i.e., the function $i, \usig, \acov \mapsto \sigma_i(\usig_i, \pi(\acov))$ is measurable (w.r.t.\ codomain $K=\{0, 1\}$ equipped with the power set as $\sigma$-algebra).
	%		Recall that we assume $\pi$ to be measurable (this is only briefly remarked in Definition \extref{def:m}).
	%		So it is sufficient to show that $i, \usig, a \mapsto \sigma_i(\usig_i, a)$ is measurable.
	%		
	%		Let us show this for the case of two slots, i.e., $K = \{0,1 \}$, the general case works analogously.
	%		Recall that we assume $\pi$ to be measurable (this is only briefly remarked in Definition \extref{def:m}).
	%		
	First note that $s_\pi$ can be written slightly more compactly than in Eq.~\ref{eqn:s}\footnote{There we used the notation based on the conditional expectation not because we refer to some average $B_i$, but only to rigorously refer to the value of $B_i$ conditioned on $W_i, V$, which is actually fully determined by these variables.}, using $\sigma$ as defined in Section \ref{supp:sec:addn}, in the following way, for all $\usig_i, \acov$:
	\begin{align}
	[s_\pi]_i(\usig_i, \acov) := \sigma_i(\usig_i, \pi(\acov)). \label{supp:eqn:sp}
	\end{align}
	
	\stress{Main derivation:}
	
	We state the following sequence of equalities for the case of finite $\users$ with stochastic $V, W$; the case of infinite $\users$, where $V, W$ are constant, is analogous but even simpler (essentially one has to drop all the occurring $V, W, v, w$).
	We have, for all $i, \usig_i$ and almost all $\acov$,
	\begingroup
	\allowdisplaybreaks
	\begin{align}
	&[s_\pi]_{i}(\acov, \usig_i)  \label{supp:eqn:de0}\\
	&= \sigma_i(\usig_i, \pi(\acov))  \label{supp:eqn:de01} \\
	%&= P_{G, s_\pi}(C_i=c_{i}|w, \usig_i) \\
	%&=P_{M, \pi}(\sigma(\Usig_i, \pi(W))=c_i|w, \usig_i) \label{supp:eqn:de1}\\
	%&=P_{M, \pi}(C_i=c_i|w, \usig_i) \label{supp:eqn:de1}\\
	&\in \arg\max_{\uact_i'} \E_{Y' \sim \pi(\acov)}   \left( \uta_i(\usig_i, \uact_i', h_i(Y')) \right) \label{supp:eqn:de1}\\
	&= \arg\max_{\uact_i'} \E_{Y' \sim P_{M, \pi}(\sur|\acov)}   \left( \uta_i(\usig_i, \uact_i', h_i(Y')) \right) \label{supp:eqn:de2}\\
	%&= \delta \left(\arg\max_{c_i'} \E_{Y' \sim P_{M, \pi}(\sur'|w)}   \left( \uta_i(\usig_i, c_i', h_i(Y')) \right) \right)(c_i) \label{supp:eqn:de3}\\
	&= \arg\max_{\uact_i'} \E_{Z_i' \sim P_{M, \pi}(h_i(\sur)|\acov)}   \left( \uta_i(\usig_i, \uact_i', Z_i') \right)  \label{supp:eqn:de9} \todo{?}\\
	&= \arg\max_{\uact_i'} \E_{Z_i' \sim P_{M, \pi}(Z_i|\acov)}   \left( \uta_i(\usig_i, \uact_i', Z_i') \right)  \label{supp:eqn:de8}\\
	&= \arg\max_{\uact_i'} \E_{Z_i' \sim P_{M, \pi}(Z_i|\acov, \usig_i)}   \left( \uta_i(\usig_i, \uact_i', Z_i') \right) \todo{?} \label{supp:eqn:dei}\\
	&= \arg\max_{\uact_i'} \E_{Z_i' \sim P_{G, s_\pi}(Z_i|\acov, \usig_i)}   \left( \uta_i(\usig_i, \uact_i', Z_i') \right)  \label{supp:eqn:deg}\\
	%&= \delta \left(\arg\max_{c_i'} \E_{Z_i' \sim P_{G, s_\pi}(h_i(\sur)|w, \usig_i)}   \left( \uta_i(\usig_i, c_i', Z_i') \right) \right)(c_i) \label{supp:eqn:de6} \todo{?!?!?}\\
	%&= \delta \left(\arg\max_{c_i'} \E_{Z_i' \sim P_{G, s_\pi}(Z_i|w, \usig_i)}   \left( \uta_i(\usig_i, c_i', Z_i') \right) \right)(c_i) \label{supp:eqn:de7} \todo{?!?!?}\\
	&= \arg\max_{\uact_i'} \E_{Z_i' \sim P_{G, [s_\pi]_{-i}}(Z_i|\acov, \usig_i)}   \left( \uta_i(\usig_i, \uact_i', Z_i') \right)  ,\label{supp:eqn:de7} 
	%&= \delta \left( \arg\max_{c_i'} U_i(\usig_i, c_i', Q_{G, [s_e\pi]_{-i}}(\surm_i|w, \usig_i) ) \right)(c_i), \label{supp:eqn:de7}
	\end{align}
	\endgroup
	where: \todo{i guess here, measurability of $s$ in $\Acov, \Usig_i$ is needed (but not in $i$!?!) or do i actually only need measurability of $\obsmech(s(....))$ ?}
	\begin{citem}
		\item Eqs.~\ref{supp:eqn:de01}, \ref{supp:eqn:de1} hold by definition of $s_\pi$ and the definition of the user behavior $\sigma_i$ (Section \ref{supp:sec:addn}, Eq.~\extref{eqn:mbr} -- assumption ``assistant-best-responding''), respectively.
		\item To understand Eq.~\ref{supp:eqn:de2}, let us look at what our assumption $\lpred_\pi=0$ implies.
		Based on its very definition, it implies
		%Recall that \todo{!??!??!} $L_\pi = \E \left( \| \pi(W) - P_\pi(\sur|W) \|_2^2 \right)$. Then $L_\pi=0$ implies
		\begin{align}
		\pi(\acov) &= P_{M, \pi}(\sur|\acov) %\\
		%&= \pi(w)(\surm_i) \label{supp:eqn:iq1} \\
		%&= P_\pi(\surm_i|w) \\
		%&= P_{M, \pi}(\sur|w, \usig_i) \label{supp:eqn:iq2},
		\end{align}
		for almost all $\acov$.
		\item Eq.~\ref{supp:eqn:dei} follows from our assumption \asssep that $Z_i \ind \Usig_i | \Acov$ \maybe{in $M$} for any $\pi, \sigma$.\todo{maybe slightly update this to the current formulation} % (part of Assumption \extref{asm:predne}). % together with the fact that $h_i(Y) = Z_i$.
		\item To understand Eq.~\ref{supp:eqn:deg}, observe that the only thing that can be different between $M_\pi$ (when ignoring $A$) and $G$ (with a ``plugged in'' strategy profile) is the mechanism that generates $\Uact$ from $V, W$.
		But, by our definition of $s_\pi$, this mechanism is in fact the same in $M_\pi$ and $G$ with ``plugged in'' $s_\pi$. Therefore, all (random) variables, in particular $Z_i$, coincide between $M_\pi$ and $G$ with $s_\pi$.
		%			$s_\pi$
		%			based on Eq.~\extref{eqn:s}, 
		%			\begin{align}
		%			\Uact_{M, \pi} = \Uact_{G, s_\pi}. \label{supp:eqn:pp1}
		%			\end{align}
		%			But the mechanism that generates $\Uact$ from $V, W$ is
		%			%(or rather: mapping on $\Omega$) \todo{$C$ is no RV anymore -- still works?} 
		%			%(leaving aside $W$) 
		%			that could change between $M_\pi$ and $G$ with $s_\pi$. Therefore, all random variables, in particular $Z_i$, coincide between $M_\pi$ and $G$ with $s_\pi$.
		\item To understand Eq.~\ref{supp:eqn:de7} note that $Z_i$ is defined without $\Uact_i$ needing to be defined. Therefore, it is already defined in $G$ with ``incomplete strategy profile'' $[s_\pi]_{-i}$ alone.
		\item Generally, note that terms like ``$P_{M, \pi}(\sur|\acov)$'' -- a regular conditional distribution -- though we would not necessarily always need them, are well-defined and exist in our setting (of discrete or Euclidean ranges) \citep[Theorem 8.37]{klenke2013probability}.
	\end{citem}
	
	But Eqs.~\ref{supp:eqn:de0} through \ref{supp:eqn:de7} %, together with our assumption of inference-assistability (Eq.~\extref{eqn:uta}), 
	mean that for almost no $\acov, \usig_i$, player $i \in \users$ could improve his utility by deviating from $[s_\pi]_i(\acov, \usig_i)$.
	
	\stress{Measurability discussion:}
	
	Generally, note that
	%$P_{M, \pi}(Y|w)$ and related quantities are well defined because we assumed $Y$ to be a random variable;
	$\uta_i, i \in \users$ is measurable also when we fix some of its arguments, because we assumed it to be measurable w.r.t.\ the respective %)w.r.t.\ the Borel $\sigma$-algebra (which, in case of $R^l$, coincides with the 
	product $\sigma$-algebra \citep[Lemma 14.13 and Theorem 14.16]{klenke2013probability}.
	
	Still for the case of $\users$ finite, note that $s_\pi$ is measurable w.r.t.\ the product $\sigma$-algebra $\alg_W \otimes \alg_V$ to $\alg_B$ (which is necessary for it to be a strategy profile), for the following reasons: We assumed $\pi$ to be measurable w.r.t.\ $\alg_V$ to $\alg_A$. And $\sigma=(\sigma_i)_{i \in \users}$ is measurable w.r.t.\ $\alg_W \otimes \alg_A$ to $\alg_B$ due to the first part of Lemma \ref{supp:lem:m}.\footnote{Since $\alg_B$ is the product $\sigma$-algebra, $\sigma$ measurable is equivalent to $\sigma_i$ measurable for all $i$ \citep[Corollary 1.82]{klenke2013probability}.} 
	But $s_\pi$ is just the composition $\sigma(\cdot, \pi(\cdot))$. % \todo{maybe double check in Klenke but should be true}.
	
	It remains to be shown that $s_\pi$ is a strategy profile, in terms of measurability  (in the sense of Section \ref{supp:sec:addass}), also for the \emph{case of $\users$ infinite and $V, W$ constant}. 
	Specifically, we have to show that for any values $v, w$ that $V, W$ are fixed to, that
	\begin{align}
	i \mapsto [s_\pi]_i(w_i, \acov) = \sigma_i(w_i, \pi(\acov)) 
	\end{align}
	is measurable in $i$ (w.r.t.\ codomain $\slots=\{0, 1\}$ equipped with the power set as $\sigma$-algebra).
	But this directly follows from the second part of Lemma \ref{supp:lem:m} (plugging in $\pi(v)$ for $a$).
	%
	%		But this is 
	%		Let us explicate this for the case of $\slots = \{0,1\}$. The case of more slots works similarly.
	%		We show that for any values $v, w$ that $V, W$ are fixed to,
	%		\begin{align}
	%		i \mapsto [s_\pi]_i(w_i, \acov) = \sigma_i(w_i, \pi(\acov)) 
	%		\end{align}
	%		is measurable in $i$ (w.r.t.\ codomain $\slots=\{0, 1\}$ equipped with the power set as $\sigma$-algebra).
	
	%		
	%		
	%		
	%		First note that for $s_\pi$ (Eq.~\extref{eqn:s}) the following holds:
	%		First,
	%		Second, for fixed $w, v$, the mapping $i \mapsto [s_\pi]_i(w_i, v)$ is measurable w.r.t.\ $\usersalg$ and $2^\slots$ since $[s_\pi]_i(w_i, v) = \sigma_i(w_i, \pi(v))$ and we assumed the analogous property for $\sigma_i$.
	
	Everything together implies that $s_\pi$ is a BNE of $G$.
	%		
	%		Furthermore, we assumed $\sigma, \pi$ to be measurable (we stated this, for the sake of brevity, only vaguely in Definition \extref{def:m}; to see that $\sigma$ usually satisfies this measurability assumption, see also Corollary \extref{cor:nonat} and in particular Eq.~\ref{supp:eqn:measu} and the following lines in its proof), hence also their composition $s_\pi$ (see Eq.~\extref{eqn:s}) is measurable.
	%		Together, this implies that 
	%		Hence $s_\pi$ is a BNE.

	%Eq.\ref{supp:eqn:de1} is simply the definition of $s_\pi$ \todo{since $S$ concatenates $\pi$ and $\sigma$?!?! see above};
	%Eq.\ \ref{supp:eqn:de2} is our assumption about the customer behavior $\sigma_i$; and
	%Eq.\ \ref{supp:eqn:de3} is simply Eq.\ \ref{supp:eqn:fin2}.
	%

\end{proof}

\begin{Claim}
	\label{claim:yomama2}
	Conversely, if $s$ is a strict BNE of $G$, then $\lpred_{\pi_s}=0$.
\end{Claim}

\begin{proof}[Proof of Claim \ref{claim:yomama2}]
	Let $s=(s_i)_{i\in \users}$ be a strict BNE of $G$. That is, for all $i, \usig_i, \acov$,
	\begin{align}
	s_i(\usig_i, \acov) \in \arg\max_{\uact_i'} \E   \left( \utb_i(X, (\uact_i, (s_j(\Usig_j, \Acov ))_{j \in \users \setminus \{i\}} \middle| \usig_i, \acov  \right) , \label{supp:eqn:am}
	%\arg\max_{c_i'} \E_{G, s_{-i}} \left( \uta_i(\usig_i, c_i', Z_i) \middle| \usig_i, w \right) \todo{?}\label{supp:eqn:am}
	\end{align}
	with the argmax being unique.

	%\stress{First}, recall that $\pi_s(w) = P_{G, s}(Y|w)$ for all $w$.
	%%Let $\pi_{s, i} = \pi_s(C_{-i})$ \todo{clarify this notation} if $\sur=C$ and $\pi_{s, i} = \pi_s$ if $\sur = \agg$.
	%This implies that for all $w$,
	%\begin{align}
	%\pi_{s}(w) &= P_{G, s}(\sur|w) \label{supp:eqn:j1} \\
	%&= Q_{G, s_{-i}}(\surm_i|w) \label{supp:eqn:j2}\\
	%&= Q_{G, s_{-i}}(\surm_i|w, \usig_i) \label{supp:eqn:j3}
	%\end{align}
	%where, to understand Eq.~\ref{supp:eqn:j1}, recall that in case $\sur=\agg$, we have $\sur=\surm_i$ for all $i \in N$, as described in Eq.~\ref{supp:eqn:nu};
	%Eq.~\ref{supp:eqn:j2} is due to $\surm_i$ only being influenced by $C_{-i}$, i.e., independent of $C_i$, and thus being determined by $s_{-i}$ already; and
	%Eq.\ \ref{supp:eqn:j3} is due to the fact that $\surm_i$ is a function of only $C_{-i}$, and so our assumption $C_{-i} \ind \Usig_{i} | W$ implies $\surm_i \ind \Usig_{i} | W$ \todo{NO PROBLEMS DUE TO INFINITY HERE?}.
	
	%In what follows, let $C_{M, \pi_s, i}$ denote the random variable $C_i$ (customer $i$'s action) in the model $M$ under assistant $\pi_s$, and similarly $C_{M, \pi_s}$.
	Similar as above, we state the following derivation for the case of finite $\users$ with stochastic $V, W$; the case of infinite $\users$, where $V, W$ are constant, is analogous but even simpler (essentially one has to drop all the occurring $V, W, v, w$).

	\stress{First}, for the case of the assistant's policy being $\pi_s$, we have for all $i,\acov, \usig_i$,
	\todo{there seemed to be some incoherent leftovers ($f_{W_i}$, intermedita $\bar{U}$ etc. in the above commented align. here's a cleaned up one, although im not entirely sure if now everything is fine:}
	\begingroup
	\allowdisplaybreaks
	\begin{align}
	%C_{M, \pi_s, i} &= 
	%&\sigma_i(\usig_i, a) \\
	&\sigma_i(\usig_i, \pi_s(\acov)) \\
	&\in \arg\max_{\uact_i'} \E_{Y' \sim \pi_s(\acov)}   \left( \uta_i(\usig_i, \uact_i', h_i(Y')) \right) \label{supp:eqn:um3} \\
	&= \arg\max_{\uact_i'} \E_{Y' \sim P_{G, s}(Y|\acov)}   \left( \uta_i(\usig_i, \uact_i', h_i(Y')) \right) \label{supp:eqn:um4} \\
	&= \arg\max_{\uact_i'} \E_{Z_i' \sim P_{G, s}(h_i(Y)|\acov)}   \left( \uta_i(\usig_i, \uact_i', Z_i') \right) \label{supp:eqn:um5} \\
	&= \arg\max_{\uact_i'} \E_{Z_i' \sim P_{G, s}(Z_i|\acov)}   \left( \uta_i(\usig_i, \uact_i', Z_i') \right) \label{supp:eqn:um6} \\
	&= \arg\max_{\uact_i'} \E_{Z_i' \sim P_{G, s}(Z_i|\acov, \usig_i)}   \left( \uta_i(\usig_i, \uact_i', Z_i') \right) \label{supp:eqn:um7} \\
	&= \arg\max_{\uact_i'} \E_{Z_i' \sim P_{G, s_{-i}}(Z_i|\acov, \usig_i)}   \left( \uta_i(\usig_i, \uact_i', Z_i') \right) \label{supp:eqn:um8} \todo{?} \\
	&= \arg\max_{\uact_i'} \E_{G, s_{-i}}   \left( \uta_i(\Usig_i, \uact_i', Z_i) \middle| \usig_i, \acov \right) \label{supp:eqn:um9} \todo{does this follow directly based on the defi of the conditional expectation?} \\
	%&= \arg\max_{\uact_i'} \E_{G, s_{-i}}   \left( \uta_i(\Usig_i, \uact_i', Z_i) \middle| \usig_i, \acov \right) \label{supp:eqn:um9a} \\
	%&= \arg\max_{\uact_i'} \E   \left( \uta_i(f_{\Usig_i}(X), \uact_i', f_{Z_i}(X, (s_j(\Usig_j, \Acov ))_{j \in N}) ) \middle| \usig_i, \acov \right) \label{supp:eqn:um10} \\
	&= \arg\max_{\uact_i'} \E_{G, s_{-i}}   \left( \uta_i(\Usig_i, \uact_i', h_i(\obsmech(X, \uact_i', (s_j(\Usig_j, \Acov ))_{j \in \users \setminus \{i\}}) ) ) \middle| \usig_i, \acov \right) \label{supp:eqn:um11} \todo{?} \\
	%&= \arg\max_{\uact_i'} \E_{G, s_{-i}}   \left( \utb_i(X, \uact_i', (s_j(\Usig_j, \Acov ))_{j \in \users \setminus \{i\}}) \middle| \usig_i, \acov \right) \label{supp:eqn:um} \todo{?} \\
	&\ni s_i(\usig_i, \acov) , \label{supp:eqn:um12}
	\end{align}
	\endgroup
	where: \todo{i may have to adapt the fact that $\Usig_i$ is a function of $X$. } 
	%\todo{furthermore,i guess here, measurability of $s$ in $W, \Usig_i$ is needed (but not in $i$!?!)}
	\begin{citem}
		\item Eq.~\ref{supp:eqn:um3} is our assumption that users are assistant-best-responding.
		\item Eq.~\ref{supp:eqn:um7} follows from our assumption \asssep that $Z_i \ind \Usig_i | \Acov$ for any $\pi, \sigma$ and thus also in $G$ with any $s$.
		\item Eq.~\ref{supp:eqn:um11} follows from how we defined $Z_i$. % $f_{Z_i}$.
		\item Eq.~\ref{supp:eqn:um12} is due to $s$ being a BNE (and our assumption that users are inference-assistable (Eq.~\extref{eqn:uta})).
	\end{citem}
	
	\maybe{\todo{not sure if the following makes sense since we did not assume $U_i$ to be measurable on $\Uact$} Generally note that quantities like 
		\[ \E_{G, s_{-i}}   \left( \utb_i(X, (\uact_i, (s_j(\Usig_j, \Acov ))_{j \in \users \setminus \{i\}} \middle| \usig_i, w \label{supp:eqn:am} \right) \]
		can be read as
		\[ \E_{G, s_{-i}}   \left( \utb_i(X, (s_i(\Usig_j, \Acov ), (s_j(\Usig_j, \Acov ))_{j \in \users \setminus \{i\}} \middle| \usig_i, w \label{supp:eqn:am} \right) \]
		for some appropriate $s_i$ and thus are well defined since we assumed $U_i, i \in \users$, which is given by $\utb_i(X, (s_i(\Usig_j, \Acov ), (s_j(\Usig_j, \Acov ))_{j \in \users \setminus \{i\}}$ to be measurable.
		\notsure{THIS IS NOT FULLY SATISFACTORY, instead it may be clearer to say that $U_i$ is a random variable on $\Acov, \Usig$ in our assumptions .... but what is an approproate sigma-algebra on $\Usig$?
			generally, there may be an issue when fixing one argument of an uncountably infinite product becuse then the klenke argument may not hold anymore \url{https://math.stackexchange.com/questions/248032/is-product-of-borel-sigma-algebras-the-borel-sigma-algebra-of-the-product-of/248587}
			ACTUALLY  THIS IS MAYBE THE WRONG INTUITION AT ALL - CONFUSING OBSERVING AND DOING;
			THE SIMPLEST SOLUTION may be to equip $\Usig$ by a product sigma algebra of the $\Usig_i$. because then the klenke Lemma 14.13 argument should still hold (or rather: be extendable to this case).
			ACUALLY: maybe this argument fully makes sense after all and i don't have to worry aout product salgebras etc. because that's what its all about: is there another *valid* strategy profile modified only on $i$ which is better (not even sure if it has to be better on all $\usig_i$ -- see MAS argument). BUT: this argument does maybe not hold for the assistant-based setting -- unless for $\uta$.
	}}
	
	%Eq.~\ref{supp:eqn:um} is well-defined based on the argmax being unique (due to strictness of the BNE);
	%Eq.~\ref{supp:eqn:um2} is Eq.\ \ref{supp:eqn:j3};
	%and Eq.~\ref{supp:eqn:um3} is Eq.\ \ref{supp:eqn:am}.
	
	Since we assumed the above argmax to be unique, we get that for all $i, \usig_i, \acov$,
	\begin{align}
	\sigma_i(\usig_i, \pi_s(\acov)) = s_i(\usig_i, \acov) . \label{supp:eqn:umm}
	\end{align}

	%Based on this, we have
	%\begin{align}
	%%Q_{\pi}(Y|w) 
	%P_{M, \pi}(C|w)
	%&= P_{M, \pi}((C_i)_{i \in N}|w) \\
	%&= P_{M}((s_i(\Usig_i, W))_{i \in N}|w) \label{supp:eqn:y1} \\
	%&= P_{G}((s_i(\Usig_i, W))_{i \in N}|w) \label{supp:eqn:y2} \\
	%&= P_{G, s}(C|w) , \label{supp:eqn:y3}
	%%&= \pi(w)
	%\end{align}
	%where 
	%Eq.~\ref{supp:eqn:y1} is Eq.~\ref{supp:eqn:um3}; and
	%Eq.~\ref{supp:eqn:y2} is due to the joint distribution of $\Usig, W$ being the same in the assistant-based model $M$ and the \bg $G$.

	\stress{Second}, we have,  for all $i, \acov$,
	\begin{align}
	&P_{M, \pi_s}(Y|\acov) \\
	%&= P_{M, \pi_s}(\obsmech(X, C)|w) \todo{\text{any issues due to infinite-dim $C$?}}\\
	&= P\left(\obsmech\left(X, \left(\sigma_i\left(\Usig_i, \pi_s(\Acov)\right)\right)_{i \in \users}\right)|\acov\right) \\
	&= P\left(\obsmech\left(X, \left(s_i\left(\Usig_i, \Acov \right)\right)_{i \in \users}\right)|\acov\right) \label{supp:eqn:z0} \\
	&= P_{G, s}(Y|\acov) \label{supp:eqn:z1} \\
	&= \pi_s(\acov), \label{supp:eqn:z2}
	\end{align}
	where:
	\begin{citem}
		\item Eq.~\ref{supp:eqn:z0} follows from above's derivation ending with Eq.~\ref{supp:eqn:umm}.
		\item Eq.~\ref{supp:eqn:z2} is just the definition of $\pi_s$.
	\end{citem}
	%where 
	%Eq.~\ref{supp:eqn:z1} is due to Eq.~\ref{supp:eqn:y3}; and
	%Eq.~\ref{supp:eqn:z2} is just the definition of $\pi_s$.
	
	%Case $\sur = \agg$:
	%Let $f$ denote the aggregation mapping, i.e., $f: c \mapsto (\int_{i \in N} [ C_i = k] di)_k$.
	%Then we have
	%\begin{align}
	%Q_{\pi}(Y|w) &= \E_{M, \pi}(f(C)|w) \\
	%&= \E_{G, s}(f(C)|w) \label{supp:eqn:x1} \\
	%&= \pi_s(w), \label{supp:eqn:x2}
	%\end{align}
	%where 
	%Eq.~\ref{supp:eqn:x1} is due to Eq.~\ref{supp:eqn:y3}; and
	%Eq.~\ref{supp:eqn:x2} is just the definition of $\pi_s$.
	
	This implies $\lpred_{\pi_s}=0$, which is what had to be shown.

	%		\todo{point out that everything is measurable in M under pis because of sigma being measurable as stated above}
	%			
	%			\todo{Since we assumed $\dom_{Y}$ (together with the implicit metric) to be Polish, we have, based on \citep[Theorem 8.37 or Theorem 8.29 for the case of $\R$]{klenke2013probability}, that there exists a regular conditional distribution $P(Y|v)$. But this means that for each $v$, $\pi_s(v) = P(Y|v)$ is a probability distribution of $Y$ and thus an element of .... $\mathcal{M}_1$ ...... as required ....}
	
	\stress{Measurability discussion:} 
	
	For the case of finitely many users, see Lemma \ref{supp:lem:pis}. In the case of infinitely many, nothing has to be shown (regarding the correctness of the codomain of $\pi_s$, see Remark \ref{supp:rem:pf}). %\todo{IS THIS CORRECT? but is pi measurable? via continuity? or something like the measurablity for a fixed argument (of the joint density) like in fubini }

\end{proof}

\end{proof}

\renewcommand{\game}{\gnonat}

\subsection{Corollary \extref{cor:nonat}}
\label{sec:pr_cn}

Before proving it, let us restate the result form the main paper:

\begin{CCorollary}{\extrefs{cor:nonat}}
Setting~\extref{set:agg} is a special case of Setting \extref{set:gen}.
In particular, it satisfies the conditions of Theorem \extref{prop:predne} and hence the theorem's implications hold for $M = \mlarge$ and $\gamegen = \gnonat$.
%Setting~\extref{set:agg} together with Assumption \extref{asm:aggr} satisfies the conditions of Theorem \extref{prop:predne}, and therefore its implications hold for $\gamegen = \gnonat$.
\end{CCorollary}

\begin{proof}[Proof of Corollary \extref{cor:nonat}]

%\remark{maybe $C$ can be turned into a random variable by considering the product sigma algebra on it (which is the whole power set). but the question remains how to show that then $Y(X, C)$ is a measurable function of the product sigma algebra on $X,C$}

%\remark{first double-check of the proof completed}

Throughout this proof, let $\acov, \usig$ be arbitrary but fixed.

%\todo{what about assumptions A2-4?}

\stress{Part 1: show that general model assumptions of Section \extref{sec:model} are satisfied}

\stress{Regarding correctness of the range of $\sigma(a, \pi(\acov))$ (i.e., showing that it ranges within $\dom_\Uact = \{ \uact \in \N^\users : \text{ for all \( k \in \Slots, \{ j \in \users : \uact_j = k \} \) is measurable    } \} $):}

This follows from the fact that Setting~\extref{set:agg} satisfies the requirements of Setting~\extref{set:gen} w.r.t.\ the continuity of $\tilde{U}_i$ (that we stated in detail in Section \ref{supp:sec:addass}), which was all we needed in Lemma \ref{supp:lem:m}. Because the (second part of the) lemma implies that for all $w, v$,
\begin{align}
i &\mapsto \sigma_i(w_i, \pi(\acov)) 
\end{align}
is measurable, which is what had to be shown.

\stress{Regarding the measurability of all mechanisms:} %\todo{need to adapt to multidimensional $Y$?}

%\todo{OLD commented below:} 
%\emph{Regarding measurability of $x \mapsto \obsmech(x, \sigma(\usig, \pi(w)))$:}
%
%\remark{this measurability i need to be able to talk about $P_{M \pi} (Y|...)$ in the proof of prop1}
%
%Recall that we assumed that $W=w, \Usig=\usig$ are fixed. So it is enough to show that $x \mapsto \obsmech(x, c)$ is measurable, for $c = \sigma(\usig, \pi(w))$.
%
%We have
%\begin{align}
%[ \obsmech(x, c) ]_1 = \int c_i R_\mpara(d i) = R_\mpara(\{ j \in N : c_j = 1 \}).
%\end{align}
%
%But above we showed that $\{ j \in N : c_j = 1 \}$ is measurable (i.e., is in $\mathcal{N}$) and so, based on our assumption about $R_\mpara$, $R_\mpara(\{ j \in N : c_j = 1 \})$ is measurable in $x$ (more specifically, $\alg_X$-Borel-measurable).

\emph{Regarding product measurability of $x \mapsto \obsmech(x,\uact)$} % \todo{STAYS but simplify}}

We have to show that $x \mapsto \int \uact_i r(i|x) d i$ is measurable for fixed $b$.

We show the more general statement that $\dom_X \times \dom_\Uact \to \dom_Y = \R; (x, \uact) \mapsto f(x,\uact) := \int \uact_i r(i|x) d i $ is $\alg_X \otimes \alg_\Uact$-Borel measurable,
where, just for the sake of this proof, we assume $\dom_B$ to be equipped with a $\sigma$-algebra $\alg_\Uact$ as will be detailed below. (This implies what needs to be shown because in our setting measurability/contiuity in both arguments implies the same for the individual arguments when fixing the respective other \citep[Lemma 14.13]{klenke2013probability}.)

Here, let $L^2$ denote the Lebesgue space of square integrable functions over $[0, 1]$ w.r.t.\ the Lebesgue measure (usually denoted $L^2([0, 1])$).

\newcommand{\bor}{\mathcal{B}}
Let $\tau_X$ denote the topology of $\dom_X$.
Let $\tau_{L^2}$ denote the topology of $L^2$.
Let $\bor_\tau$ denote the Borel $\sigma$-algebra induced by a topology $\tau$.

Recapture our assumptions:
\begin{citem}
	\item $\alg_X = \bor_{\tau_X}$,
	\item $\alg_\Uact = \bor_{\tau_{L^2}}$,
	\item $x \mapsto r(\cdot|x) =: r_x$ is continuous from $\tau_X$ to $\tau_{L^2}$.
\end{citem}

Let $\tau_1$ be the product topology of $\dom_X, L^2$ and
$\tau_2$ be the product topology of $L^2, L^2$.
%
%So
%\begin{align}
%f: x &\mapsto r(\cdot|x) =: r_x \\
%g: (q, c) &\mapsto \int c_i q(i) di
%\end{align}
%are continuous mappings,
%the latter one w.r.t.\ the product topology on $L^2, L^2$ to $\R$.

It follows from our assumptions that the mapping
\begin{align}
f: (x, \uact) \mapsto (r_x, \uact)
\end{align}
is continuous w.r.t.\ source topology $\tau_1$ and target topology $\tau_2$.

But the $L^2$ inner product $\langle \cdot, \cdot \rangle$ is continuous w.r.t.\ source topology $\tau_2$ and target topology $\R$. Therefore, the concatenation $\langle \cdot, \cdot \rangle \circ f$ is continuous from $\tau_1$ to $\R$. Hence it is measurable w.r.t.\ the Borel $\sigma$-algebra $\bor_{\tau_1}$ to the Borel sets on $\R$.

We assumed $(\dom_X, \tau_X)$ to be Polish.
And $L^2$, the space of $\Uact$, is Polish.
Therefore (based on \citep{klenke2013probability}[Theorem 14.8]), $\bor_{\tau_1}$, the Borel $\sigma$-algebra induced by the product topology on $\dom_X, L^2$, coincides with the product $\sigma$-algebra $\alg_X \otimes \alg_\Uact$.

%	
%	
%	\emph{Regarding $\sigma$:  \todo{DROP since no meas assumpiton anymore}} \todo{MAYBE: if algB is the one generated by the singletons then we can take the same one for VW because any constant function should be measurable w.r.t. it!? POTENTIAL PROBLEM: then barY is not measurable w.r.t. it ...}
%	
%	\todo{this may be the nastiest but easiest way:}
%	W.l.o.g. we can let $\dom_{A}$ be a singleton. Together with $\Usig$ being constant (as with $\pi$ below), measurability of $\sigma$ follows.
%	
%	
%	
%	
%	\emph{Regarding $\pi$:  \todo{DROP since no meas assumpiton anymore}}
%	
%	$\pi$ is trivially measurable since its inputs are constant by assumption and so w.l.o.g.\ we can assume $\pi$ to be a constant function and thus measurable.
%	

\stress{Regarding correctness of $\dom_A$ (i.e., that it contains the distribution over $Y$ that is entailed by it):}
We have to show that $P(Y|A=a) \in \dom_A$ for any $a \in \dom_A$.
To see this, note that we above showed that $x \mapsto \obsmech(x, \uact)$ is measurable with $\alg_Y$ being the Borel $\sigma$-algebra on $\dom_Y$.
Therefore the pushforward measure $P(Y|A=a) = P(g(X)|A=a)$, for $g(x) = \obsmech(x, \uact)$, is a Borel measure, i.e., element of $\range_A$.

%	\todo{are there any other mechanisms for which measurability have to be shown? what about $U_i$?}

%\todo{SKETCH:}
%Let $X=x$ be fixed and thus also $\Usig_0=\usig_0, A=a$ is fixed.
%Under this $x$ we have
%\begin{align}
%\sigma(i, \usig_0, a) = 0
%\end{align}
%iff
%\begin{align}
%\E_{Y' \sim a}   \left( \uta_i(\usig_0, 0, h_i(Y'))\right) - \max_{c_i} \E_{Y' \sim a}   \left( \uta_i(\usig_0, c_0, h_i(Y'))\right) = 0 .
%\end{align}
%But the l.h.s.\ of the latter equation is a composition of measurable functions in $i$ (as assumed) that is measurable again. So the measurablity
%In particular, let $k=0$, then
%\begin{align}
%\{ j \in N : C_j = k \} = \{ j \in N : \sigma(j, \usig_0, a) = 0 \}
%\end{align}
%is measurable. And $k=1$ is just the compliment.

%\todo{PROOF IDEA FOR THE CASE THAT $W$ IS NOT CONSTANT:
%$Y = \int C(i, f(x)) d R_x(i)$ (with $C$ itself being generated by a similar integral, when considering $A|W$ as a conditional distribution over $Y$).
%If we assume that $C$ and $R$ depend continuously on $x$ (for $R$ maybe use the space of Borel measure with weak topology or so), then also the integral should (since the inner product is continous -- sum of products). either look this up (maybe this is a special case of Prop 3.13(iiii) in Haim Brezis) or proof it in an analogous way. and then continuity implies measurability.}
%
%\todo{what is a good assumption which is enough and at the same time simple?
%- tie breaking rule?
%- or simply abstractly assume that it is measureable?
%}

\stress{Part 2: show that the remaining conditions of the underlying theorem are satisfied}

\stress{Regarding assumption inference-assistability:}
Let $h_i$ simply be the identity. 
Then, based on Setting \extref{set:agg}, we have
%Assumption \extref{asm:aggr}, we have
\begin{align}
%\bar{U}_i(x, \uact) = 
\uta_i(f_{\Usig_i}(x), \uact_i, \obsmech(x, \uact)) = \uta_i(f_{\Usig_i}(x), \uact_i, h_i(\obsmech(x, \uact)))
\end{align}
and $h_i(\obsmech(x, \uact))$ does not depend on $\uact_i$ since 
\[Y = R_\Mpara( \{ j \in \users : \Uact_j = k \} ) = R_\Mpara( \{ j \in \users : \Uact_j = k, j \neq i \}  ) \]
based on the fact that $R_\mpara$ -- which is the measure induced by the density $r_x$ -- has a density w.r.t.\ the Lebesgue measure.

%	
%	
%	\stress{Regarding assumption assistant-best-responding:}
%	
%	This we assumed explicitly in Setting \extref{set:agg}. %Assumption \extref{asm:aggr}.
%	
%	
%	
%	%Furthermore, it is trivially measurable since its inputs are constant by assumption.
%	%(More specifically: we consider $\dom_w, \dom_\Usig$ as singletons. But any function on a singleton is measurable, the preimage of any set is the singleton itself.)
%	

\stress{Regarding assumption of assistant-separability:}

Since we assumed that $\Usig_i$ is constant for all $i \in \users$, it follows trivially that $Z_{i} \ind \Usig_{i} | \Acov$.

%\todo{finish this}

%
%
%
%Therefoer their composition
%
%
%But this implies that 
%
%
%
%$L^2([0,1])$ is separable.
%
%\maybe{this may even be extendable to stochastic $W, \Usig$ simply by decomposing the complete map  into $(R_x, C_w)$ concatenated with $\int \cdot \cdot$}
%
%
%
%Assume $R_x$ is continuous in $x$
%identity $C \mapsto C$ is continuous. % in $i$.
%
%Then $(f, g) = (R_x, C_i)$ is continuous w.r.t. the product topology on the domain and the codomain, based on the classical result from topology. ??
%
%Furthermore, $e() := \int . . $ is a continuous map on the product space.
%
%But then the concatenation $e(f,g)$ is continuous and thus measurable.
%
%but the (Borel) $\sigma$-algebra generated by the product topology on $\dom_X \times \dom_C$ is the same product $\sigma$-algebra generated by the Borel $\sigma$-algebras on $\dom_X, \dom_C$.

\end{proof}

\subsection{Theorem \extref{thm:ex}}
\label{sec:expr}

%
%
%
%Before proving it, let us restate the result from the paper:
%\todo{... (for arbitrary number of slots)}

The following statement generalizes Theorem \extref{thm:ex} in that here we allow an arbitrary finite number $|K|$ of (time) slots, not just two (i.e., arbitrary finite $K$, not just $K = \{0, 1\}$).

\begin{Theorem}
There exists a self-fulfilling prophecy policy $\pi$ in the assistant-based system $\mlarge$, in Setting~\extref{set:agg} but with an arbitrary finite number $|K|$ of slots.
%	In Setting~\extref{set:agg} \todo{, with the difference that here we allow an arbitrary finite number of slots in $K$, not just two,}, let Assumption~\extref{asm:aggr} hold true.
%	%Then there exists an assistant policy $\pi$ such that prediction accuracy $\lpred_\pi = 0$.
%	Then there exists a self-fulfilling prophecy policy $\pi$ in the assistant-based system $M$.
%In Setting~\extref{set:agg}, let Assumption \extref{asm:aggr} hold true.
%%In the \aggr specified above (in the paper), let Assumptions \extref{asm:aggr} hold true.
%%Then there exists an assistant policy $\pi$ such that prediction accuracy $\lpred_\pi = 0$.
%Then there exists a self-fulfilling prophecy $\pi$.
\end{Theorem}

Here is the proof of this generalized version:

\begin{proof}[Proof for Theorem \extref{thm:ex}]

Keep in mind that in the current large-scale setting, $i \in I$ are interpreted as \emph{types} of users (with the same utility function), not users themselves.
Also keep in mind that the set of slots (i.e., actions available to the users) is $K = \{1, \ldots, |K| \}$.
Since $\Usig$ is constant, here, instead of $\Ua_{i}(\usig, b, k)$, we will write $\Ua_{i}(\yvec,k)$ for utility of type $i$ when choosing slot $k \in \Slots$ given amounts of types $\yvec \in \mathbb{R}^{\newK}$ at slots $1$ to $|K|$. % (represented as of vector in $$).

In what follows, we will consider the space $\mathcal{M}_1$ of Borel probability
measures on the standard $\newK-1$ simplex % $S_{\newK}$. 
$$
S_{\newK}=\left\{z\in \mathbb{R}^{\newK} \middle| z\geq 0,\,\sum_{k=1}^{\newK} z_k=1\right\} .
$$
Furthermore, let, for any $k, l$, $H^k_l$ be the common Heaviside step function (i.e., taking value $0$ upon input below $0$, and value $1$ upon input above $0$), defining it in a special way for the point $0$ (to implement a tie breaking rule that favors lower slots $k \in K$):
\[
H^k_l(0) := \left\{\begin{array}{ll}
1, & \text{ if $k \leq l$,} \\
0, & \text{ if $ k > l$.}
\end{array}
\right.
\]

First we note that for any measure $\mu \in \mathcal{M}_1$
% on the standard $\newK-1$ simplex 
%$$
%\textcolor{red}{S_{\newK}=\left\{z\in \mathbb{R}^{\newK}|z\geq 0,\,\sum_{k=1}^{\newK} z_k=1\right\}}\,,
%$$
the expected
proportion of users choosing slot $k$ conditioned on $\Mpara=\mpara$ with
assistant prediction $A=\mu$ is
\[
F^{\Fcomp}_{\mu}(\mpara):=\int\left[\prod_{l\neq k} H_l^k\left(\int \Diffu_l^k(i,\yvec)d\mu(\yvec)\right)\right]r(i|\mpara)di\,,
\]
with $\Diffu_l^k(i,y)=\Ua_{i}(y,k)-\Ua_{i}(y,l)$. % and $H$ the Heaviside step function (we take the version where $H(0)=1$). 
Given our measurability
assumptions, $$F_{\mu} := (F_\mu^1, \ldots, F_\mu^{\newK})$$ is well defined and measurable, such that the
pushforward measure of $\Mpara$ by it is also a Borel probability measure.

Let $\mu,\,\mu_{1},\,\mu_{2},\,\dots\,\in\mathcal{M}_{1}$,
we say that $\left(\mu_{n}\right)_{n\in\mathbb{N}}$ \emph{converges
	weakly} to $\mu$ if for any $f$ continuous on $S_{\newK}$ 
\[
\int fd\mu_{n}\underset{n\rightarrow+\infty}{\longrightarrow}\int f(y)d\mu\,.
\]

Weak convergence induces the weak topology $\tau$ on $\mathcal{M}_{1}$,
and $\mathcal{M}_{1}$ is compact for this topology (see for
example \cite[Section 13.2]{klenke2013probability}). As a consequence $\mathcal{M}_{1}$
is a non-empty compact convex set of the locally convex topological vector space of bounded signed
measures on $S_{\newK}$.

In order to prove the existence of a fixed point in $\mathcal{M}_{1}$,
according to Leray-Schauder-Tychonoff fixed point theorem \cite[p151]{reed1972methods},
what remains is to prove that the mapping 

\[
\mapp: \left\{\begin{array}{ccc}
\mathcal{M}_{1} & \rightarrow & \mathcal{M}_{1}\\
\mu & \mapsto & P(F_{\mu}(\Mpara)) %\boldsymbol{\Mpara})) %,\,\Mpara\sim P_{\Mpara}
\end{array}
\right.
\]
is continuous for the above defined weak topology. 
%($\mapp$ is just introduced for this proof and not to be confused with the game $\mapp$ of the paper.)

Consider $\mu_{n}\rightarrow\mu$ (for this weak topology). We have
to show that for any $f$ continuous on $S_{\newK}$
\[
\int fd \mapp(\mu_{n})\rightarrow\int fd \mapp(\mu).
\]

We rewrite the left-hand side (using basic change of variable in the
Lebesgue integral), assuming $P_{\Mpara}$ has a bounded density $p_{\Mpara}$
with respect to Lebesgue measure:

\[
\int fd \mapp(\mu_{n})=\mathbb{E}\left[f(F_{\mu_{n}}(\Mpara))\right]=\int f(F_{\mu_{n}}(\mpara))p_{\Mpara}(\mpara)d\mpara .
\]

Since $f$ is also uniformly continuous on this (compact) simplex,
proving uniform convergence on $S_{\newK}$ of $F_{\mu_{n}}$ to $F_{\mu}$
will be enough to conclude. Let us fix an $\epsilon$, we want to
bound $\left\|F_{\mu_{n}}-F_{\mu}\right\|$ by $\epsilon$ uniformly
over $S_{\newK}$. %\textcolor{red}{
Since $S_{\newK}$ is included in a finite dimensional Euclidean space, a uniform bound on each component will be enough to conclude.
%} 
We first note that for any component $k$
\begin{equation}\label{eq:Fdifference}
|F^k_{\mu_{\newk}}(\mpara)-F^k_{\mu}(\mpara)|\leq\int\left|\prod_l H_l^k\left(\int \Diffu_l^k(i,\yvec)d\mu_{n}(\yvec)\right)-\prod_l H_l^k\left(\int \Diffu_l^k(i,\yvec)d\mu(\yvec)\right)\right|r(i|\mpara)di
\end{equation}
%\todo{Btw, where are we actually using the fact that $\Diffu_l^k(i,\yvec$ just depends on $y_k$? Shouldnt this occur somewhere?}\
%\michel{I don't see why it should depend only on $y_k$, but maybe I am missing something. I guess it relates to our discussion on interactions between slots... In the end integrating with respect to the whole $y$ seems to pose no problem...}

%\michel{here I added more details on how to deal with the product: }
We notice we can rewrite the difference of Heaviside products as
\begin{align}
\sum_m \prod_{l<m} H_l^k\left(\int \Diffu_l^k(i,\yvec)d\mu(\yvec)\right) \cdot  \left(H_m^k\left(\int \Diffu_m^k(i,\yvec)d\mu_{n}(\yvec)\right)-H_m^k\left(\int \Diffu_m^k(i,\yvec)d\mu(\yvec)\right)\right) \notag \\
\cdot \prod_{l>m} H_l^k\left(\int \Diffu_l^k(i,\yvec)d\mu_{n}(\yvec)\right) \label{eq:prodDecomp}
\end{align}
such that we can bound the absolute difference of Eq.~(\ref{eq:Fdifference}) using (based on the terms of the product that are not differences being at most 1 anyway)
\begin{equation}\label{eq:proddiff}
|F^k_{\mu_{\newk}}(\mpara)-F^k_{\mu}(\mpara)|\leq \sum_m \int\left| H_m^k\left(\int \Diffu_m^k(i,\yvec)d\mu_{n}(\yvec)\right)- H_m^k\left(\int \Diffu_m^k(i,\yvec)d\mu(\yvec)\right)\right| r(i|\mpara)di\,.
\end{equation}
We will thus focus first on bounding an arbitrary term 
\begin{equation}\label{eq:proddiffoneterm}
\int \left|H_m^k\left(\int \Diffu_m^k(i,\yvec)d\mu_{n}(\yvec)\right)-H_m^k\left(\int \Diffu_m^k(i,\yvec)d\mu(\yvec)\right)\right|r(i|x)di\,,
\end{equation}
 dropping the indices $m$ and $k$ to ease notations in the following paragraph.

Our assumptions (Setting \extref{set:agg}) imply any $\Diffu_m^k$ is a polynomial in $(i,y)$, that can be written $\sum_{m=0}^d i^m q_m(y)$. Then integrating
the quantity inside each $H$ yields polynomials in $i$, $p(i)$ and $p_n(i)$, of maximum order $d$, whose coefficients are a linear combination
of the moments of $\mu$ and $\mu_{n}$ respectively, up to some order
$d'$. Convergence of $\mu_{n}$ to $\mu$ thus guarantees convergence
of the coefficients of $\diffu_{n}$ to those of $\diffu$, and uniform convergence
of the $\diffu_{n}$ to $\diffu$ on the unit interval.

%\todo{So in this and the following paragraph, we focus on one single term of the product? $\prod_l \left( H\left(\int \Diffu_l^k(i,\yvec)d\mu_{n}(\yvec)\right)- H\left(\int \Diffu_l^k(i,\yvec)d\mu(\yvec)\right) \right)$}	\michel{yes: see added decomposition above}
The discontinuity of $H$ does not allow us to further use uniform
continuity to bound the term of Eq.~(\ref{eq:proddiffoneterm}), but we notice that the absolute difference between the
two $H$ terms is either zero or one, the later occurring only when the signs of $\diffu_{n}$ an $\diffu$ differ. Using the assumption that there exists at least one $m$ such that $q_m(y)$ is constant and non-zero implies $\diffu$ is a non-zero polynomial. There is then only two possible cases to consider:
\begin{itemize}
	\item If $\diffu$ has no root on the unit interval (e.g. $\diffu$ is constant), then uniform convergence guaranties we can choose $N$ large enough such that $\diffu$ and $\diffu_{n}$ have the same sign on the unit interval, implying the difference in Heaviside function is zero on the whole interval and the corresponding term of Eq.~(\ref{eq:proddiffoneterm}) can be ignored. 
	\item Alternatively, $\diffu$ has a finite number of roots, and for $\newk$ large enough the difference inside the integral in Eq.~(\ref{eq:proddiffoneterm}) can be non-zero only on a finite
	number of intervals surrounding these roots, where the sign of the $\diffu$ and $\diffu_{\newk}$ may differ. We will thus focus on this case and show the length of these intervals can be bounded.
\end{itemize} 
Let $\left(i_{l}\right)$ be all the (finite) collection of roots
of $\diffu$ on the unit interval, then there exists a $\eta_{0}$
thus that $\diffu$ is strictly monotonous in all right and left $\eta_{0}$-neighborhoods
of each $i_{l}$ (one-sided neighborhoods are needed for roots with even multiplicity), and thus admits a family of one-sided monotonous continuous local inverse
functions $\left(\kappa_{l}^{+},\,\kappa_{l}^{-}\right)$, up to a change in sign, such that for each $l$,
\[
\kappa_{l}^{+}(|\diffu(i)|)=i-i_{l},\quad i-i_{l}\in[0,\,\eta_{0}]
\]
and
\[
\kappa_{l}^{-}(|\diffu(i)|)=i-i_{l},\quad i-i_{l}\in[-\eta_{0},\,0]
\]
(note continuity of the inverse is guaranteed by continuity and strict
monotonicity, while the implicit function theorem does not directly
apply at multiple roots due to vanishing of the derivative). Let $\epsilon_{0}'$
be the maximum radius such that the interval $[0,\,\epsilon_{0}']$
is included in the intersection of the domains of all $\kappa_{l}^{-}$ and $\kappa_{l}^{+}$. We additionally choose $\epsilon_0\leq \epsilon_0'$ such that $|p(i)|>\epsilon_0$ for any $i$ outside the union of intervals  $[i_l+\kappa_{l}^{-}(\epsilon_0),\,i_l+\kappa_{l}^{+}(\epsilon_0)]$ associated to each root $l$ (this can be done by picking the minimum between $\epsilon_0'$ and the lower bound of $|p|$ outside of the neighborhoods of each root).
Let us choose $N_{0}$ such that for $n>N_0$, $|\diffu-\diffu_{n}|<\epsilon_{0}$ uniformly
on the unit interval. Then the Lebesgue measure ($\lambda(I_{\epsilon_{0}}^n)$) of $I_{\epsilon_{0}}^{n}$,  
the union of all intervals such that
\[
\left|H_m^k\left(\int \Diffu_m^k(i,\yvec) d\mu_{\newk}(\yvec)\right)-H_m^k\left(\int \Diffu_m^k(i,\yvec) d\mu(\yvec)\right)\right|=\left|H_m^k\left(\diffu_{n}(i)\right)-H_m^k\left(\diffu(i)\right)\right|>0\,(=1)\,,
\]
is inferior to $\sum_{l}\left|\kappa_{l}^{+}(\epsilon_{0})-\kappa_{l}^{-}(-\epsilon_{0})\right|$. This is because outside of $I_{\epsilon_{0}}^{n}$, $p$ is at least as far away from $0$ as $\epsilon_0$, so $p_n$ has he same sign as $p$.
By (uniform) continuity of all $\kappa_{l}^{+}$ and $\kappa_{l}^{-}$, for $\eta_{1}$ arbitrary small, we can choose $\epsilon_{1}<\epsilon_{0}$
and $N_{1}>N_{0}$ such that for $n>N_1$,  $\lambda(I_{\epsilon_{1}}^n)<\eta_{1}$, such that, since $r(i|\mpara)$ is continuous (and thus bounded) on $[0,\,1]^{2}$, we get
\begin{equation}
\int\left| H_m^k\left(\int \Diffu_m^k(i,\yvec)d\mu_{n}(\yvec)\right)- H_m^k\left(\int \Diffu_m^k(i,\yvec)d\mu(\yvec)\right)\right| r(i|\mpara)di\leq \eta_1 \max (r)\,.
\end{equation}

As this procedure can be done for all $m\neq k$,
, we can bound the $k$-th component of $F_\mu$ using 
\[
|F^\Fcomp_{\mu_{n}}(\mpara)-F^\Fcomp_{\mu}(\mpara)|\leq (|K|-1)\eta_{1}\max(r)\,,
\]
for $\eta_{1}$ arbitrary small. We thus get a uniform bound for $\|F_{\mu_{n}}(\mpara)-F_{\mu}(\mpara)\|$, which is enough to ensure that $\mapp$ is continuous for the weak topology. 

This
implies the existence of a fixed point of $\mapp$ in $\mathcal{M}_{1}$
according to the Leray-Schauder-Tychonoff fixed point theorem \cite[p151]{reed1972methods}.

\end{proof}

\section{Proofs and extensions for Section \extref{sec:alg}}

\subsection{Extended version of the proposition and proof for Section \extref{sec:expodamp}}
%Proofs for Propositions \extref{prop:opt_ctl} and \extref{prop:stability}}
\label{sec:pr_stab}

Let us state and proof a proposition that is a slight generalization of Proposition \extref{prop:expod}.

\begin{SuppProposition}[Optimality and Convergence Rate of Expodamp -- Generalized Version of Proposition \extref{prop:expod}]
\label{supp:prop:expodfull}
%\todo{maybe shorten this whole section by completely dropping observation noise}
In the dynamic \aggr (Section \extref{sec:defdyn}), let Assumption \extref{asm:exposm} hold true.
%Let $\pi$ be defined by 
%\begin{align}
%A_\pi^0 &:= \gamma (1-\beta)^{-1} \E(\tth^0),\\
%A_\pi^{t+1} &:= A_\pi^t +  \gamma (1-\beta)^{-1} Q_t  ( Y^t - A_\pi^t ) , %\text{  for all $t\geq 0$,}
%\end{align}
%for all $t\geq 0$
%with $Q_t$ some function of the covariance structure as detailed in the proof of this proposition. 
%\spar
%In particular, $Q_t$ is such that, if there is no observation noise in the latent-state model, i.e., $\var(E_{Y}^0) = 0$, then $\pi$ coincides with Expodamp (Algorithm \ref{alg:expodamp}) when setting $\alpha:=(1-\beta)^{-1}$. \spar
\begin{citem}
	\item \emph{Stochastic case:} 
	%Let the assistant's policy $\pi$ be Expodamp (Algorithm \ref{alg:expodamp}) with parameter $\alpha:=(1-\beta)^{-1}$ for the true $\beta$ of Eq.~\ref{eqn:gamma}.
	%Assume $\var(E_{Y}^0) = 0$. 
	Let the assistant's policy  $\pi$ be defined by 
	\begin{align}
	A_\pi^0 &:= \gamma (1-\beta)^{-1} \E(\tth^0), \label{supp:eqn:def0}\\
	A_\pi^{t+1} &:= A_\pi^t +  \gamma (1-\beta)^{-1} Q_t  ( Y^t - A_\pi^t ) , \label{supp:eqn:def1}  %\text{  for all $t\geq 0$,} 
	\end{align}
	for all $t\geq 0$
	with $Q_t$ some function of the covariance structure as detailed in Eq.~\ref{supp:eqn:Q} of the proof of this proposition. 
	\spar
	In particular, $Q_t$ is such that, if there is no observation noise in the latent-state model, i.e., $\var(E_{Y}^t) = 0$, then $\pi$ coincides with Expodamp (Algorithm \extref{alg:expodamp}) when setting $\alpha:=(1-\beta)^{-1}$ for the true $\beta$ of Eq.~\extref{eqn:gamma}. \spar
	Assume $\E(E_{Y}^t) = 0, t \geq 1$.
	Then%
	\footnote{One can also make the more general statements about $\pi$ minimizing the \emph{cumulative (over time)} loss.
		To see that the ``local'' statement (for individual $t$) implies more global statements observe two things: First, $A^t$ influences only $Y^t$ but no future $Y^{t'}, t' > t$, and so term-wise optimization coincides with cumulative optimization. Second, $\pi$ as defined above does the optimal thing at stage $t$ regardless of how $A^{t'}, t' < t$ was picked, in case we feed what $\pi$ would have outputted at stage $t-1$, instead of the actual $A^{t-1}$, into $\pi$ at stage $t$.}%
	, at each stage $t$, $\lpointpredt = 0$ and
	\begin{align}
	A^{t} {=} \arg\min_{a'} \E((A^{t} {-} Y^{t})^2|\doc{A^{t}{=}a'}, A^{0:t-1}\!, Y^{0:t-1}).  \label{supp:eqn:sa}
	\end{align}
	%\todo{not sure how to exactly formulate this ... conditional? policy based? in sense its a POMDP problem, so how do they define it? seems like they condition the value function on the belief. see \url{https://en.wikipedia.org/wiki/Partially_observable_Markov_decision_process} or barto}
	%\begin{align}
	%\pi {=} \arg\min_{\pi'} \E_{\pi'}((A^{t} {-} Y^{t})^2).  \label{supp:eqn:sa}
	%\end{align}
	%
	%\begin{align}
	%\pi {=} \arg\min_{\pi'} \E_{\pi'}((A^{t} {-} Y^{t})^2|A^{0:t-1}, Y^{0:t-1}).  \label{supp:eqn:sa}
	%\end{align}
	%
	%\begin{align}
	%A_\pi^{t} {=} \arg\min_{a'} \E((A^{t} {-} Y^{t})^2|A^{t}{=}a', A^{0:t-1}, Y^{0:t-1}).  \label{supp:eqn:sa}
	%\end{align}
	%
	%
	%
	%\begin{align}
	%A^{t} {=} \arg\min_{a'} \E((A^{t} {-} Y^{t})^2|A^{t}{=}a', A^{0:t-1}\!, Y^{0:t-1}).  \label{supp:eqn:sa}
	%\end{align}
	%and, in particular, 
	%.
	\item \emph{Deterministic case:}
	Let the assistant's policy $\pi$ be Expodamp (Algorithm \extref{alg:expodamp}).
	Assume that $\tth^t = \bth$ is constant for $t \geq 0$, that $\beta = (1-\gamma)$ and that $E_{\tth}^t = E_{Y}^t = 0$.
	Then 
	\begin{align*}
	Y^{t} &= \bth +(1-\gamma)(A^0-\bth)(1-\alpha\gamma)^t , \text{  for all $t\geq 0$}. %\\
	%\todo{L^{t}} &= 
	\end{align*}
	That is, $Y^{t}$ converges exponentially with rate $\gamma\alpha$ towards the ``optimum''/fixed point $\bth$ (and thus also $A^t$ converges to $x$ based on Expodamp's formula) iff
	$0<\gamma\alpha<2$.
\end{citem}
\end{SuppProposition}

\begin{proof}[Proof of Proposition \ref{supp:prop:expodfull}]

\parag{First part of the proposition -- stochastic case:}

\todo{improve writing towards end by explcityly referring to eq47? what do we want to show here Lpartpred and the other objective (but why)... (3)} 

\todo{Problems with referencing: Eq. \ref{supp:eqn:ff1}}

\emph{\textbf{Prerequisites:}}

Consider the complete dynamical system, consisting of Assumption \extref{asm:exposm}, the state-space model (without assistant's behavior), together with Eq.~\extref{eqn:ye}, the assistant's behavior under policy $\pi$.
For this model let, for $Z \in \{Y, \tth\}$ and $t\geq t'$,
\begin{align}
Z^{t|t'} &:= \E_\pi( Z^t | Y^{0:t'}, A^{0:t'} ) \label{supp:eqn:eisbaer}\\
\Sigma^{t|t'}_{Z} &:= \var_\pi( Z^t | Y^{0:t'}, A^{0:t'} ) .
\end{align}
To be as explicit as possible, note that, for $Z \in \{Y, \tth\}$ and $t\geq t'$ (due to the causal structure)
\begin{align*}
\E_\pi( Z^t | Y^{0:t'}, A^{0:t'}=a^{0:t'} ) \left( = \E_\pi( Z^t | Y^{0:t'}, \doc{A^{0:t'}=a^{0:t'}}) \right) = \E( Z^t | Y^{0:t'}, a^{0:t'}) , \label{supp:eqn:pom}
\end{align*}
where the latter expectation is taken in the POMDP model of Assumption \extref{asm:exposm} when setting $A^t$ to constants $a^t$, for $t\geq 0$, and \emph{not plugging in any assistant policy}.
The analogous holds for $\var_\pi( Z^t | Y^{0:t'}, A^{0:t'} )$.
So we can use the classical Kalman filter recursive equations \citep[Section 18.3.1]{Luetkepohl2006}, which hold for the POMDP model, and thus, based on Eq.~\ref{supp:eqn:pom}, also for $Z^{t|t'}, \Sigma_Z^{t|t'}$, for $Z \in \{Y, \tth\}$ and $t\geq t'$ defined in the complete dynamical system \emph{including} assistant policy.
Specifically, the relevant equations are as follows:
\begin{align}
\hat{\tth}^{t+1|t} 
&= \hat{\tth}^{t|t} \\
&= \hat{\tth}^{t|t-1} + Q_t (Y^t - \hat{Y}^{t|t-1})\\
&= \hat{\tth}^{t|t-1} + Q_t  ( Y^t - \gamma \hat{\tth}^{t|t-1} - \beta A^t) ,  \label{supp:eqn:ef1}
\end{align}
for 
\begin{align}
Q_t &:= \gamma \Sigma_{\tth}^{t|t-1} (\Sigma_Y^{t|t-1})^{-1} \label{supp:eqn:Q} \\
&= \gamma \Sigma_{\tth}^{t|t-1}  (\gamma^2 \Sigma_{\tth}^{t|t-1} + \var(E_{Y}^0) )^{-1} . \label{supp:eqn:af1}
\end{align}
Note that $Q_t$ does not depend on the $A^t, t \in \N$, when considering $A^t, t \in \N$ as a parameter. %\todo{is this obvious enough?}
%
%Consider the equations of Assumption \extref{asm:exposm}, and consider $A^t, t \in \N$ in these equations as exogenous inputs, i.e., replace them by constants $a^t, t \in \N$.
%
%For $Z \in \{Y, \tth\}$ and $t\geq t'$, let
%\begin{align}
%Z^{t|t'}_{a^{0:t'}} &:= \E( Z^t | Y^{0:t'}, a^{0:t'} ) \label{supp:eqn:eisbaer}\\
%\Sigma^{t|t'}_{Z, a^{0:t'}} &:= \var( Z^t | Y^{0:t'}, a^{0:t'} ) .
%\end{align}
%Keep in mind that Assumption \extref{asm:exposm} (with $A^t = a^t, t \in \N$) %(when considering $A^t, t \in \N$ as exogenous inputs) 
%is a classical state-space model under which the classical Kalman filter recursive equations \citep[Section 18.3.1]{Luetkepohl2006} give formulas for $Z^{t|t'}, \Sigma_Z^{t|t'}$.
%More specifically, these formulas are (building on \citep[Section 18.3.1]{Luetkepohl2006})
%%Based on the Kalman filter recursive equations \citep[Section 18.3.1]{Luetkepohl2006}, we have
%\begin{align}
%\hat{\tth}^{t+1|t} 
%&= \hat{\tth}^{t|t} \\
%&= \hat{\tth}^{t|t-1} + Q_t (Y^t - \hat{Y}^{t|t-1})\\
%&= \hat{\tth}^{t|t-1} + Q_t  ( Y^t - \gamma \hat{\tth}^{t|t-1} - \beta a^t) ,  \label{supp:eqn:ef1}
%\end{align}
%for 
%\begin{align}
%Q_t &:= \gamma \Sigma_{\tth}^{t|t-1} (\Sigma_Y^{t|t-1})^{-1} \label{supp:eqn:Q} \\
%&= \gamma \Sigma_{\tth}^{t|t-1}  (\gamma^2 \Sigma_{\tth}^{t|t-1} + \var(E_{Y}^0) )^{-1} . \label{supp:eqn:af1}
%\end{align}
%Note that $Q_t$ does not depend on the (parameter sequence) $a^t, t \in \N$. \todo{is this obvious enough?}

\emph{\textbf{Showing Eq.~\ref{supp:eqn:sa}:}}

Assume the conditions of the proposition, i.e., that the assistant's policy $\pi$ is defined by Eq.~\ref{supp:eqn:def0} and \ref{supp:eqn:def1} (for convenience we may drop the subscript $\pi$ of $A^{t}_\pi$ in what follows),
with $Q_t$ from Eq.~\ref{supp:eqn:Q}. % (note that $Q_t$ is defined does not depend on ),
and Assumption \extref{asm:exposm}.

%Now, as assumed in the proposition, let the assistant's policy $\pi$ be defined by (for convenience we drop the the subscript $\pi$ here)
%\begin{align}
%A^0 &:= \gamma (1-\beta)^{-1} \E(\tth^0),\\
%A^{t+1} &:= A^t +  \gamma (1-\beta)^{-1} Q_t  ( Y^t - A^t ), \text{ for all $t\geq 0$} %\label{supp
%\end{align}
%%for all $t\geq 0$
%with $Q_t$ from Eq.~\ref{supp:eqn:Q} (note that $Q_t$ is defined does not depend on ),
%and assume the equattions of Assumption \extref{asm:exposm}.

Let us show via induction that
\begin{align}
%A^0 &= \gamma (1-\beta)^{-1} \E(\tth^0) , \\
A^{t} &= \gamma (1-\beta)^{-1} \hat{\tth}^{t|t-1}, \text{ for all $t\geq 0$} , \label{supp:eqn:tf}
\end{align}
where we let $\hat{\tth}^{0|1} := \E(\tth^0)$.

\emph{Base case:} For $t=0$, the statement holds by definition.

\emph{Induction step:}
Assume the statement holds for $t$.
Then we have %, on the one hand, observe that this $\pi$ coincides with how we defined it in the statement of the proposition, since we have (based on Eq. \ref{supp:eqn:ef1})
\begin{align}
A^{t+1} 
&= A^t +  \gamma (1-\beta)^{-1} Q_t  \left( Y^t - A^t \right) \\
&=\gamma (1-\beta)^{-1} \left( \gamma^{-1} (1-\beta) A^t + Q_t  ( Y^t -  A^t) \right) \\
&=\gamma (1-\beta)^{-1} \left( \gamma^{-1} (1-\beta) A^t + Q_t  ( Y^t - \gamma \gamma^{-1} (1-\beta) A^t  - \beta A^t) \right) \\
&=\gamma (1-\beta)^{-1} \left( \hat{\tth}^{t|t-1} + Q_t  ( Y^t - \gamma \hat{\tth}^{t|t-1}  - \beta A^t) \right) \label{supp:eqn:a1} \\
&= \gamma (1-\beta)^{-1} \hat{\tth}^{t+1|t}, \label{supp:eqn:a2}
%&= \gamma (1-\beta)^{-1} \hat{\tth}^{t|t-1} +  \gamma Q_t  \left(  (1-\beta)^{-1} Y^t - (1-\beta)^{-1} \gamma \hat{\tth}^{t|t-1} \right) \\
%&= (1 - \gamma Q_t) A^t + \gamma Q_t(1-\beta)^{-1} Y^t .
\end{align}
where 
Eq.~\ref{supp:eqn:a1} is due to the inductive assumption, and
Eq.~\ref{supp:eqn:a2} is based on Eq. \ref{supp:eqn:ef1}.
This completes the induction for Eq.~\ref{supp:eqn:tf}.

Now observe that the statement we need to show, Eq.~\ref{supp:eqn:sa}, is equivalent to
\begin{align}
0 &= \frac{d}{d a} \E((A^{t} - Y^{t})^2|\doc{A^{t}=a}, A^{0:t-1}, Y^{0:t-1}) \\
&= \frac{d}{d a} \E((A^{t} - \beta A^{t} - \gamma \tth^{t} - E_{Y}^{t})^2|\doc{A^{t}=a}, A^{0:t-1}, Y^{0:t-1}) \\
&= \frac{d}{d a} \E( ((1 - \beta) A^{t} - \gamma \tth^{t} - E_{Y}^{t})^2|\doc{A^{t}=a}, A^{0:t-1}, Y^{0:t-1}) \\
&= \frac{d}{d a} \E( ((1 - \beta) a - \gamma \tth^{t} - E_{Y}^{t})^2  | A^{0:t-1}, Y^{0:t-1}) \\
&=  \E( \frac{d}{d a} ((1 - \beta) a - \gamma \tth^{t} - E_{Y}^{t})^2   |A^{0:t-1}, Y^{0:t-1}) \\
&=  \E( 2 ((1 - \beta) a - \gamma \tth^{t} - E_{Y}^{t}) (1-\beta)   |A^{0:t-1}, Y^{0:t-1}) \\
&= 2 (1-\beta) \E(  (1 - \beta) a - \gamma \tth^{t} - E_{Y}^{t}   |A^{0:t-1}, Y^{0:t-1}) ,\label{supp:eqn:si} %\\
%&= \E(A^{t+1} - Y^{t+1}|A_{t+1}=a, A^{0:t}=a^{0:t}, Y^{0:t}=y^{0:t}) \\
%&= \E( (1-\beta) a + \gamma \tth^{t+1} + E_{Y}^{t+1} | A^{0:t}=a^{0:t}, Y^{0:t}=y^{0:t})
\end{align}
%\todo{not sure if this holds .... rewrite it starting from the explicit formulat below squared and then see if the square actually drops at all like in the MSE case}
which in turn is equivalent to
\begin{align}
a &= (1-\beta)^{-1} \left( \E(\gamma \tth^{t} | A^{0:t-1}, Y^{0:t-1}) + \E(E_{Y}^{t}) \right) 
%&= \gamma (1-\beta)^{-1} \E(\tth^{t+1} | A^{0:t}=a^{0:t}, Y^{0:t}=y^{0:t}) . \label{supp:eqn:ff1}
\end{align}
which in turn is equivalent to (based on our assumption $\E(E_{Y}^t) = 0, t \geq 1$)
\begin{align}
%a &= (1-\beta)^{-1} \left( \E(\gamma \tth^{t+1} | A^{0:t}=a^{0:t}, Y^{0:t}=y^{0:t}) + \E(E_{Y}^{t+1}) \right) \\
a &= \gamma (1-\beta)^{-1} \E(\tth^{t+1} | A^{0:t-1}, Y^{0:t-1}) , \label{supp:eqn:ff4}
\end{align}
which in turn is equivalent to (simply plugging in the definition in Eq.~\ref{supp:eqn:eisbaer})
\begin{align}
%a &= (1-\beta)^{-1} \left( \E(\gamma \tth^{t+1} | A^{0:t}=a^{0:t}, Y^{0:t}=y^{0:t}) + \E(E_{Y}^{t+1}) \right) \\
a &= \gamma (1-\beta)^{-1} \hat{\tth}^{t|t-1} . \label{supp:eqn:ff41}
\end{align}
Since we know, based on Eq.~\ref{supp:eqn:tf}, that $A^t$ under $\pi$ satisfies Eq.~\ref{supp:eqn:ff41} when plugging it in for $a$, based on the chain of equivalences above, we also know that it satisfies Eq.~\ref{supp:eqn:sa}, which is what needed to be shown.

The statement that Expodamp (Algorithm \extref{alg:expodamp}) is a special case of the assistant policy defined in Eq.~\ref{supp:eqn:def0} and \ref{supp:eqn:def1} can easily be seen as follows:
If there is no observation noise in the latent-state model, i.e., $\var(E_{Y}^t) = 0$, then Eq. \ref{supp:eqn:af1} implies that $Q_t = \gamma^{-1}$.
Hence, when setting $\alpha = (1-\beta)^{-1}$, we have
\begin{align}
A^{t+1} &= A^t + (1-\beta)^{-1} \left( Y^t - A^t \right) = A^t + \alpha \left( Y^t - A^t \right) ,
\end{align}
i.e., we get Expodamp (Algorithm \extref{alg:expodamp}) as special case.
\emph{\textbf{Showing that $\lpointpredt_\pi=0$:}}

To also show the first statement, $\lpointpredt_\pi=0$, observe that this means
\begin{align}
0  &= \E \left( ( A^{t} - \E( Y^{t} |A^{t}, A^{0:t-1}, Y^{0:t-1} ) )^2 \right)
\end{align}
which is equivalent to 
\begin{align}
A^{t} &= \E( Y^{t} |A^{t}, A^{0:t-1}, Y^{0:t-1})
\end{align}
almost everywhere.
This in turn is equivalent to
\begin{align*}
a &= \E( Y^{t} |A^{t}=a, A^{0:t-1}, Y^{0:t-1})
\end{align*}
for all $a$,
which is equivalent to
\begin{align*}
0 &= a - \E( Y^{t} |A^{t}=a, A^{0:t-1}, Y^{0:t-1}) \\
&= \E( a - Y^{t} |A^{t}=a, A^{0:t-1}, Y^{0:t-1}) \\
&= \E( a - \beta a - \gamma \tth^{t} - E_{Y}^{t} | A^{0:t-1}, Y^{0:t-1}) \\
&= \E( (1-\beta) a - \gamma \tth^{t} - E_{Y}^{t} | A^{0:t-1}, Y^{0:t-1}) .
\end{align*}
This is equivalent to Eq.~\ref{supp:eqn:si}, which was equivalent to Eq.~\ref{supp:eqn:ff41}, which, as stated above, is satisfied by $\pi$. % as defined above.

%and therefore equivalent to Eq.~\ref{supp:eqn:sa} and hence satisfied by $\pi$ as defined above.

%iff
%\begin{align}
%&\E(Y^{t+1} - A^{t+1} | A_{t+1}=a, A^{0:t}, Y^{0:t}) \\
%&= \E( (1-\beta) a + \gamma \tth^{t+1} + E_{Y}^{t+1} | A^{0:t}, Y^{0:t})
%\end{align}
%
%
%As usual, we take the derivative w.r.t. $a'$ and commute with the expectation.
%Based on this, we have $A^{t+1} \in \arg\min_{a'} \E((A^{t+1} - Y^{t+1})^2|A_{t+1}=a', A^{0:t}, Y^{0:t})$ iff $A^{t+1} = \E(Y^{t+1}|A_{t+1}, A^{0:t}, Y^{0:t})$

%
%$L_{point}$
%
%$L_{mean}$
%
%
%$L_{prob}$
%
%
%Mean-SFP: $\E(Y|A=\mu, past) = \mu(past)$
%
%Varince-ignoring fixed-point
%
%\todo{generally, also in the opt-BNE-case, the key notion seems to be \emph{conditionally (on past obs/act)} SFP}

%\end{proof}
%
%
%
%
%
%
%\begin{proof}{Proof of Proposition \ref{prop:stability}}

\parag{Second part of the proposition -- deterministic case:}

\newcommand{\Xlike}{X}
\newcommand{\xlike}{x}
\newcommand{\barxlike}{x}

In this proof (and only here) let us, for simplicity, use the following notation:
\begin{itemize}
	\item $Y^t, t \in \N$ denotes a sample path (instead of a random process),
	\item and $Y$ denotes $(Y^t)_{t \in \N}$.
\end{itemize}
%Furthermore, we write $\tilde{\Usig}^t$ instead of $X^t$ and $\bar{\usig}$ instead of $\bar{x}$.

Let $\widetilde{Y}$ denote the one-sided Z-transform of $Y$ \cite{proakis1996}, defined as the Laurent series (considered formally without considerations on the domain of convergence)
\begin{align*}
\widetilde{Y}(z) = \sum_{k=0}^{+\infty} Y^k z^{-k}, z\in \mathbb{C} ,
\end{align*}
and similarly for $X, A$.
The assumed dynamics equation %Equation \ref{supp:eqn:gamma} 
expressed in the Z-domain leads to
\begin{align*}
\widetilde{Y}(z) &= (1-\gamma)\widetilde{A}(z)+\gamma \widetilde{\Xlike}(z)\,. \\
\end{align*}
The equation that defines Expodamp %Equation \ref{supp:eqn:alpha} 
implies (using the time-shifting formula \cite[p. 208]{proakis1996}) 
\begin{align*}
z\left(\widetilde{A}(z)-A^0\right) &= (1-\alpha)\widetilde{A}+\alpha \widetilde{Y}\,. \\
\end{align*}

Combining the above equations results in the following expression for $Y$ in the Z-domain:
\begin{align*}
\widetilde{Y}(z) &= \frac{(\gamma-1)\barxlike}{1-z^{-1}(1-\alpha\gamma)}+\frac{\barxlike}{1-z^{-1}}+\frac{(1-\gamma)A^0}{1-z^{-1}(1-\alpha\gamma)}\,.\\
\end{align*}

By classical inversion formulas of the Z-transform, we finally get
\begin{align*}
Y^{t} &= \barxlike +(1-\gamma)\left[A^0-\barxlike\right](1-\alpha\gamma)^t   \,,t\geq0\,. \\
\end{align*}
which shows the exponential convergence for any $(A^0,\,\barxlike)$ under the condition $0<\gamma\alpha<2$

%Being slightly sloppy with notation, here 
%(1) $Y^t, t \in (-\infty, \infty)$ denotes one a sample path (instead of a random process),
%(2) $Y$ denotes $(Y^t)_{t \in (-\infty, \infty)}$.
%
%Let $\tilde{Y}$ denote the Z-transform of $Y$ \cite{proakis1996}.
%
%Equations \ref{supp:eqn:alpha} and \ref{supp:eqn:gamma} imply the following transfer function for the mapping from from $\ztt$ to $\zy$:
%\begin{align*}
%h(z) = \frac{\gamma (z + \gamma-1)}{z-1+\alpha \gamma} .
%\end{align*}
%
%
%\begin{align*}
%\widetilde{(Y^{t-k})_{t}} &= z^{-k} ( \zy + \sum_{j=1}^k Y^{-j}z^{j}) \\
%\end{align*}

\end{proof}

\subsection{Full algorithms, proposition and proof for Section \extref{sec:partpred}}
\label{sec:pr_partpred}

In this section, let us give the ``equilibrium selection objective'' (Section \extref{sec:pre}) a formal loss function:
\begin{align}
\lnet_\pi := \left \{ \begin{array}{ll} 0, & \text{  if $s_\pi$ %(as a mapping from $\Usig_i^t, \Acov^t$ to $\Uact$) 
is a BNE of $G$ (w.r.t.\ the variables of the $t$-th stage) } \\ 1, & \text{ else. } \end{array} \right.
\end{align} 
Furthermore, let NE stand for (complete-information) Nash equilibrium.

\subsubsection{Detailed algorithms}

Consider Algorithm \ref{supp:alg:partpredfull} together with Algorithms \ref{supp:alg:consistupdate} and \ref{supp:alg:updategeneral}, respectively, as subroutines.
It is a rigorous  version of Algorithm \extref{alg:partpredsk} that is also more general in that it allows $V$ to vary.

%\begin{wrapfigure}[25]{t!}{.5\textwidth}
%\begin{minipage}{1\textwidth}
%\vspace{0pt}  
\floatname{algorithm}{}
\begin{algorithm} %[H] %[H] %[tb]
% for custom keywords etc., see 9.5.1 in algorithm2e.pdf
\caption{$\an{Partpred}$}
\label{supp:alg:integrated}
\label{supp:alg:partpred}
\label{supp:alg:partpredfull}
\KwIn{parameters: $\bar{A} = ( \bar{A}_{\acov} )_{\acov \in \dom_\Acov}$, $r$, $\an{UpdateFunction}$}
For each $\acov \in \dom_\Acov$, initialize $a_\acov \in \bar{A}_{\acov}$ randomly\\
\For{each stage $t\geq 0$}{
	\KwIn{$\acov:= \Acov^{t}$}
	%Let $I := \{ i : \text{ player $i$ has best-responded to $a_w$ $J$ times} \}$\\
	\eIf{$a_\acov$ has been announced less than $r$ times or has converged under $\Acov=\acov$}{
		\KwOut{$A^t := a_\acov$}
		%\KwIn{$C^t$}
	}{
		Let $\hat{P}_{\acov, a_\acov}^r$ be the empirical distribution of $\Uact$ in the $r$ times it was sampled under $\Acov=\acov, A=a_\acov$\\
		Let $a_\acov' := \arg\min_{a_\acov'' \in \bar{A}_{\acov}} \| a_\acov'' - \hat{P}_{\acov, a_\acov}^r \|$\\
		%Let $E := \emptyset$\\
		%\For{$i \in I$ (in random order)}{
		%Let $\tilde{P}_{w, i}$ be the empirical distribution of $C_i$ in these $J$ times\\
		%\If{$\vn{Consist}(E, \{\tilde{P}_{w, j}\}_{j \in E}, i, \tilde{P}_{w, i})$}{
		%$E := E \cup \{i\}$ REPLACE THIS WHOLE THING BY COnist directly returing consistent set
		%}
		%}
		Let $a_\acov'' := \an{UpdateFunction}(a_\acov, a_\acov', \acov)$\\ %\tcp{defined in supplement}
		%Let $a'_w = a_w$\\
		%\For{$i \in E$}{
		%Let $[a'_w]_i := [c_w]_i$
		%}
		\uIf{$a_\acov'' = a_\acov$}{
			Remember that for $\Acov=\acov$, convergence happened %\\
			%\KwOut{$A^t := a_w''$}
		}
		\uElseIf{$a''_\acov$ and all other $a_\acov \in \bar{A}_\acov$ have been tried $r$ times \label{line:forces}}{
			Set $a_\acov'' := \arg\min_{a_\acov'} \| a_\acov' - \hat{P}_{\acov, a_\acov'}^r \|$\\
			Remember that for $\Acov=\acov$, convergence happened \label{line:forcee}
			%Let $a_w \in \times_{i \in N} \bar{P}_{w, i}$ be s.t. it has not been announced before, or, if everything has been tried, stick with ...\\
		}
		\uElseIf{$a''_\acov$ has been tried $r$ times}{
			Pick unused $a_\acov'' \in \bar{A}_\acov$ at random \label{code:rands}
		} 
		Set $a_\acov = a''_\acov$ \label{code:rande}\\
		\KwOut{$A^t := a_\acov$}
	}
}
\end{algorithm}
%\end{minipage}
%\end{wrapfigure}

%\begin{wrapfigure}[15]{t!}{0.5\textwidth}
%\begin{minipage}{1\textwidth}
%\vspace{0pt}  
\begin{algorithm}%[H] %[H] %[tb]
% for custom keywords etc., see 9.5.1 in algorithm2e.pdf
\caption{\an{UpdateFunctionCongestion}}
\label{supp:alg:consistupdate}
%\KwSty{test command block} 
\tcp{For simplicity, consider $a_\acov, a_\acov'$ as action profiles in $\dom_{\Uact}$ instead of (Dirac) distributions over action profiles}
%Let $E := \emptyset$\\
%\For{$i \in I$ (in random order)}{
%Let $\tilde{P}_{w, i}$ be the empirical distribution of $C_i$ in these $J$ times\\
%\If{$\vn{Consist}(E, \{\tilde{P}_{w, j}\}_{j \in E}, i, \tilde{P}_{w, i})$}{
%$E := E \cup \{i\}$
%}
%}

\KwIn{$a_\acov, a_\acov', \acov$}
Let $a_\acov'' := a_\acov$ and $J := \emptyset$\\
\While{There is $i \in \users$ s.t. $[a_\acov]_i \neq [a_\acov]_j$ and $[a_\acov']_i \neq [a_\acov']_j$ for all $j \in J$}{
	%\For{$i \in N$ (in random order)}{
	%\If{$c_i \neq c_j$ for any $j \in E$ and $a_i \neq a_j$ for any $j \in E$}{
	$J := J \cup \{i\}$\\
	$[a_\acov'']_i := [a_\acov']_i$
	%}
	%}
}
%Let $a^E := (b_i : b_i = c_i \text{ if } i \in E \text {, else } b_i=a_i, i \in N) \in \mathcal{C}$\\
\KwOut{$a_\acov''$}
\end{algorithm}
%\end{minipage}
%\end{wrapfigure}

%\begin{wrapfigure}[15]{t!}{0.5\textwidth}
%\begin{minipage}{1\textwidth}
%\vspace{0pt}  
\begin{algorithm}%[H] %[H] %[tb]
% for custom keywords etc., see 9.5.1 in algorithm2e.pdf
\caption{\an{UpdateFunctionGeneral}}
\label{supp:alg:updategeneral}
%\KwSty{test command block} 
\tcp{For simplicity, in what follows, let $[a_\acov]_i$ denote the marginal distribution of $\Uact_i$ under $a_\acov$}
%Let $E := \emptyset$\\
%\For{$i \in I$ (in random order)}{
%Let $\tilde{P}_{w, i}$ be the empirical distribution of $C_i$ in these $J$ times\\
%\If{$\vn{Consist}(E, \{\tilde{P}_{w, j}\}_{j \in E}, i, \tilde{P}_{w, i})$}{
%$E := E \cup \{i\}$
%}
%}

\KwIn{$a_\acov, a_\acov', \acov$}
Let $a_\acov'' := a_\acov$\\
\If{There exists an $i \in \users$ s.t. $[a_\acov]_i \neq [a_\acov']_i$}{
	pick one such $i$ at random, if there are several \\
	$[a_\acov'']_i := [a_\acov']_i$
	%}
	%}
}
%Let $a^E := (b_i : b_i = c_i \text{ if } i \in E \text {, else } b_i=a_i, i \in N) \in \mathcal{C}$\\
\KwOut{$a_\acov''$}
\end{algorithm}
%\end{minipage}
%\end{wrapfigure}

\todo{OLD:
The assistant we propose for this setting is \defi{Consistresp} described in Algorithm \ref{supp:alg:consistrep}.
The intuition is that, based on our assumptions, the assistant knows w.r.t. which belief the customers best-responded. This, together with the ``prior'' on the payoff structure (congestion game) allows to reason and select \emph{``consistent''} responses whose underlying beliefs do ``not collide''. The following proposition is proved (in generalized form) in Section \extref{sec:pr_improvement_selection}.}

%Let us know establish a convgerence gurantee. For a proof, see Proposition \extref{prop:improvement_selection}, which is a slightly extended version.

\subsection{Generalized proposition and proof}
\label{sec:pr_improvement_selection}

Let us state a proposition that generalizes Proposition \extref{prop:consistresp}.

\todo{introduce $\lnet$ and/or $\lne$ in case the defi was dropped in the main paper}

\begin{Proposition}[Convergence of Algorithm \ref{supp:alg:partpred}; sketch]
\label{prop:regret}
\label{prop:partpred_stoch}
\label{prop:partpred}
In the dynamic small-scale setting (Section \ref{sec:defdyn}), assume $X^t$ to be independent of $X^{1:t-1}$, and that in $\gat$ exists a strict BNE.
Then the following holds true:
\todo{update this based on the paper versions}
\begin{citem}

	\item \emph{General stochastic case:} Let, for all $\acov$,
	\begin{align}
	\bar{A}_\acov := \{ P_{G, s}(\Uact|\acov) : \text{ $s$ is a (deterministic) strategy profile of the \bg $\gat$} \}. \label{supp:eqn:ba}
	\end{align}
	(Note that $\bar{A}_\acov$ is finite since the range of all variables $\Usig, \Uact$ is finite.)
	Let $\bar{A} := ( \bar{A}_{\acov} )_{\acov \in \dom_\Acov}$. 
	%In the (dynamic) small-scale setting, let all users be inference-assistable and assistant-separable (with $h_i(B) = B_{-i}$ in Eq.~\extref{eqn:uta}).
	%Assume in $M$ exists a self-fulfilling prophecy policy.
	%Let $\gat$ have a strict BNE.
	Let the assistant's policy $\pi_r$ be $\vn{Partpred}(\bar{A}, r, \an{UpdateFunctionGeneral})$ as defined in Algorithm \ref{supp:alg:partpred} 
	with $\an{UpdateFunctionGeneral}$ as defined in Algorithm \ref{supp:alg:updategeneral},
	and $\bar{A}$ as defined above. %in Eq.~\ref{supp:eqn:ba}.
	
	Then, for any $\varepsilon > 0$, there exists $R, T$ such that for all $r > R, t > T$, it holds that $P( \lpredt_{\pi_r} = 0 ) > 1 - \varepsilon$ and $P( \lnet_{\pi_r} = 0 ) > 1 - \varepsilon$.

	\item \emph{Directed convergence in in complete-information congestion game case:} 
	(Note that a version of this part of the proposition can be formulated where not best, but just improving responses are assumed for the customers, which can even speed up convergence in certain cases.)
	Let $\Usig$ be fully determined by $\Acov$ and for each value of $\Acov$, let the (complete information) game $\gat$ be a congestion game \citep{roughgarden2016twenty} where all Nash equilibria are strict. %, which can be seen as our ``prior'' on customer utilities.
	%For simplicity, in this deterministic setting, let $A \in \mathcal{C}$ (i.e., an action profile instead of (dirac) distributions over action profiles).
	For simplicity, in this deterministic setting, assume $A \in \dom_{\Uact}$ (i.e., an action profile instead of (Dirac) distributions over action profiles).
	Let the assistant's policy $\pi_r$ be given by $\vn{Partpred}(\bar{A}, r, \an{UpdateFunctionCongestion})$ (Algorithm \ref{supp:alg:partpred}) 
	with $\an{UpdateFunctionCongestion}$ as defined in Algorithm \ref{supp:alg:consistupdate},
	and $\bar{A} = \dom_{\Uact}$ the set of all action profiles.
	%Furthermore, let Assumption \ref{asm:unibr} hold true.
	Then $\lpredt_{\pi_r}, \lnet_{\pi_r} \to 0$ for $t \to \infty$ without ever invoking line \ref{code:rands}, i.e., without needing ``undirected'' search.

\end{citem}

%\todo{Maybe state a ``Sketch of Proposition ... in the supplement'', only referring to sketch of algo.}
%Let Assumption \ref{asm:unibr} \todo{which assumptions exactly are needed?} hold true. And let $G$ has a BNE $s$ with unique best responses, i.e., 
%$ |\arg\max_{c_i} \E_{s}(U_i(\usig_i, c_i, C_{-i})) | = 1 $, for all $i \in N$.
%%Let $\bar{A}$ be the set of possible distributions of $C$ coming from all possible $\sigma$ (which is finite because there are only finitely many strategies).
%Let $\pi$ be $\vn{Partpred}(\bar{A}, r, \an{UpdateFunctionGeneral})$ as defined in Algorithm \extref{supp:alg:partpred}) with $\an{UpdateFunctionGeneral}$ as defined in Algorithm \extref{supp:alg:updategeneral}. The resulting algorithm is sketched in Algorithm \ref{supp:alg:partpredsk}.
%%Let $s^t$ be the strategy profile with $s^t_i(c_i|\usig_i, r, w) := P(c^t_i|\usig_i^t, w^t), i \in N$ (a random variable since $\pi$'s output is random due to it depending on past samples).
%Let $s^t$ be the strategy profile of the \bg $G$ defined by
%\begin{align}
%[s^t]_i(c_i|\usig_i, w) = P_{M, \pi}(C_i^t=c_i|\Usig_i^t=\usig_i, W^t=w), i \in N
%\end{align}
%Then, for any $\varepsilon > 0$, there exists $R, T$ such that for all $r > R, t > T$, it holds that $P( \lprobt = 0 ) > 1 - \varepsilon$ and, in particular, $P( \lnet = 0 ) > 1 - \varepsilon$.

%\todo{old: $P( \text{$A^t$ is a probabilistic fixed point} ) > 1 - \varepsilon$}, and, in particular, based on Proposition \ref{prop:argmin_ne}, $P( \text{$s^t$ is a BNE of $G$} ) > 1 - \varepsilon$.
\maybe{Maybe (in the supplement) state a version where (1) $\Acov_i$ that would be available to user $i$ depends on the time the users accesses the web page; (2) the assistant provides forecasts based on the time of request; (3) to not have to completely switch to a intra-day-dynamics setting, assume for simplicity that the request vector is known in advance (in $\Acov$). Or rather assume that the requests actually happen completely before the actions; i.e., it's not the case that $A$ is updated based on observed actions; it's only updated based on received requests. $\Acov=(\Acov_daytime=1, \Acov_daytime=2, ....)$}
%, where $G$ is the assistant-free indeced game (Section ...).
%
%
%
%If $\pi$ is Consistresp($\bar{A}$) (Algorithm \ref{supp:alg:consistrep}), then $P( \text{ exists $t'$ s.t. for all $t'' > t'$} s^{t''} \text{ is BNE of $G$} ) \to 1$ for $t \todo{J} \to \infty$, where $G$ is the assistant-free indeced game (Section ...). \todo{Or put $P(L = 0)$ here and then as a corallary this implies NE?} then the strategy profile $s$ with $s_i(c_i|\usig_i, r, w) := P^M(c_i|\usig_i, r, w), i \in N$ is a BNE of $G$.
\end{Proposition}

\begin{proof}[Proof for Proposition \ref{prop:partpred}]

\newcommand{\Ngroups}{J}

\stress{First part of the proposition: General stochastic case:}

\stress{Prerequisites.}

Let $P_{\acov, a_\acov} := P_M(\Uact|\Acov=\acov, A=a_\acov)$.
Keep in mind that, as usual, by a fixed point/self-fulfilling prophecy under $\Acov=\acov$ we mean $a_\acov$ with $P_{\acov, a_\acov} = a_\acov$.
By assumption, there exists a strict BNE in $\gat$.
Then Corollary \extref{cor:at} implies that there is $\pi$ with $0 = \lpred_\pi = \E( d(P_{M}(\Uact|\Acov, A=\pi(\Acov)), \pi(\Acov)) )$. Hence, for each $\Acov=\acov$ there exists a fixed point.

Now let $\acov$ be arbitrary but fixed.
Keep in mind that by a \defi{(same-covariate, same-prediction) group (of stages)} we mean the subsequence of $R$ stages $(t^\acov_j)_j$ where $\Acov^t=\acov$ and $A^t = a_\acov$ for some $a_\acov \in \bar{A}_\acov$.
Furthermore, let us say \defi{the algorithm converges at that and that group of stages with covariate $\acov$}, if after that group of stages it will always output the same $a_\acov$.
Let $\Ngroups := | \bar{A}_\acov |$. 
Let $d := \min_{a_\acov, a_\acov' \in \bar{A}_\acov} \| P_{\acov, a_\acov} - P_{\acov, a_\acov'} \|$.

\emph{Observe that the algorithm certainly converges in finite time -- at the latest after sampling has happened $r$ times (corresponding to one group) under all $a_\acov \in \bar{A}_\acov$, i.e., after $\Ngroups r$ stages.}
So we have to show that with growing $R$ the probability that the reason for convergence is not that it found an actual fixed point (self-fulfilling prophecy) goes to zero.
Observe that in order for it to not converge due to finding an actual fixed point either of the following two events has to happen:
\begin{citem}
	\item the algorithm converges at some action that is not a fixed point by wrongly taking it for a fixed point;
	\item it converges after the $\Ngroups$ groups of stages by the criterion to force convergence after $\Ngroups$ (lines \ref{line:forces} to \ref{line:forcee}), and has missed the actual fixed point (or one of the actual fixed points).
\end{citem}
So it suffices to show for these events individually, that with growing $R$ the probability that they happen goes to zero.
%
%Let's bound the probability that this $a_w$ won't be a fixed point.
%
%Observe that in order for $a_w$ to converge, either of two events has to happen:
%1. it is an ``empirical'' fixed point
%2. or no empirical fixed point has been found.
%
%
%
%
%Let's bound the probability of the complementary event of what we want to show.

\stress{Bound the probability that the algorithm converges at some action that is not a fixed point by wrongly taking it for a fixed point.}

Observe that during the phase where the algorithm has not converged yet, each $a_\acov \in \bar{A}_\acov$ is chosen as action $A$ during at most one group of stages and let us denote the corresponding empirical distribution of $\Uact$ by $\hat{P}_{\acov, a_\acov}$.

The phase where the algorithm has not converged yet consists of at most $\Ngroups$ groups of stages, and at most $\Ngroups-1$ groups of stages where $A$ is an $a_\acov$ that is not a fixed point. Given any $\varepsilon$, we have to show that there is $R$, such that for any $r > R$, we can bound the probability that the algorithm converges due to ``wrongly taking $a_\acov$ as a fixed point'' at the end any of these groups of stages by $\varepsilon$. We do so by bounding the probability that this happens at any individual group of the at most $\Ngroups-1$ groups where $a_\acov$ is not a fixed point, and then sum them up and apply the union bound.

Let $R \in \N$ be such that, for all $r > R$ and for all $a_\acov \in \bar{A}_\acov$ that are used during these most $\Ngroups-1$ groups where $a_\acov$ is not a fixed point: $P( \| P_{\acov, a_\acov} - \hat{P}_{\acov, a_\acov}^r \| > \frac{d}{2} ) < \frac{\varepsilon}{\Ngroups-1}$, with $\hat{P}_{\acov, a_\acov}^r$ for the respective used $a_\acov$ as defined in Algorithm \ref{supp:alg:partpred}. (Such $R$ exists based on the weak law of large numbers \cite{klenke2013probability} and the fact that $\bar{A}_\acov$ is finite.)

So let us fix one of these groups of stages where $A$ is an $a_w$ that is not a fixed point. In particular, $a_\acov \neq P_{\acov, a_\acov}$. (Keep in mind that nonetheless, $P_{\acov, a_\acov} \in \bar{A}_\acov$.)
For all $r > R$, the probability that the algorithm converges at the end of this group of stages coincides with (or rather: is bounded by) the probability that $\| a_\acov - \hat{P}_{\acov, a_\acov}^r \| \leq \| P_{\acov, a_\acov} - \hat{P}_{\acov, a_\acov}^r \|$. But 
\begin{align*}
&P\left( \| a_\acov - \hat{P}_{\acov, a_\acov}^r \| \leq \| P_{\acov, a_\acov} - \hat{P}_{\acov, a_\acov}^r \|  \right) \\
&\leq P( \| P_{\acov, a_\acov} - \hat{P}_{\acov, a_\acov}^r \| > \frac{d}{2} ) < \frac{\varepsilon}{\Ngroups-1}.
\end{align*}

To see why the inequality holds true, observe that the event $\| P_{\acov, a_\acov} - \hat{P}_{\acov, a_\acov}^r \| \leq \frac{d}{2}$ 
%together with the fact that $\frac{d}{2}$
implies the event $\| a_\acov - \hat{P}_{\acov, a_\acov}^r \| \geq \frac{d}{2} \geq \| P_{\acov, a_\acov} - \hat{P}_{\acov, a_\acov}^r \|$. (To see the first inequality, assume otherwise. Then $\| P_{\acov, a_\acov} - a_\acov \| \leq \| P_{\acov, a_\acov} - \hat{P}_{\acov, a_\acov}^r \| + \| a_\acov - \hat{P}_{\acov, a_\acov}^r \| < d$, which contradicts what we assumed.)
%But together, this impliesthein turn implies 
%Now observe that 
%But $\| a_w - \hat{P}_{w, a_w}^r \| \geq \frac{d}{2}$ in turn implies 
%
%)

So the probability that the algorithm converges at the end of any of these groups of stages (where $A$ is an $a_\acov$ that is not a fixed point) is bounded by $(\Ngroups-1) \frac{\varepsilon}{\Ngroups-1} = \varepsilon$. This is what had to be shown.

%\todo{OLD:}
%
%
%the algorithm could converge without them being a fixed point.
%Con
%So assume that during the $K := \bar{A}_w$
%So let $a_w \neq \hat{P}_{w, a_w}$.
%%Let $t_0$ be the time point of convergence.
%Note that there are at most $K-1$ ``wrong'' $a_w \in \bar{A}_w \setminus A^*_w$ that are taken as $A^t$ during at most $(K-1) ....$ stages $t$.
%
%Observe that no $a$ of these is visited twice unless it is the terminal $a$.
%
%Let us look at the probability that for any of them
%$a_w = \arg\min_{a_w'} \bar{L}_{a_w'}$.
%
%Empirical distribution $\hat{P}_{w, a_w}$.
%$P_{w, a_w} := P(C|W=w, \dc{A=a_w}) \in R^{n ...}$.
%
%
%
%$\bar{L}_{w, a_w} = | \hat{P}_{w, a_w} -  P_{w, a_w} |$.

\stress{Bound the probability of convergence of the algorithm after the $\Ngroups$ groups of stages by the criterion to force convergence after $\Ngroups$, and having missed the actual fixed point (or one of the actual fixed points).}

What we have to do here is bound the probability that a fixed point is not taken as a fixed point.
Let us be more specific.
Given any trajectory of the algorithm with some ordering of the groups of stages, let $a_\acov \in \bar{A}_\acov$ be a (the first one, if there are several) fixed point, i.e., $a_\acov = P_{\acov, a_\acov}$, which is taken as $A$ at some point during the trajectory.
Given any $\varepsilon$, we have to show that there is $R$, such that for any $r > R$, we can bound the probability that $a_\acov$ is ``not recognized as a fixed point'' by $\varepsilon$.

Let $R \in \N$ be such that $P( \| P_{\acov, a_\acov} - \hat{P}_{\acov, a_\acov}^r \| > \frac{d}{2} ) < \varepsilon$ for all $r > R$ and for all $a_\acov \in \bar{A}_\acov$. (Such $R$ exists based on the weak law of large numbers \cite{klenke2013probability} and the fact that $\bar{A}_\acov$ is finite.)
Then for all $r > R$, the probability that it is not recognized as a fixed point is
\begin{align*}
&P\left( \| a_\acov' - \hat{P}_{\acov, a_\acov}^r \| < \| a_\acov - \hat{P}_{\acov, a_\acov}^r \| \text{ for some } a_\acov'  \right) \\
&=P\left( \| a_\acov' - \hat{P}_{\acov, a_\acov}^r \| < \| P_{\acov, a_\acov} - \hat{P}_{\acov, a_\acov}^r \| \text{ for some } a_\acov'  \right) \\
&\leq P( \| P_{\acov, a_\acov} - \hat{P}_{\acov, a_\acov}^r \| > \frac{d}{2} ) < \varepsilon.
\end{align*}
(Since $d$ is the minimum distance between $a_\acov'$ and $P_{\acov, a_\acov}$ -- the analogous argument as before.)

%
%
%We consider the case that $a_w$ is chosen as $A$ during some group of stages.
%$a_w$, i.e., $a_w = P_{w, a_w}$, , i.e., $\bar{L}_{a_w}$
%
%
%
%
%\stress{2. bound the probability that it does not stop at an actual fixed point.}

\stress{Finally.}

Now simply take $R, T$ large enough such that:
\begin{citem}
	\item With high probability, each $\Acov=\acov$ (with positive probability) has been observed at least $\Ngroups R$ times.
	\item Within the event that each $\Acov=\acov$ (with positive probability) has been observed at least $\Ngroups R$ times: for $r > R$, under algorithm $\pi_r$, the probability that converges against a fixed point occurred under all $\Acov=\acov$ (which is a product of $|\dom_\Acov|$ probabilities that each go to 1 with growing $r$, based on the above) is high enough.
\end{citem}

%Now for each $r \in \N$ there exists $T \in \N$ such that with high probability, each $W=w$ (with positive probability) has been observed at least $K r$ times.
%
%Above we showed that, for each $W=w$, during these $K r$ stages, converges against a fixed point with high probability. 
%
%
%Now simply take $R$ large enough such that for $r > R$ (and the corresponding $T$), the probability that converges against a fixed point occurred under all $W=w$ (which is a product of $|\dom_W|$ probabilities that each go to 1 with growing $r$) is high enough. 

\todo{this last ``Finally'' seems correct but maybe read through it again.}

%Take the maximum of the two $R$ defined above.

%\todo{So we showed that for each $W=w$, the algorithm converges to a fixed point with probability going to one. ...}

\stress{Second part of the proposition: Directed convergence in complete-information congestion game case:} 

We write down the proof for the case of a fixed $\Acov$. The general case works analogously.

Let $\Phi$ denote the \emph{potential function} (the bigger the utilities, the bigger the potential function) \citep{roughgarden2016twenty} of the congestion game (and thus potential game) $\gat$.

Let stage $t$, announcement $A^t=a$ and outcome $\Uact^t=\uact$ be arbitrary but fixed. 
%To simplify the formulation of our result, let us first make the following definitions:
In what follows, we say a set $E \subset \users$ of players is \emph{collision-free} if
(1) $\uact_i \neq \uact_j$ for any $i,j \in E$ (no two players in $E$ move to the same ``target'' slot), and
(2) $a_i \neq a_j$ for any $i,j \in E$ (no two players in $E$ move from the same ``source'' slot).
Let us denote
\begin{align*}
a^E := (b_i : b_i = \uact_i \text{ if } i \in E \text {, else } b_i=a_i, i \in \users) \in \mathcal{\Uact} ,
\end{align*}
i.e., applying all moves of players in $E$ to $a$.

\stress{Claim:} If $E \subset F \subset \users$ are collision-free, then $\Phi(a^F) {\geq} \Phi(a^E)$. So, roughly speaking, setting $A^{t+1}:=a^E$ for any collision-free $E$, such that no superset $F \supseteq E$ is collision-free, is a reasonable policy for the assistant.

To see why this holds, let $\uact^1=a, \uact^2, \ldots, \uact^k=a^E$ be a \defi{path} from $a$ to $a^E$, meaning that at each step $j$ from $\uact^j$ to $\uact^{j+1}$, only one player $i_j \in E$ applies her move $[\uact - a]_{i_j}$ to $\uact^{j}$.
%Let $i_j$ be such that in between $c_{j+1}$ and $c_{j}$, only $i_j$ moves.

For the potential function $\Phi$ \citep{roughgarden2016twenty} we have
\begin{align}
&\Phi(\uact^k) - \Phi(\uact^1) \\
&= \sum_j \Phi(\uact^{j+1}) - \Phi(\uact^{j}) \\
&= \sum_j u_{i_j}(\uact^{j+1}) - u_{i_j}(\uact^{j}) .
\end{align}

Hence, it suffices to show that $u_{i_j}(\uact^{j+1}) - u_{i_j}(\uact^{j}) \geq 0$ for all $j$, because then we cannot do better than $a^E$ by taking $a^F$ for any subset $F \subset E$. % (note that $E$ and $F$ swapped their role here compared to the statement of Proposition \ref{prop:improvement_selection}).
To prove this, we establish that for all $j$,
\begin{align}
u_{i_j}(\uact^{j+1}) - u_{i_j}(\uact^{j}) \geq u_{i_j}(\uact_{i_j}^{j+1}, \uact_{-i_j}^1) - u_{i_j}(\uact^1) \geq 0. \label{supp:eqn:ineq}
\end{align}
The second inequality directly follows from our assumption that player $i_j$ makes an improvement move.
To prove the first inequality, we show that for all $j$,
\begin{align}
u_{i_j}(\uact^{j+1}) &\geq u_{i_j}(\uact_{i_j}^{j+1}, \uact_{-i_j}^1), \label{supp:eqn:up1}\\
u_{i_j}(\uact^{j}) &\leq u_{i_j}(\uact^1) .\label{supp:eqn:up2}
\end{align}
%This suffices, because, obviously, for the r.h.s.\ of Inequality \ref{supp:eqn:ineq}, we have $u_{i_j}(c_{i_j}^{j+1}, c_{-i_j}^1) - u_{i_j}(c^1) \geq 0$ based on 

Keep in mind that in the congestion game, the utility only depends on the number of other players at the same slot. 

For each $j$, based on the assumption that no two players move to the same slot, either the number of other players $l \neq i_j$ at slot $\uact_{i_j}^{j+1}$ in action profile $\uact^{j+1}$ is the same or it drops compared to $(\uact_{i_j}^{j+1}, \uact_{-i_j}^1)$, which implies Inequality \ref{supp:eqn:up1}.

Furthermore, for each $j$, based on the assumption that no two players move from the same slot, the number of other players $l \neq i_j$ at slot $\uact_{i_j}^{j}$ in action profile $\uact^{j}$ is the same or it increases compared to $c^1$, which implies Inequality \ref{supp:eqn:up2}. 

This is also the reason why we cannot allow two players to move from the same slot: because it could happen, that the change of circumstances due to the second one moving renders the move of the first one a worsening move.

%\stress{Second bullet point:}

\stress{Claim:} $\lpredt_{\pi_r}, \lnet_{\pi_r} \to 0$ for $t \to \infty$ without ever invoking line \ref{code:rands}, i.e., without needing ``undirected'' search.

This is the analogous argument of convergence of classical best-response dynamics in congestion games against a NE  \citep{roughgarden2016twenty}: also in our case the $\Phi$ is guaranteed to strictly increase (with some constant lower bound on each decrease since the game is finite) until it reaches a ``local'' minimum, since we always let at least one customer improve. Therefore we will never reach the same action profile again, i.e., never invoke lines \ref{code:rands} to \ref{code:rande}. 

And due to the assumed strictness of the NE, we will stay at a NE once it was announced. Then apply Corollary \extref{cor:at}.
%And for congestion games, there always is a pure NE \citep{roughgarden2016twenty}.

\end{proof}

\stress{Comment.} Note that we assume $\bar{A}$ to be given. This is to be more modular and better express the algorithm which is on the proof-of-concept level. In principle $\bar{A}$ can be inferred from data as well.

\clearpage

\bibliographystyle{abbrvnat}
\bibliography{incl/cafco,incl/phil_master}

\begin{thebibliography}{42}
\providecommand{\natexlab}[1]{#1}
\providecommand{\url}[1]{\texttt{#1}}
\expandafter\ifx\csname urlstyle\endcsname\relax
  \providecommand{\doi}[1]{doi: #1}\else
  \providecommand{\doi}{doi: \begingroup \urlstyle{rm}\Url}\fi

\bibitem[{ASFA}(2019)]{franceforecasts}
{ASFA}.
\newblock {France motorways: traffic forecast on french motorways - ASFA}.
\newblock \url{https://www.autoroutes.fr/en/traffic-forecast.htm}, 2019.
\newblock Accessed: 2019-01-25.

\bibitem[Balcan et~al.(2016)Balcan, Sandholm, and Vitercik]{balcan2016sample}
M.-F. Balcan, T.~Sandholm, and E.~Vitercik.
\newblock {Sample complexity of automated mechanism design}.
\newblock In \emph{{Advances in Neural Information Processing Systems}}, pages
  2083--2091, 2016.

\bibitem[Bergemann and Morris(2017)]{bergemann2017information}
D.~Bergemann and S.~Morris.
\newblock {Information design: A unified perspective}.
\newblock 2017.

\bibitem[Bilancini and Boncinelli(2016)]{bilancini2016strict}
E.~Bilancini and L.~Boncinelli.
\newblock {Strict Nash equilibria in non-atomic games with strict single
  crossing in players (or types) and actions}.
\newblock \emph{Economic Theory Bulletin}, 4\penalty0 (1):\penalty0 95--109,
  2016.

\bibitem[Brezis(2010)]{brezis2010functional}
H.~Brezis.
\newblock \emph{Functional analysis, Sobolev spaces and partial differential
  equations}.
\newblock Springer Science \& Business Media, 2010.

\bibitem[{DB}(2019)]{dbtravel}
{DB}.
\newblock {DB BAHN -- TravelService -- Query page}.
\newblock \url{https://reiseauskunft.bahn.de/bin/query.exe/en}, 2019.
\newblock Specifically, we refer to the ``Demand'' symbol that occurs for some
  train connections (which may be informed by bookings, besides past behavioral
  data). Accessed: 2019-06-04.

\bibitem[Duetting et~al.(2019)Duetting, Feng, Narasimhan, Parkes, and
  Ravindranath]{dutting2019optimal}
P.~Duetting, Z.~Feng, H.~Narasimhan, D.~C. Parkes, and S.~S. Ravindranath.
\newblock Optimal auctions through deep learning.
\newblock In \emph{{36th International Conference on Machine Learning (ICML)}},
  2019.

\bibitem[{Google}(2019)]{googlepopulartimes}
{Google}.
\newblock {Google Maps}.
\newblock \url{https://www.google.com/maps}, 2019.
\newblock Specifically, we refer to the ``Popular times'' bar diagram that is
  often displayed for facilities like swimming pools, train stations, etc.
  Accessed: 2019-06-04.

\bibitem[Harsanyi(1967)]{harsanyi1967games}
J.~C. Harsanyi.
\newblock {Games with incomplete information played by “Bayesian” players,
  I--III Part I. The basic model}.
\newblock \emph{Management science}, 14\penalty0 (3):\penalty0 159--182, 1967.

\bibitem[Harsanyi(1973)]{harsanyi1973games}
J.~C. Harsanyi.
\newblock {Games with randomly disturbed payoffs: A new rationale for
  mixed-strategy equilibrium points}.
\newblock \emph{International journal of game theory}, 2\penalty0 (1):\penalty0
  1--23, 1973.

\bibitem[Hellman and Levy(2017)]{hellman2017bayesian}
Z.~Hellman and Y.~J. Levy.
\newblock Bayesian games with a continuum of states.
\newblock \emph{Theoretical Economics}, 12\penalty0 (3):\penalty0 1089--1120,
  2017.

\bibitem[{Häusler} et~al.(2014){Häusler}, {Ordóñnez-Hurtado}, {Griggs},
  {Radusch}, and {Shorten}]{haeusler2014closed}
F.~{Häusler}, R.~H. {Ordóñnez-Hurtado}, W.~M. {Griggs}, I.~{Radusch}, and
  R.~N. {Shorten}.
\newblock Closed-loop flow regulation with balanced routing.
\newblock In \emph{2014 International Conference on Connected Vehicles and Expo
  (ICCVE)}, pages 1054--1055, 2014.

\bibitem[Hyndman et~al.(2008)Hyndman, Koehler, Ord, and
  Snyder]{hyndman2008forecasting}
R.~Hyndman, A.~B. Koehler, J.~K. Ord, and R.~D. Snyder.
\newblock \emph{{Forecasting with exponential smoothing: the state space
  approach}}.
\newblock Springer Science \& Business Media, 2008.

\bibitem[Kearns et~al.(2014)Kearns, Pai, Roth, and Ullman]{kearns2014mechanism}
M.~Kearns, M.~Pai, A.~Roth, and J.~Ullman.
\newblock {Mechanism design in large games: Incentives and privacy}.
\newblock In \emph{{Proceedings of the 5th conference on Innovations in
  theoretical computer science}}, pages 403--410. ACM, 2014.

\bibitem[Kim and Yannelis(1997)]{kim1997existence}
T.~Kim and N.~C. Yannelis.
\newblock {Existence of equilibrium in Bayesian games with infinitely many
  players}.
\newblock \emph{Journal of economic theory}, 77\penalty0 (2):\penalty0
  330--353, 1997.

\bibitem[Klenke(2013)]{klenke2013probability}
A.~Klenke.
\newblock \emph{Probability theory: a comprehensive course}.
\newblock Springer Science \& Business Media, 2013.

\bibitem[Ling et~al.(2018)Ling, Fang, and Kolter]{ling2018game}
C.~K. Ling, F.~Fang, and J.~Z. Kolter.
\newblock {What game are we playing? End-to-end learning in normal and
  extensive form games}.
\newblock \emph{arXiv preprint arXiv:1805.02777}, 2018.

\bibitem[L\"utkepohl(2006)]{Luetkepohl2006}
H.~L\"utkepohl.
\newblock \emph{New Introduction to Multiple Time Series Analysis}.
\newblock Springer, Berlin, Heidelberg, New York, oxford statistical science
  series edition, 2006.

\bibitem[Mare{\v{c}}ek et~al.(2015)Mare{\v{c}}ek, Shorten, and
  Yu]{marevcek2015signalling}
J.~Mare{\v{c}}ek, R.~Shorten, and J.~Y. Yu.
\newblock {Signalling and obfuscation for congestion control}.
\newblock \emph{International Journal of Control}, 88\penalty0 (10):\penalty0
  2086--2096, 2015.

\bibitem[Mare{\v{c}}ek et~al.(2016)Mare{\v{c}}ek, Shorten, and
  Yu]{marevcek2016r}
J.~Mare{\v{c}}ek, R.~Shorten, and J.~Y. Yu.
\newblock {r-extreme signalling for congestion control}.
\newblock \emph{International Journal of Control}, 89\penalty0 (10):\penalty0
  1972--1984, 2016.

\bibitem[Nisan et~al.(2007)Nisan, Roughgarden, Tardos, and
  Vazirani]{nisan2007algorithmic}
N.~Nisan, T.~Roughgarden, E.~Tardos, and V.~V. Vazirani.
\newblock \emph{{Algorithmic game theory}}, volume~1.
\newblock Cambridge University Press Cambridge, 2007.

\bibitem[Nisan et~al.(2011)Nisan, Schapira, Valiant, and Zohar]{nisan2011best}
N.~Nisan, M.~Schapira, G.~Valiant, and A.~Zohar.
\newblock {Best-Response Mechanisms}.
\newblock In \emph{ICS}, pages 155--165, 2011.

\bibitem[Osborne and Rubinstein(1994)]{osborne1994course}
M.~J. Osborne and A.~Rubinstein.
\newblock \emph{{A course in game theory}}.
\newblock MIT press, 1994.

\bibitem[Pacuit and Roy(2017)]{sep-epistemic-game}
E.~Pacuit and O.~Roy.
\newblock {Epistemic Foundations of Game Theory}.
\newblock In E.~N. Zalta, editor, \emph{{The Stanford Encyclopedia of
  Philosophy}}. Metaphysics Research Lab, Stanford University, summer 2017
  edition, 2017.

\bibitem[Pearl(2000)]{Pearl2000}
J.~Pearl.
\newblock \emph{Causality}.
\newblock Cambridge University Press, 2000.

\bibitem[Proakis and Manolakis(1996)]{proakis1996}
J.~G. Proakis and D.~G. Manolakis.
\newblock \emph{{Digital Signal Processing (3rd Ed.): Principles, Algorithms,
  and Applications}}.
\newblock Prentice-Hall, Inc., 1996.
\newblock ISBN 0-13-373762-4.

\bibitem[Rath(1992)]{rath1992direct}
K.~P. Rath.
\newblock {A direct proof of the existence of pure strategy equilibria in games
  with a continuum of players}.
\newblock \emph{Economic Theory}, 2\penalty0 (3):\penalty0 427--433, 1992.

\bibitem[Reed and Simon(1972)]{reed1972methods}
M.~Reed and B.~Simon.
\newblock \emph{{Methods of Modern Mathematical Physics: Functional
  Analysis.-1972.-(RU-idnr: M103448034)}}.
\newblock Academic Press, 1972.

\bibitem[Roughgarden(2005)]{roughgarden2005selfish}
T.~Roughgarden.
\newblock \emph{{Selfish routing and the price of anarchy}}, volume 174.
\newblock MIT press Cambridge, 2005.

\bibitem[Roughgarden(2016)]{roughgarden2016twenty}
T.~Roughgarden.
\newblock \emph{{Twenty lectures on algorithmic game theory}}.
\newblock Cambridge University Press, 2016.

\bibitem[Sabourian(1990)]{sabourian1990anonymous}
H.~Sabourian.
\newblock {Anonymous repeated games with a large number of players and random
  outcomes}.
\newblock \emph{Journal of Economic Theory}, 51\penalty0 (1):\penalty0 92--110,
  1990.

\bibitem[Schlote et~al.(2014)Schlote, Chen, and Shorten]{schlote2014closed}
A.~Schlote, B.~Chen, and R.~Shorten.
\newblock On closed-loop bicycle availability prediction.
\newblock \emph{IEEE Transactions on Intelligent Transportation Systems},
  16\penalty0 (3):\penalty0 1449--1455, 2014.

\bibitem[Schmeidler(1973)]{schmeidler1973equilibrium}
D.~Schmeidler.
\newblock {Equilibrium points of nonatomic games}.
\newblock \emph{Journal of statistical Physics}, 7\penalty0 (4):\penalty0
  295--300, 1973.

\bibitem[Shoham and Leyton-Brown(2008)]{shoham2008multiagent}
Y.~Shoham and K.~Leyton-Brown.
\newblock \emph{{Multiagent systems: Algorithmic, game-theoretic, and logical
  foundations}}.
\newblock Cambridge University Press, 2008.

\bibitem[Simon(1954)]{simon1954bandwagon}
H.~A. Simon.
\newblock {Bandwagon and underdog effects and the possibility of election
  predictions}.
\newblock \emph{Public Opinion Quarterly}, 18\penalty0 (3):\penalty0 245--253,
  1954.

\bibitem[Smyrnakis and Leslie(2010)]{smyrnakis2010dynamic}
M.~Smyrnakis and D.~S. Leslie.
\newblock {Dynamic opponent modelling in fictitious play}.
\newblock \emph{The Computer Journal}, 53\penalty0 (9):\penalty0 1344--1359,
  2010.

\bibitem[Spohn(1982)]{spohn1982how}
W.~Spohn.
\newblock {"How to make Sense of Game Theory"}.
\newblock In W.~Stegm{\"u}ller, W.~Balzer, and W.~Spohn, editors,
  \emph{Philosophy of Economics}, pages 239--270. Springer Berlin Heidelberg,
  1982.

\bibitem[Sutton and Barto(1998)]{sutton1998reinforcement}
R.~S. Sutton and A.~G. Barto.
\newblock \emph{Reinforcement learning: An introduction}.
\newblock MIT press, 1998.

\bibitem[Taneva(2015)]{taneva2015information}
I.~A. Taneva.
\newblock {Information design}.
\newblock 2015.

\bibitem[Tang(2017)]{tang2017reinforcement}
P.~Tang.
\newblock {Reinforcement mechanism design}.
\newblock In \emph{{IJCAI}}, volume~17, pages 26--30, 2017.

\bibitem[Wirth et~al.(2019)Wirth, St{\"u}dli, Yu, Corless, and
  Shorten]{wirth2019nonhomogeneous}
F.~R. Wirth, S.~St{\"u}dli, J.~Y. Yu, M.~Corless, and R.~Shorten.
\newblock Nonhomogeneous place-dependent markov chains, unsynchronised aimd,
  and optimisation.
\newblock \emph{Journal of the ACM (JACM)}, 66\penalty0 (4):\penalty0 24, 2019.

\bibitem[Zhang et~al.(2013)Zhang, Nguyen, and Zhang]{zhang2013wait}
Y.~Zhang, L.~T. Nguyen, and J.~Zhang.
\newblock {Wait time prediction: how to avoid waiting in lines?}
\newblock In \emph{{Proceedings of the 2013 ACM conference on Pervasive and
  ubiquitous computing adjunct publication}}, pages 481--490. ACM, 2013.

\end{thebibliography}

\end{document}